\documentclass[11pt,superscriptaddress,onecolumn,tightenlines,notitlepage,prx]{revtex4-2}

\usepackage[dvips]{graphicx}
\usepackage{amsmath,amssymb,amsthm,mathrsfs,amsfonts}
\usepackage{bm}
\usepackage{color}
\usepackage{enumerate}
\usepackage{physics,complexity}
\usepackage{soul,xcolor}
\usepackage{textcomp}
\usepackage{verbatim}
\usepackage{makecell}
\usepackage{multirow}
\usepackage{algpseudocode}
\usepackage{algorithm}
\usepackage{newfloat}
\usepackage{adjustbox}
\usepackage[caption=false]{subfig}
\usepackage{tikz}
\newsavebox{\measurebox}
\usetikzlibrary{quantikz}
\usepackage[colorlinks = true]{hyperref}
\usepackage{longtable}
\usepackage{xcolor}
\definecolor{darkred}  {rgb}{0.5,0,0}
\definecolor{darkblue} {rgb}{0,0,0.5}
\definecolor{darkgreen}{rgb}{0,0.5,0}

\hypersetup{
  urlcolor   = blue,         
  linkcolor  = darkblue,     
  citecolor  = darkgreen,    
  filecolor   = darkred       
}

\newtheorem{lemma}{Lemma}

\newtheorem{theorem}{Theorem}

\newcommand{\mc}{\mathcal}
\renewcommand{\E}{\mathop{\mathbb E\/}}

\DeclareMathOperator{\Var}{Var}

\newenvironment{nalign}{
    \begin{equation}
    \begin{aligned}
}{
    \end{aligned}
    \end{equation}
    \ignorespacesafterend{}
}
\newtheorem*{theorem*}{Theorem}
\usepackage{wrapfig}
\usepackage{capt-of}

\begin{document}
\title{Benchmarking near-term quantum computers via random circuit sampling}

\author{Yunchao Liu}
\email{yunchaoliu@berkeley.edu}
\affiliation{Department of Electrical Engineering and Computer Sciences, University of California, Berkeley, California 94720, USA}

\author{Matthew Otten}
\email{mjotten@hrl.com}
\affiliation{HRL Laboratories, LLC, 3011 Malibu Canyon Rd., Malibu, California 90265, USA}

\author{Roozbeh Bassirianjahromi}
\email{roozbeh@uchicago.edu}
\affiliation{Department of Computer Science, University of Chicago, Chicago, IL 60637, USA}

\author{Liang Jiang}
\email{liang.jiang@uchicago.edu}
\affiliation{Pritzker School of Molecular Engineering, University of Chicago, Chicago, IL 60637, USA}

\author{Bill Fefferman}
\email{wjf@uchicago.edu}
\affiliation{Department of Computer Science, University of Chicago, Chicago, IL 60637, USA}

\begin{abstract}
The increasing scale of near-term quantum hardware motivates the need for efficient noise characterization methods, since qubit and gate level techniques cannot capture crosstalk and correlated noise in many qubit systems. While scalable approaches, such as cycle benchmarking, are known for special classes of quantum circuits, the characterization of noise in general circuits with non-Clifford gates has been an unreachable task. We develop an algorithm that can sample-efficiently estimate the total amount of noise induced by a layer of arbitrary non-Clifford gates, including all crosstalks, and experimentally demonstrate the method on IBM Quantum hardware. Our algorithm is inspired by Google's quantum supremacy experiment and is based on random circuit sampling. In their paper, Google observed that their experimental linear cross entropy was consistent with a simple uncorrelated noise model, and claimed this coincidence indicated that the noise in their device was uncorrelated -- a key step in hardware development towards fault tolerance. As an application, we show that our result provides formal evidence to support such a conclusion.

\end{abstract}

\maketitle

With the recent exciting progress in NISQ (Noisy Intermediate Scale Quantum) experiments, the characterization of noise in quantum devices has become a central challenge in the field~\cite{Preskill2018quantumcomputingin,Eisert2020quantum}. This goes beyond the benchmarking of individual gates; the characterization of many-body quantum noise can be expected to play an important role in building scalable quantum computers~\cite{Erhard2019characterizing,Harper2020efficient}. This is because with scaling, understanding crosstalk between gates assumes increasing significance, in both improving fidelity of near-term experiments as well as making progress towards fault tolerance~\cite{Sarovar2020detectingcrosstalk,hashim2021randomized,Winick2021simulating,iyer2021efficient}.

Randomized benchmarking~\cite{Emerson2005scalable,Knill2008randomized,Dankert2009exact,Magesan2011scalable,Magesan2012characterizing} has been the standard approach in traditional noise benchmarking, but in the context of characterizing quantum computers it does not scale beyond 2-3 qubits due to large circuit depth~\cite{McKay2019threequbit}, and is therefore most useful for the benchmarking of individual gates. Cycle benchmarking and its variants~\cite{Erhard2019characterizing,Flammia2020efficient,Harper2020efficient,harper2020fast,flammia2021pauli,flammia2021averaged} provide a scalable approach to benchmark Clifford circuits. However, the characterization of noise in general circuits with non-Clifford gates has been an unreachable task due to the lack of group structure.

Here we address this challenge by drawing inspiration from Google's ``quantum supremacy" experiment~\cite{arute2019quantum}. Google observed that their experimental estimate of the linear cross entropy benchmark -- a proxy for the global fidelity of random circuits -- was consistent with an uncorrelated noise model defined by multiplying individual gate fidelities, and they claimed that these experimental results could be considered a way of verifying that the noise channel acting on a layer of gates is uncorrelated across each gate. This observation is remarkable in the sense that the fidelity of highly complex random circuits could be predicted by such a simple noise model, but also intriguing as little theoretical evidence has been shown that supports this observation~\cite{Martinis2019talk,Aaronson2019panel}. A natural question is how convincing is this observation from the theoretical perspective, and more importantly, could this new observation be the germ of a new way to benchmark noise in general quantum circuits, even without assuming locality and independence in the noise model? 

In this paper we develop a noise benchmarking algorithm based on random circuit sampling (RCS). We show that the total amount of noise in a global and arbitrarily correlated noise model can be sample-efficiently extracted by measuring the linear cross entropy. While the noise parameter to be estimated is the same as cycle benchmarking~\cite{Erhard2019characterizing}, RCS benchmarking works for arbitrary non-Clifford gates and therefore can be used to benchmark general circuits.
As an application, our results imply that if the noise in Google's device were highly non-local and correlated, this would cause the linear cross entropy and the uncorrelated noise model to deviate from each other in Google's experiment. This provides some formal evidence that supports Google's claim that the coincidence they observed between the two metrics indicated that the noise in their device was uncorrelated~\cite{arute2019quantum}, which raises the potential for achieving fault tolerance.

\vspace{1.5mm}
\noindent\textbf{RCS benchmarking.} The noise benchmarking problem can be formulated as follows. Consider the Pauli noise channel induced by a layer of arbitrary non-Clifford gates $\mc N(\rho)=\sum_{\alpha\in\{0,1,2,3\}^n}p_\alpha\sigma_\alpha\rho\sigma_\alpha$ where $\sigma_\alpha$ are $n$-qubit Pauli operators and $p_\alpha$ are the corresponding Pauli error rates. The assumption that noise can be described by a Pauli channel is without loss of generality, as we show below that general noise channels will be effectively twirled into a Pauli channel by RCS. The main challenge then is designing an efficient experimental procedure to estimate the total error $\lambda=\sum_{\alpha\neq 0^n}p_\alpha$. In RCS benchmarking (Algorithm~\ref{alg:rcsbenchmarkingapp}), the gates to be characterized are applied in an alternating architecture and interleaved by Haar random single-qubit gates (see Fig.~\ref{fig:rcsbenchmarking}), followed by a measurement in the computational basis. The algorithm works by estimating the average fidelity of these random circuits at several different depths, and then fit the average fidelity as an exponential decay function of depth, which gives the noise parameter $\lambda$. This algorithm can be implemented on today's hardware and is robust to state preparation and measurement (SPAM) errors.

\begin{figure}[t]
\begin{minipage}{0.44\textwidth}
  \tikzset{operator/.append style={fill=blue!50,minimum size=0.5cm}}
  \tikzset{operator/.append style={}}
    \centering
    \begin{quantikz}[column sep=0.4cm,row sep=0cm]
 \lstick{$\ket{0}$} & \gate[]{} & \gate[2,style={fill=white}]{} & \gate[]{} & \qw & \gate[]{} & \gate[2,style={fill=white}]{} &\gate[]{} & \qw & \meter{} \\
 \lstick{$\ket{0}$} & \gate[]{} & 			& \gate[]{} & \gate[2,style={fill=white}]{} & \gate[]{} &			& \gate[]{} & \gate[2,style={fill=white}]{}  & \meter{} \\
 \lstick{$\ket{0}$} & \gate[]{} & \gate[2,style={fill=white}]{} & \gate[]{} & 			& \gate[]{} & \gate[2,style={fill=white}]{} & \gate[]{} &			 & \meter{} \\
 \lstick{$\ket{0}$} & \gate[]{} & 	 		& \gate[]{} & \gate[2,style={fill=white}]{} & \gate[]{} &			& \gate[]{} & \gate[2,style={fill=white}]{}  & \meter{} \\
 \lstick{$\ket{0}$} & \gate[]{} & \qw 		& \gate[]{} & 			& \gate[]{} & \qw	    & \gate[]{} & 			 & \meter{}
 \end{quantikz}
  \caption{RCS benchmarking is an efficient algorithm to estimate the total amount of noise, including all crosstalks, on a layer of arbitrary two-qubit gates (white boxes) by implementing an alternating architecture interleaved with Haar random single qubit gates (blue boxes).}
  \label{fig:rcsbenchmarking}
\end{minipage}
\hfill 
\begin{minipage}{0.53\textwidth}
\centering
  \begin{algorithm}[H]
    \caption{\raggedright RCS benchmarking (simplified)}\label{alg:rcsbenchmarking}
    \raggedright\textbf{Input:} number of qubits $n$, maximum circuit depth $D$, number of circuits $L$\\
    \textbf{Output:} effective noise rate (ENR)
    \begin{algorithmic}[1]
    \For {$d=1\dots D$}
        \For {$i=1\dots L$}
        \State sample a random circuit $C_i\sim\mathrm{RQC}(n,d)$
        \State estimate the fidelity of $C_i$, denote as $\hat{F}_{d,i}$ \State\Comment{fidelity estimation}
        \EndFor
        \State $\hat{F}_d:=\frac{1}{L}\sum_{i=1}^L \hat{F}_{d,i}$
    \EndFor
    \State fit exponential decay $F=Ae^{-\lambda d}$ using data $\{\hat{F}_d\}_{d=1}^D$
    \State\textbf{Return} $\lambda$
    \end{algorithmic}
    \end{algorithm}
\end{minipage}
\end{figure}

\vspace{1.5mm}
\noindent\textbf{Result 1: exponential decay of average fidelity.} Let $\mathrm{RQC}(n,d)$ denote the ensemble of random quantum circuits with $n$ qubits and depth $d$ as shown in Fig.~\ref{fig:rcsbenchmarking}. An ideal implementation of a random circuit $C\sim\mathrm{RQC}(n,d)$ creates a pure state $\ket{\psi}=C\ket{0^n}$, while due to noise the experimental implementation corresponds to a mixed state $\rho$, and the fidelity of $C$ is $F=\expval{\rho}{\psi}=\bra{0^n}C^\dag \rho C\ket{0^n}$. The average fidelity is given by $\E_{C\sim\mathrm{RQC}(n,d)}F$ which we denote by $\E F$. Consider an arbitrary $n$-qubit noise channel acting on each layer of gates, which can be described as $\mc N(\rho)=\sum_{\alpha,\beta\in\{0,1,2,3\}^n}\chi_{\alpha\beta}\sigma_{\alpha}\rho \sigma_{\beta}$, where $(\chi_{\alpha\beta})$ is a positive semi-definite matrix known as the process matrix. We show that only the Pauli-diagonal component of the noise channel $\mc N^{\mathrm{diag}}(\rho)=\sum_{\alpha\in\{0,1,2,3\}^n}\chi_{\alpha\alpha}\sigma_{\alpha}\rho \sigma_{\alpha}$ is effective due to the twirling effect of random circuits (Appendix~\ref{sec:reducetopaulinoise}), in the sense that the average fidelity with noise channel $\mc N$ equals the average fidelity with $\mc N^{\mathrm{diag}}$, so without loss of generality we assume a Pauli noise channel $\mc N(\rho)=\sum_{\alpha\in\{0,1,2,3\}^n}p_\alpha\sigma_\alpha\rho\sigma_\alpha$, and the goal is to estimate the effective noise rate (ENR) $\lambda=\sum_{\alpha\neq 0^n}p_\alpha$. Our main result shows that $\E F\approx e^{-\lambda d}$, which implies that $\lambda$ can be extracted via estimating $\E F$.

\begin{theorem*}[Main result]\label{thm:fidelitydecay}
For random quantum circuits in 1D and 3-local noise channel with effective noise rate $\lambda$, the average fidelity is given by $
    e^{-\lambda d}\leq \E F\leq e^{-\lambda d}(1+K\lambda)$
up to a first-order approximation in $\lambda$. Here $K$ is a universal constant, and we assume $d\ll 2^n$.
\end{theorem*}

For simplicity here we focus on random circuits with Haar random 2-qubit gates. The same result is expected to hold for circuits with fixed 2-qubit gates interleaved by Haar random single-qubit gates as in Fig.~\ref{fig:rcsbenchmarking} as they share similar properties such as convergence to unitary $t$-designs~\cite{Harrow2009random,Brandao2016,harrow2018approximate,haferkamp2020quantum,haferkamp2020improved}, which we also confirm via numerical simulations. The starting point of the proof (Appendix~\ref{sec:fidelitydecay}) is to decompose $\E F$ via the law of total expectation by conditioning on the number of errors that happen in the circuit, $\E F=\E F_0+\E F_1+\sum_{k\geq 2}\E F_k$, where $\E F_k := \Pr[k\text{ errors happen}]\cdot\E[F|k\text{ errors happen}]$. Note that $\E F_0=\Pr[\text{no error happens}]\approx e^{-\lambda d}$, so the goal is to prove that all $\E F_k$ with $k\geq 1$ are small compared with $\E F_0$. Intuitively $\E F_k$ has a $\lambda^k$ factor and should decrease quickly with $k$ when $\lambda$ is small, so we make a first-order approximation by ignoring the $k\geq 2$ terms and focusing on the first-order contribution $\E F_1$. It is easy to prove that this approximation is valid when circuit depth $d\leq c/\lambda$ for some small constant $c$, and our numerical simulations suggest that it remains valid with experimentally relevant noise rates and circuit depths. Next, we show that $\E F_1=O(\lambda\cdot e^{-\lambda d})$ by mapping this quantity to the partition function of a classical spin model and then use a domain wall counting argument to analytically bound the partition function. This technique has been used to analyze properties of random quantum circuits~\cite{Nahum2018operator,Zhou2019emergent,hunterjones2019unitary,bao2020theory,barak2020spoofing,dalzell2020random} and here we generalize this to the context of noise benchmarking. While our rigorous arguments for showing $\E F_1=O(\lambda\cdot e^{-\lambda d})$ apply to arbitrary 3-local errors which capture most error sources in quantum devices, numerical simulations suggest that it holds for arbitrary errors. Combining the above arguments gives the exponential decay $\E F\approx e^{-\lambda d}$ when $\lambda$ is upper bounded by a small constant, which we also directly verify by simulating the average fidelity with different correlated noise models and gate sets. Our numerical simulation results are presented in Fig.~\ref{fig:xeb_mcwf_simulations} and Appendix~\ref{sec:rcstheory} and~\ref{app:numerical}.

\begin{figure}[t]
   \centering
   \subfloat[]{
    \centering
    \includegraphics[width=0.35\linewidth]{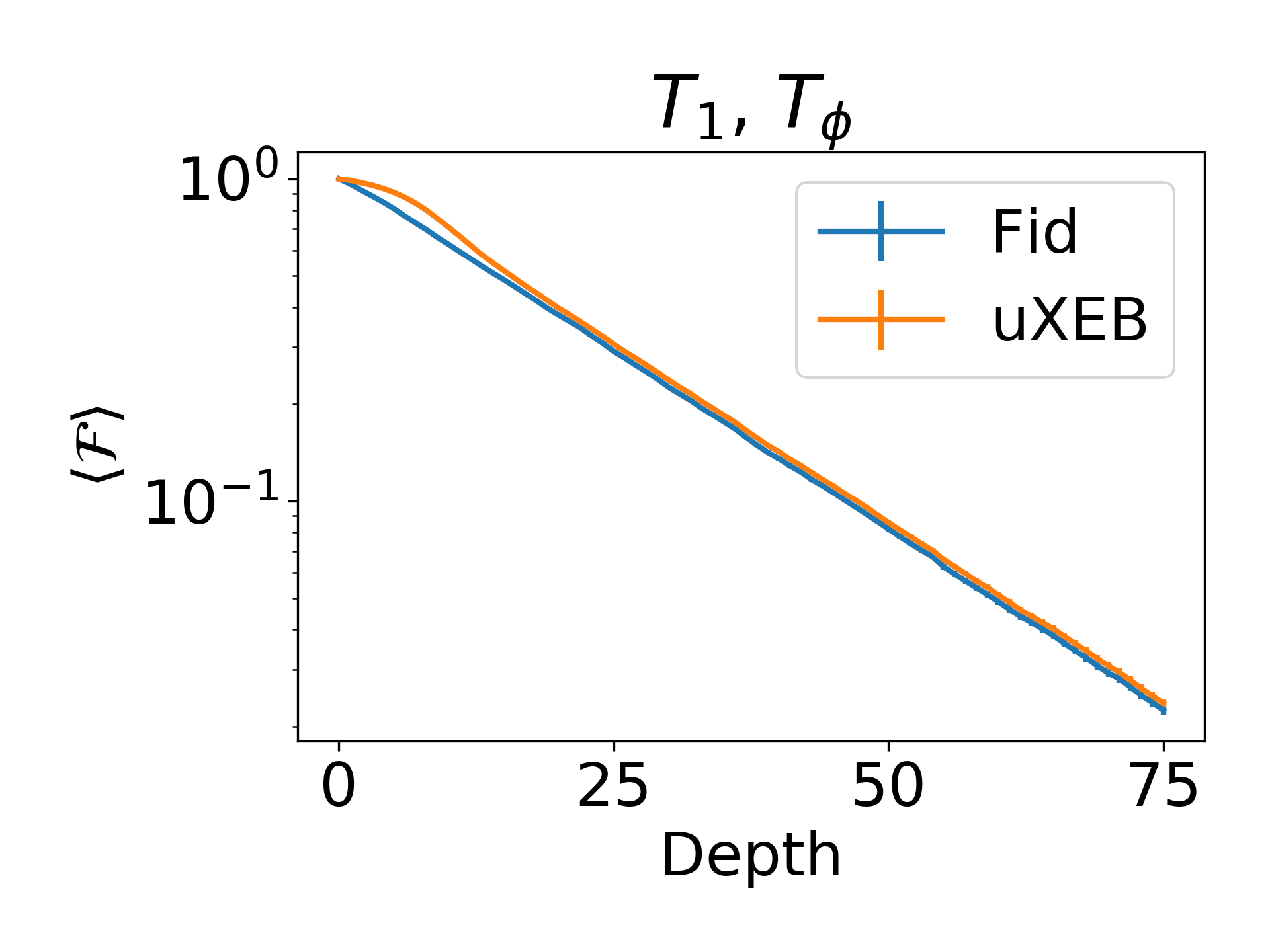}
    }
    \subfloat[]{
    \centering
    \includegraphics[width=0.35\linewidth]{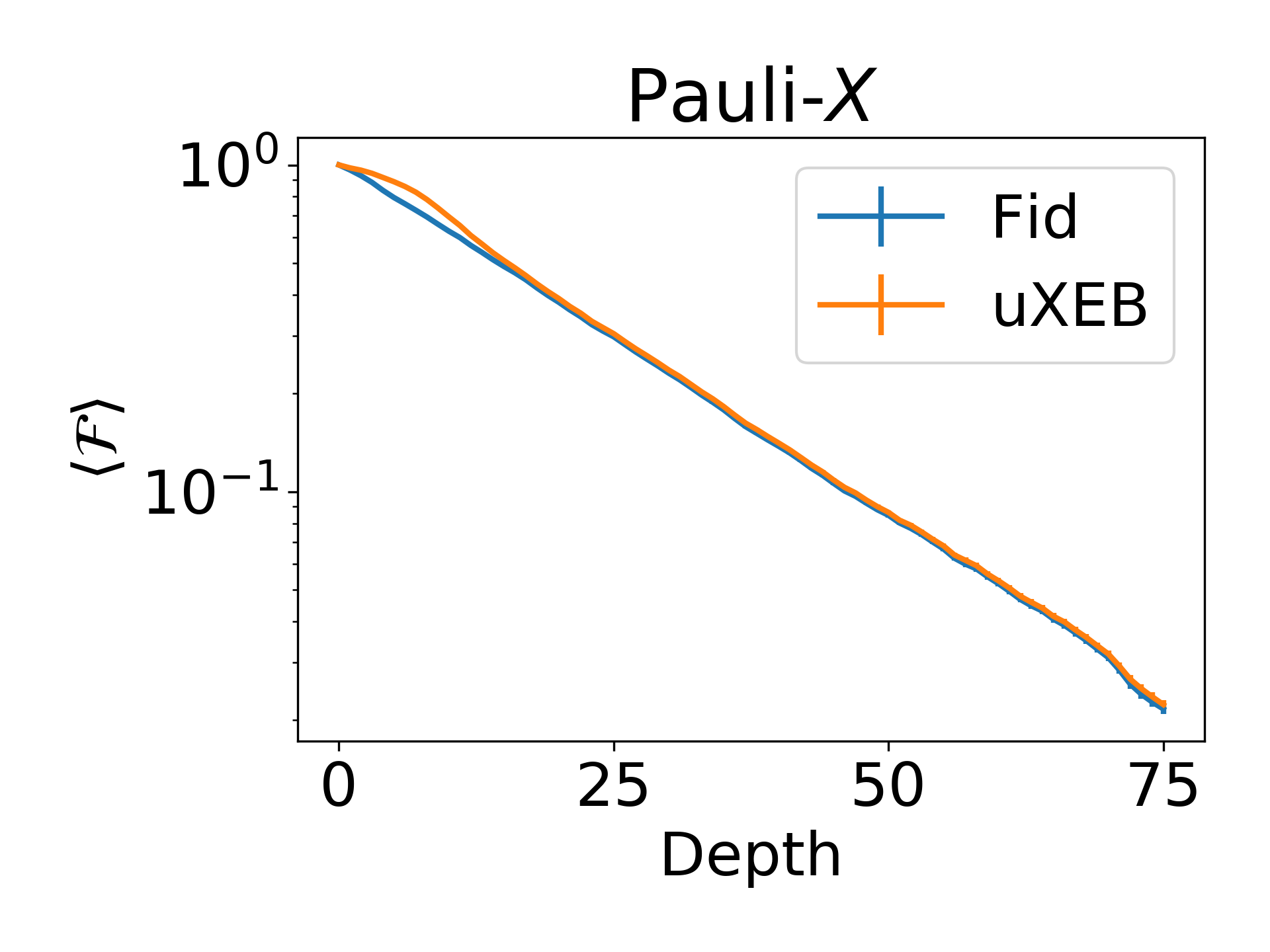}
    }
    \\
    \subfloat[]{
    \centering
    \includegraphics[width=0.35\linewidth]{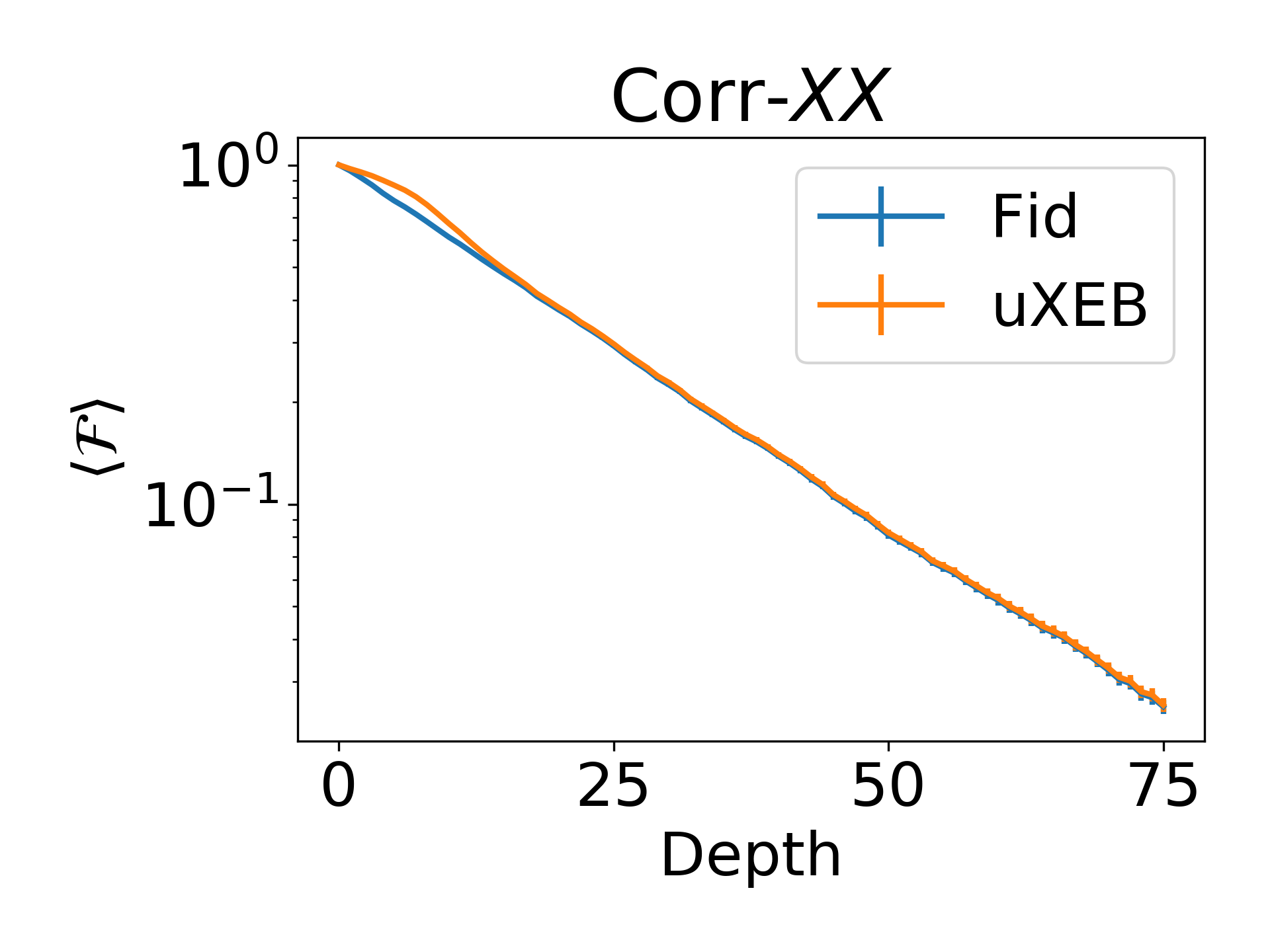}
    }
    \subfloat[]{
    \centering
    \includegraphics[width=0.35\linewidth]{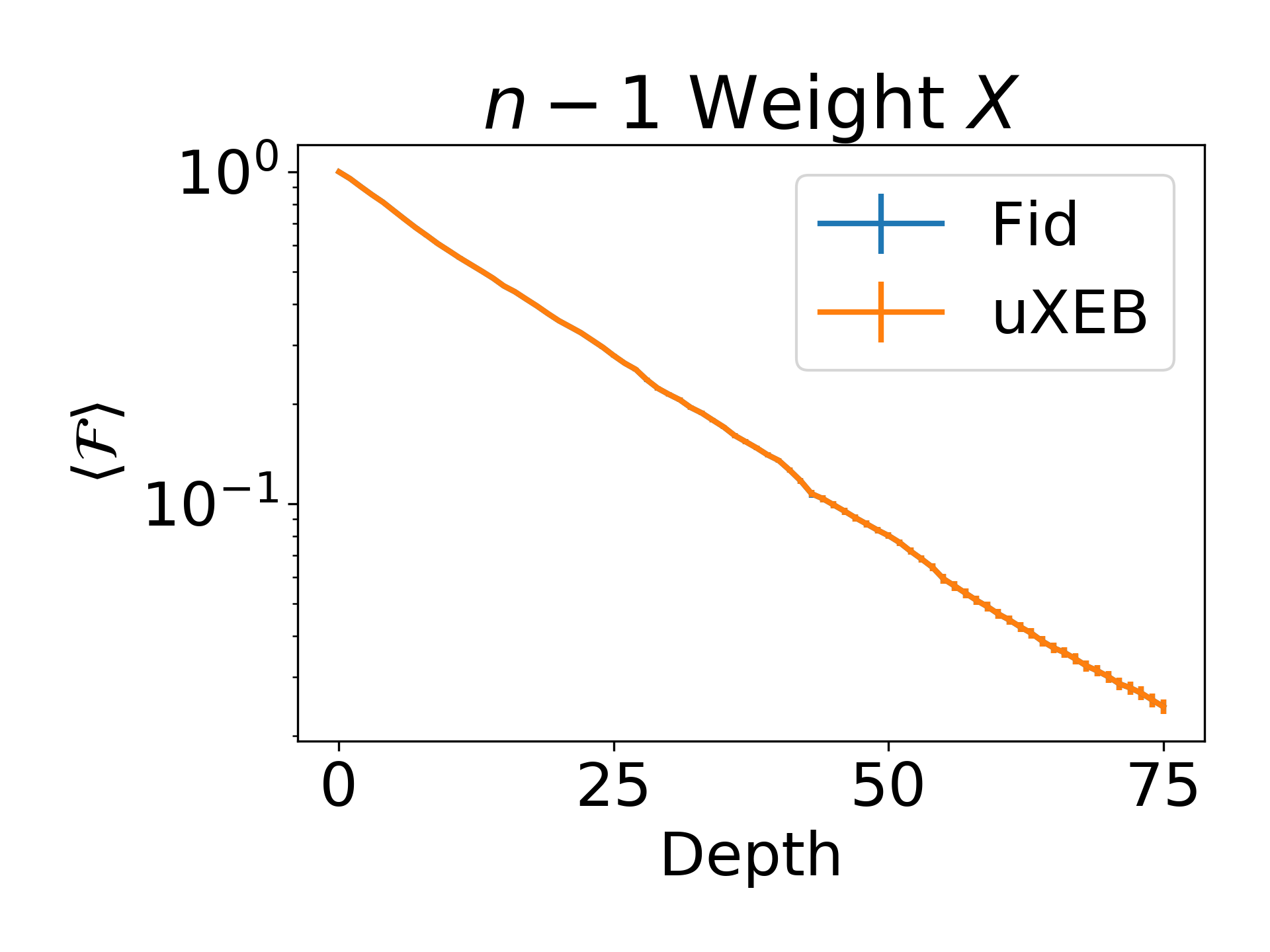}
    }

\caption{Numerical simulations using the Monte Carlo wave function (MCWF) method for various noise models. The system is modeled as perfect gates followed by evolution for one time unit under noisy channels~\cite{otten2019accounting} using the Lindblad master equation~\cite{otten-prb-2015} $\frac{\mathrm{d} \rho}{\mathrm{d} t} = \sum_i \gamma_i D[J_i](\rho)$, where the sum is over different noise channels, $D[J_i](\rho) = J_i \rho J_i^\dagger  - \frac{1}{2}(J_i^\dagger J_i \rho + \rho J_i^\dagger J_i)$ is a Lindblad superoperator for generic collapse operator $J_i$, and $\gamma_i$ controls the noise strength. The unbiased linear cross entropy agrees with the fidelity for all depths above a small threshold and correctly predicts the ENR. The noise models include: (a) single qubit amplitude-decay and pure dephasing, (b) single qubit Pauli-$X$ noise, (c) nearest-neighbor correlated $XX$ noise, (d) correlated $X$ noise with weight $n-1 = 19$, which is an artificial high-weight noise model. Here we simulate $n=20$ qubits on a 1D ring with noise strength $\gamma=0.0025$. Each global noise channel has an ENR of $\lambda_{\mathrm{true}} = n\gamma = 0.05$ by design. We average over
100 random circuits consisting of layers of two-qubit Haar-random unitaries and use 400 noise trajectories for each circuit at each depth. We fit the uXEB curves from depths 20 to 50.}
\label{fig:xeb_mcwf_simulations}

\vspace{5mm}
\begin{tabular}{ |c|c|c|c| }
    \hline
    Description & Lindblad &  $\lambda_{F}$ & $\lambda_{\mathrm{uXEB}}$ \\
    \hline
    $T_1$, $T_\phi$ & $\gamma D[\ketbra{0}{1}] + 2\gamma D[\ketbra{1}]$ & 0.0511(2) & 0.0511(2) \\\hline
Pauli-$X$ & $\gamma D[X]$ & 0.0508(2) & 0.0509(2) \\\hline
Corr-$XX$ & $\gamma D[X_i X_{i+1}]$ & 0.0505(3) & 0.0505(3) \\\hline
$n-1$ Weight $X$ & $\gamma D[\prod_{i\ne j} X_i]$ & 0.0506(3) & 0.0506(3) \\\hline
  \end{tabular}
  \captionof{table}{Curve fitting results for the numerical simulation in Fig.~\ref{fig:xeb_mcwf_simulations}. $\lambda_F$ and $\lambda_{\mathrm{uXEB}}$ shows the simulated RCS benchmarking result, which corresponds to the decay rate of fidelity and unbiased linear cross entropy, respectively.}
  \label{tab:noise_table}
\end{figure}

\vspace{1.5mm}
\noindent\textbf{Result 2: fidelity estimation and variance.} The above result outlines a procedure to extract $\lambda$ via estimating $\E F$ (Algorithm~\ref{alg:rcsbenchmarking}). Here the depth-independent coefficient $A$ corresponds to SPAM errors, and both $A$ and $\lambda$ can be extracted via curve fitting. To complete this we need a sample-efficient estimator of fidelity (line 4 and 5 of Algorithm~\ref{alg:rcsbenchmarking}). This is non-trivial as direct fidelity estimation (DFE)~\cite{flammia2011direct,dasilva2011practical} requires an exponential number of samples in the worst case. It has been suggested through heuristic arguments~\cite{Boixo2018Characterizing,arute2019quantum,choi2021emergent} that (unbiased) linear cross entropy appears to be a sample-efficient estimator for the fidelity of noisy random circuits. For a random circuit $C$ with output distribution $p_C(x)=\left|\mel{x}{C}{0^n}\right|^2$, the linear cross entropy estimator with $M$ samples $S=\{x_i\}_{i=1}^M$ is given by $\hat{F}_{\mathrm{XEB}}(S;C)=\frac{2^n}{M}\sum_{i=1}^M p_C(x_i) -1$. In experiments, after collecting the output samples, we perform exact classical simulation of the ideal circuit $C$ to compute the probabilities.
We also consider the unbiased linear cross entropy estimator~\cite{rinott2020statistical} defined as
\begin{equation}\label{eq:defunbiasedxeb}
    \hat{F}_{\mathrm{uXEB}}(S;C)=\frac{\frac{2^n}{M}\sum_{i=1}^M p_C(x_i) -1}{2^n\sum_{x\in\{0,1\}^n}p_C(x)^2-1}.
\end{equation}
The term ``unbiased" can be understood as follows: when the samples $S$ come from the ideal distribution $p_C(x)$, we have $\E_S\hat{F}_{\mathrm{uXEB}}(S;C)=1$, while $\E_S\hat{F}_{\mathrm{XEB}}(S;C)$ can be exponentially large. Note that for random quantum circuits the denominator $2^n\sum_{x\in\{0,1\}^n}p_C(x)^2-1$ approaches 1 in log depth~\cite{barak2020spoofing,dalzell2020random}, and therefore the two estimators give the same value as depth increases. In our experiments we use the unbiased linear cross entropy estimator by default, as it is more accurate at small constant depth. The main advantage of cross entropy estimators compared with DFE is that $O(1/\varepsilon^2)$ measurement samples suffice for estimating the fidelity of a random circuit within $\varepsilon$ additive error.

To further justify the connection between unbiased linear cross entropy and fidelity, we perform extensive numerical simulations under correlated noise models. Fig.~\ref{fig:xeb_mcwf_simulations} shows the results of these simulations, which include the exponential decay curves of fidelity and the unbiased linear cross entropy, and error bars correspond to the standard error of the mean across different circuits which are too small to be seen on the plot. Note that the unbiased linear cross entropy estimates the true fidelity very well in all noise models except for very small depths. The curve fitting results are shown in Table~\ref{tab:noise_table}. Here, both the decay rate of fidelity and unbiased linear cross entropy agree very well with the effective noise rate of the underlying noise model. These results verify our theoretical argument on the exponential decay of fidelity, as well as the correctness of unbiased linear cross entropy as a fidelity estimator, for both i.i.d. and highly correlated noise models. Additional simulation results with other system sizes, fidelity estimators, noise rates and gate sets are presented in Appendix~\ref{app:numerical}.

Next, we show evidence that the variance of cross entropy estimators for a random circuit scale as $O\left(1/M+\lambda^2\left(\E F\right)^2\right)$, where $M$ is the number of samples collected for each circuit (Appendix~\ref{sec:rcsvariance}). In a large scale experiment the second term is much smaller than the first term due to the exponential decay of $\E F$, and therefore it suffices to collect a large number of samples for few circuits to estimate $\E F$ within small additive error. These results provide evidence that supports Google's claim that only 10 random circuits with a large number of samples per circuit are sufficient to estimate the linear cross entropy in their experiment. On the other hand, the second term dominates for RCS benchmarking with a small number of qubits, and it is necessary to average over many ($\sim 100$) different circuits to obtain good error bars, such as in our experiments below. 

\begin{figure*}[t]
    \centering
   \subfloat[Simultaneous RB]{\includegraphics[width=0.43\linewidth]{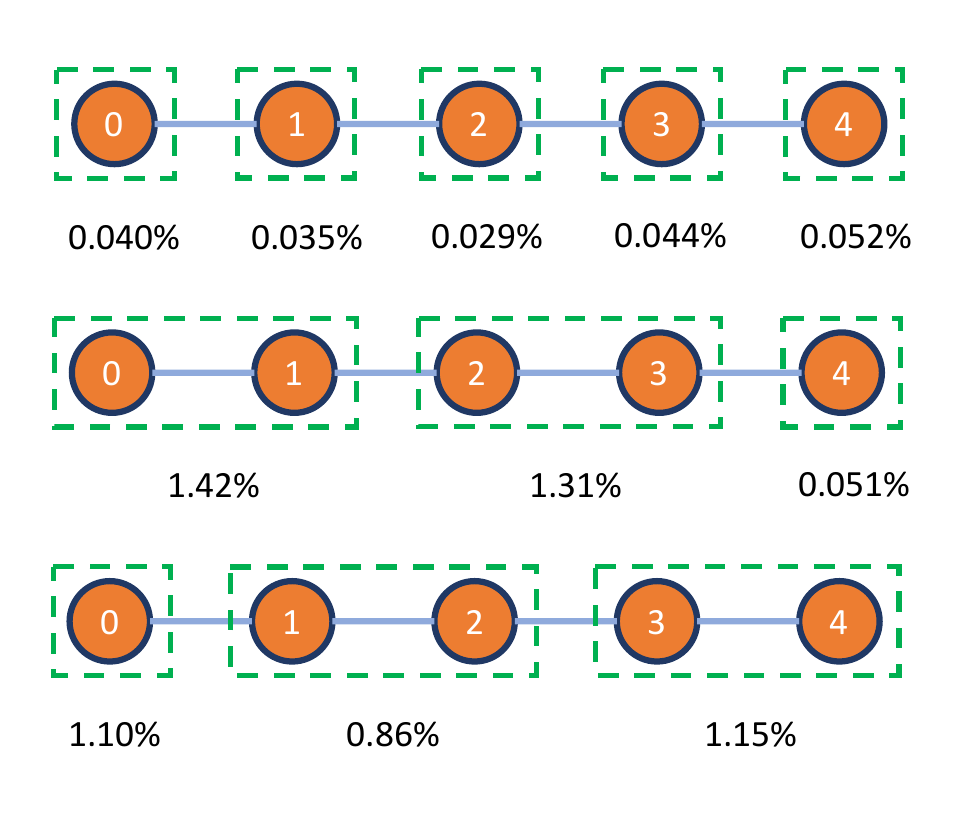}}
 \subfloat[RCS benchmarking]{\includegraphics[width=0.53\linewidth]{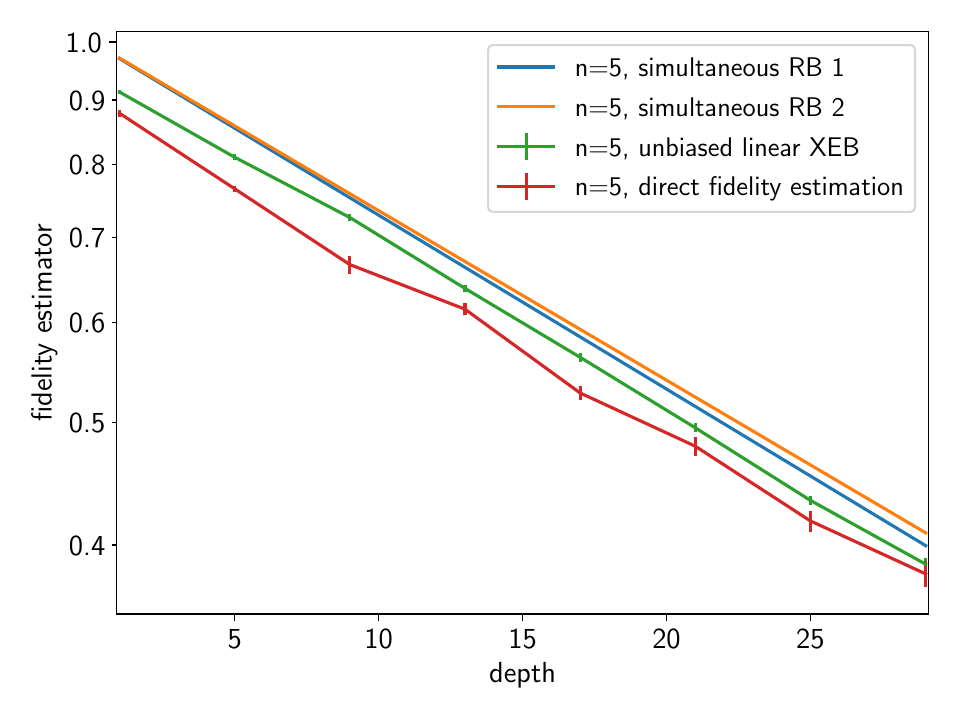}}
  \caption{Experimental implementation of RCS benchmarking on \texttt{ibmq\_athens} with 5 qubits. (a) Simultaneous RB, where parallel RB sequences are implemented for three gate patterns (green boxes) used in RCS benchmarking. The error rates underlying the single/two qubit green boxes are the RB results represented in Pauli error for Haar random single qubit gates and $\mathrm{CNOT}$ gates, respectively. (b) Exponential decay curves for RCS benchmarking with direct fidelity estimation and cross entropy. The decay rates, which represent the total amount of quantum noise per layer, are $\lambda_{\mathrm{DFE}}=3.05(5)\%$ and $\lambda_{\mathrm{uXEB}}=3.08(3)\%$. As a reference, the simultaneous RB estimator gives $\lambda_{\mathrm{sRB}}=3.13(4)\%$. We select 8 depth values ranging from 1 to 29. For DFE we implement 30 random circuits for each depth, and estimate the fidelity of each circuit by measuring 20 Fourier coefficients, which is 600 circuits for each depth. For cross entropy, we implement 100 random circuits for each depth. 8192 measurement samples are collected for each circuit. For both DFE and uXEB experiments, we use the standard error of the mean across different circuits as the error bar. The decay rate and its standard error are computed from the data points and error bars via standard least squares fitting. For the simultaneous RB estimator, we use half the difference between the two RB experiments as the standard error.}
  \label{fig:rcs5qubit}
\end{figure*}

\vspace{1.5mm}
\noindent\textbf{Result 3: experiments on IBM Quantum hardware.} We implement RCS benchmarking via experiments on IBM Quantum hardware~\cite{ibmquantum} with up to 20 superconducting qubits in one dimension. On these devices, $\mathrm{CNOT}$ is the only 2-qubit gate available, and arbitrary single-qubit gates are easy to implement, which have error rates that are roughly 2 orders of magnitude smaller than $\mathrm{CNOT}$. Therefore, in RCS benchmarking we are effectively measuring the total amount of quantum noise in a layer of $\mathrm{CNOT}$ gates. Fig.~\ref{fig:rcs5qubit} shows experiment results on a 5-qubit device, see Appendix~\ref{sec:experiments} for more details and larger experiments. 

Here we perform three types of experiments: simultaneous RB, RCS with direct fidelity estimation, and RCS with cross entropy. In simultaneous RB~\cite{Gambetta2012characterization}, the main idea is to perform different RB sequences in parallel instead of performing RB on one pair of qubits while all other qubits are idle. A similar experiment was performed in Google's experiment~\cite{arute2019quantum} where linear cross entropy was used as the post processing method instead of standard RB. We implement simultaneous RB on each of the three patterns shown in Fig.~\ref{fig:rcs5qubit}a. The numbers underlying single qubit boxes represent the error rates of two X90 pulses ($\sqrt{X}$), which can represent the Pauli error rate of a Haar random single qubit gate. The numbers underlying two qubit boxes represent the Pauli error rate of $\mathrm{CNOT}$ gates. The results demonstrate some crosstalk behavior. For example, notice that in the third pattern in Fig.~\ref{fig:rcs5qubit}a, there is a large (1.10\%) error rate on qubit 0, which is not present in the single qubit simultaneous RB. This suggests that a $\mathrm{CNOT}$ gate on qubit 1 and 2 can introduce a large error on qubit 0 due to crosstalk. Interestingly, in this experiment a $\mathrm{CNOT}$ gate on qubit 2 and 3 did not introduce additional errors on qubit 4.

Next we show results of RCS benchmarking with direct fidelity estimation and cross entropy. Note that DFE is not scalable due to the exponential sample complexity, and is implemented here mainly to verify our theoretical predictions. The results are shown in Fig.~\ref{fig:rcs5qubit}b, where all curves are exponential decays with roughly the same decay rate. The curves have different intercepts because DFE and uXEB experiments have different SPAM errors due to the additional overhead of DFE. From the curve fitting results, we can see that $\lambda_{\mathrm{DFE}}$ and $\lambda_{\mathrm{uXEB}}$ agrees with each other within the standard error. This confirms our theoretical results on the exponential decay of fidelity, and also verifies the validity of cross entropy as an efficient fidelity estimator. Also, note that uXEB is much more sample efficient than DFE, where the error bars for DFE are larger even when we collect 6 times the amount of samples in uXEB. As a reference, we implement simultaneous RB both before and after the RCS experiments, which can also be used to evaluate the error drift during the experiment period. In Fig.~\ref{fig:rcs5qubit}b we present heuristic fidelity estimators defined by multiplying individual gate fidelities measured by simultaneous RB experiments. Note that the simultaneous RB estimator gives a slightly larger prediction for the effective noise rate (also see below). 

\vspace{1.5mm}
\noindent\textbf{Application to diagnosing crosstalk.} Diagnosing and reducing crosstalk errors is a central step towards achieving fault-tolerance. These errors can be characterized as two types~\cite{Sarovar2020detectingcrosstalk}: the first is non-locality, where a correlated noise channel acts non-locally on multiple gates in the same layer; the second is dependence, where the noise channel on some gate depends on the other gates being implemented simultaneously. The second type can be demonstrated by comparing simultaneous RB results (such as in Fig.~\ref{fig:rcs5qubit}a). We show RCS benchmarking can provide information about the first type of crosstalk in a layer of arbitrary two-qubit gates.

In Google's experiment~\cite{arute2019quantum}, the second type of crosstalk is clearly present, as the average two-qubit gate error increases from 0.36\% as measured by individual RB to 0.62\% as measured by simultaneous RB. Regarding the first type of crosstalk, Google observed in their paper that their experimental linear cross entropy was consistent with a simple uncorrelated noise model $\hat{F}_{\mathrm{sRB}}=\prod_{i=1}^m(1-e_i)$, where $m$ is the number of gates and $e_i$ is the error rate of the $i$th gate measured by simultaneous RB. A similar behavior also happens in our experiment in Fig.~\ref{fig:rcs5qubit}b. Based on these experimental results, they claimed this coincidence indicated that the noise in their device was uncorrelated, i.e. the first type of crosstalk was not significant. Our results on RCS benchmarking provides formal evidence to support such a conclusion. This can be understood as follows. Imagine that a correlated (or high weight) Pauli error occurs in the device that acts on multiple gates, then it will be captured by multiple 2-qubit RB sequences. For example, if the error happens with 1\% chance and is supported on gate $i$ and gate $j$, then it will contribute 1\% to both $e_i$ and $e_j$. On the other hand, our result suggests that this error will only contribute 1\% in the effective noise rate $\lambda$. Therefore correlated errors across multiple gates will be counted multiple times by $\hat{F}_{\mathrm{sRB}}$ but only one time by the linear cross entropy and fidelity, so an experiment with significant correlated noise will demonstrate a deviation between $\hat{F}_{\mathrm{sRB}}$ and linear cross entropy. In Fig.~\ref{fig:sim_rb_overestimate} we show a concrete simulated experiment with correlated noise to demonstrate such a deviation. This provides formal evidence that correlated noise was not significant in Google's experiment.

\begin{figure}[t]
    \centering
    \includegraphics[width=0.5\linewidth]{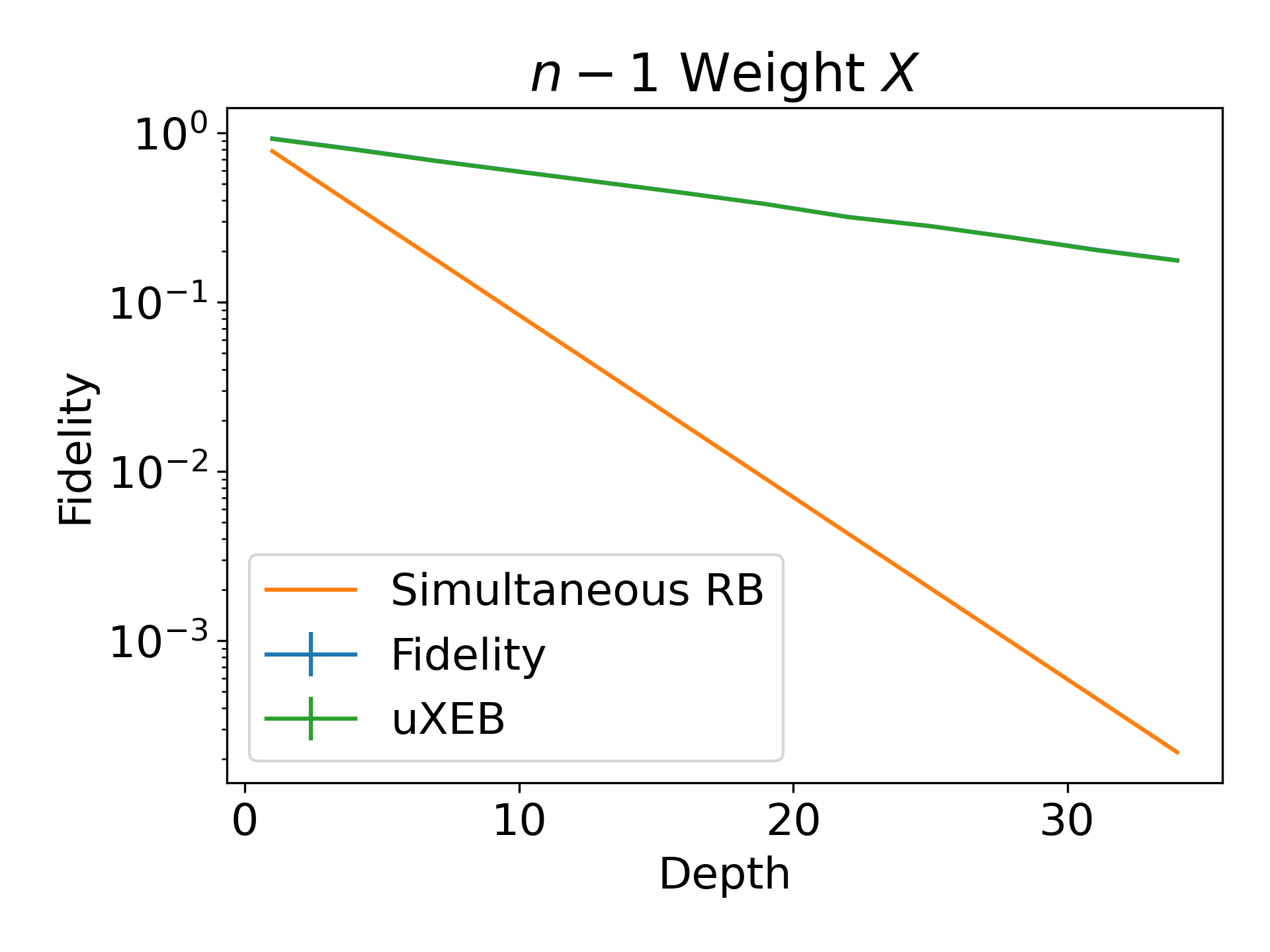}
    \caption{Simultaneous RB greatly overestimates noise when high-weight Pauli errors (the noise model in Fig.~\ref{fig:xeb_mcwf_simulations}d) are used. Here fidelity and uXEB curves overlap and are indistinguishable on the scale of the plot, and error bars are too small to be seen.}
    \label{fig:sim_rb_overestimate}
\end{figure}

\vspace{1.5mm}
\noindent\textbf{Discussion.} We show that random circuit sampling is a powerful tool for benchmarking quantum noise, which can be viewed as a practical application of sampling-based quantum supremacy experiments. As larger scale quantum devices are being built, an important task is to develop efficient benchmarking protocols that can jointly test the performance of all of the qubits. Several holistic benchmarking protocols are proposed for this purpose, including quantum volume~\cite{Cross2019validating}, accreditation protocols~\cite{Ferracin_2019,ferracin2021experimental} and others~\cite{BlumeKohout2020volumetricframework,mills2020applicationmotivated,dong2020random,proctor2020measuring,otten2019recovering}. RCS benchmarking can also be understood in this context, where the effective noise rate characterizes the global noise strength when all two-qubit couplings are turned on, which can help predict the fidelity of running large scale circuits as well as inform the design of error correction codes.

While we have demonstrated the feasibility of RCS benchmarking for $\sim 20$ qubits, there are two main challenges when considering running RCS benchmarking in a larger scale. First, as shown in our main result, a necessary condition for the correctness of RCS benchmarking is the effective noise rate $\lambda$ upper bouned by a small constant, that is, the effective noise rate per qubit scales as $O(1/n)$. This can be achieved if gate errors decrease as the number of qubits increases in hardware development. Second, the computation in RCS benchmarking becomes intractable when the number of qubits exceeds classical simulability, as computing the linear cross entropy (as well as other fidelity estimators) requires exact simulation of random quantum circuits, which makes the computation steps in RCS benchmarking inefficient. 

Here we discuss two potential ways to overcome the computation barrier of RCS benchmarking in order to scale to 50+ qubits, assuming the effective noise rate is sufficiently small. First, for RCS benchmarking with generic gate sets, we can design variants of the fidelity estimation procedure which requires simulating groups of correlated amplitudes, where tractable tensor network simulation algorithms exist~\cite{pan2021simulating}. Second, when the native 2-qubit gate is a Clifford gate, such as the $\mathrm{CNOT}$ gate in IBM's hardware platform, the computation barrier can be avoided by using random single qubit Clifford gates in between $\mathrm{CNOT}$ layers, where the output probabilities are easy to compute. We expect the exponential decay of fidelity to still hold in this case, as the analysis in our main result only involves second moments while Clifford circuits can generate unitary 3-design. 

\begin{acknowledgments}
We thank Yuri Alexeev, Yimu Bao, Adam Bouland, Soonwon Choi, Fred Chong, Steve Flammia, Pranav Gokhale, Zeph Landau, Richard Rines, Martin Suchara, Ryan Wu for helpful discussions. We would like to especially thank Umesh Vazirani for many helpful comments and suggestions. Y.L. would like to thank the IBM Quantum team for helpful discussions. Y.L. was supported by DOE NQISRC QSA grant \#FP00010905, Vannevar Bush faculty fellowship N00014-17-1-3025, MURI Grant FA9550-18-1-0161 and NSF award DMR-1747426. Part of the experiments were performed when Y.L. was a research intern at IBM. B.F. and R.B. acknowledge support from AFOSR (YIP number FA9550-18-1-0148 and FA9550-21-1-0008). B.F. additionally acknowledges support from the National Science Foundation under Grant CCF-2044923 (CAREER).  L.J. acknowledges support from the ARO (W911NF-18-1-0020, W911NF-18-1-0212), ARO MURI (W911NF-16-1-0349), AFOSR MURI (FA9550-19-1-0399), NSF (EFMA-1640959, OMA-1936118, EEC-1941583), NTT Research, and the Packard Foundation (2013-39273). This research used the Savio computational cluster resource provided by the Berkeley Research Computing program at the University of California, Berkeley (supported by the UC Berkeley Chancellor, Vice Chancellor for Research, and Chief Information Officer).

Y.L. and M.O. contributed equally to this work.
\end{acknowledgments}

\bibliography{ref}
\appendix
\tableofcontents
\section{Overview of RCS benchmarking}
\label{sec:rcsoverview}
In this section, we give an overview of the RCS benchmarking protocol after introducing basic notations, and then briefly introduce the results which we develop in the rest of the paper, including theory of RCS fidelity decay, fidelity and variance estimation. We also discuss the relationship between RCS and other benchmarking protocols.

\subsection{Setup and notations}

Fig.~\ref{fig:rcsbenchmarkingapp} shows the ensemble of random quantum circuits used throughout the paper. The system of qubits considered in our theoretical results and experiments are in one dimension, as shown in the figure, and we expect our results to be generalizable to higher dimensional lattices, where a similar alternating circuit architecture can be used such as in Google's experiment, as well as more general connectivity graphs. In the theory model of RCS (Fig.~\ref{fig:rcsbenchmarkingapp}a), we consider circuits which consist of random 2-qubit gates drawn independently from the Haar measure on $\mathbb{U}(4)$. This model is used for the analysis of fidelity decay in section~\ref{sec:rcstheory}. In practice, as random 2-qubit gates are hard to implement, we consider RCS with layers of fixed 2-qubit gates with random single qubit gates in between (Fig.~\ref{fig:rcsbenchmarkingapp}b). Our numerical simulation and experiments suggest that this model also creates an exponential decay in the same way as the theory model. The architecture in Fig.~\ref{fig:rcsbenchmarkingapp}b is suitable for implementation on current quantum platforms, which usually optimize for a fixed 2-qubit gate. For example, our experiments on IBM Quantum hardware use CNOT gates with random single qubit gates drawn from the Haar measure on $\mathbb{U}(2)$.

Let $\mathrm{RQC}(n,d)$ denote the ensemble of random quantum circuits with $n$ qubits and depth $d$ as shown in Fig.~\ref{fig:rcsbenchmarkingapp}. Here we define depth as the number of layers of 2-qubit gates, and both Fig.~\ref{fig:rcsbenchmarkingapp}a and \ref{fig:rcsbenchmarkingapp}b correspond to $n=5$ and $d=4$. An ideal implementation of a random circuit $C\sim\mathrm{RQC}(n,d)$ creates a pure state $\ket{\psi}=C\ket{0^n}$, while due to noise the experimental implementation corresponds to a mixed state $\rho$, and the fidelity of the circuit $C$ is defined as 
\begin{equation}
  F=\expval{\rho}{\psi}=\bra{0^n}C^\dag \rho C\ket{0^n},
\end{equation}
which is a random variable that depends on $C$. The average fidelity is then given by $\E_{C\sim\mathrm{RQC}(n,d)}F$, and we drop the subscript when unnecessary.

When assuming a gate-independent and Markovian noise channel between the layers, we show in section~\ref{sec:rcstheory} that $\E F\approx e^{-\lambda d}$ for shallow depth circuits using the model in Fig.~\ref{fig:rcsbenchmarkingapp}a, where $\lambda$ is the total amount of Pauli noise for each layer, which we define as the effective noise rate (ENR). More specifically, consider an $n$-qubit noise channel $\mc N$ which can be uniquely specified as
\begin{equation}
  \mc N(\rho)=\sum_{\alpha,\beta\in\{0,1,2,3\}^n}\chi_{\alpha\beta}\sigma_{\alpha}\rho \sigma_{\beta},
\end{equation}
where $(\chi_{\alpha\beta})$ is a positive semi-definite matrix known as the process matrix, and $\sigma_{\alpha}\in\{I,X,Y,Z\}^{\otimes n}$ is a $n$-qubit Pauli operator. The effective noise rate is given by the sum of diagonal elements of the process matrix which corresponds to non-zero Pauli errors,
\begin{equation}
  \mathrm{ENR}(\mc N)=\sum_{\alpha\in\{0,1,2,3\}^n\setminus\{0^n\}}\chi_{\alpha\alpha}.
\end{equation}
We also consider the effective noise rate per qubit (ENRq) defined as $\mathrm{ENRq}=\mathrm{ENR}/n$, and also denote these quantities by $\lambda$ and $\lambda_q$, respectively. When the noise channel is a Pauli channel (that is, $\chi$ is a diagonal matrix), the effective noise rate is simply the sum of probabilities of all non-zero Pauli operators.

In practice, however, noise is highly gate-dependent. For example, in today's quantum hardware the noise rate of 2-qubit gates are roughly two orders of magnitude higher than the noise rate of single qubit gates. In the practical implementation of RCS as in Fig.~\ref{fig:rcsbenchmarkingapp}b, as two qubit gates are fixed and only single qubit gates (with much smaller error rates) are random, RCS benchmarking can be effectively viewed as benchmarking the effective noise rate of the noise channel introduced by the layer of 2-qubit gates, or more precisely, the average ENR of the two alternating layers of 2-qubit gates. This noise rate extracted from RCS benchmarking captures all local and cross talk errors among all qubits.

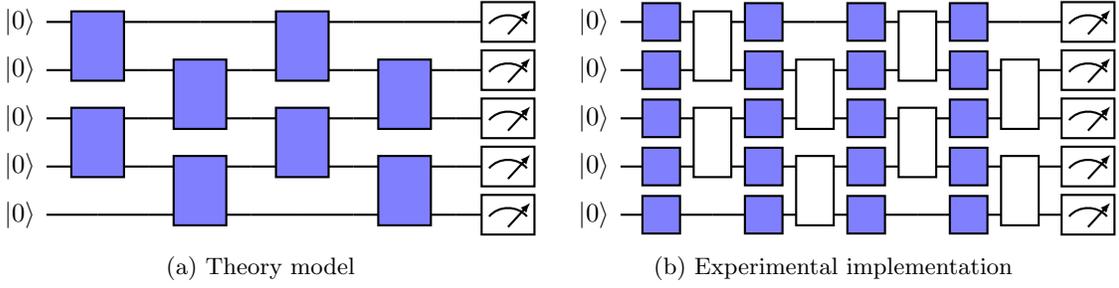
\begin{figure}[t]
  \tikzset{operator/.append style={fill=blue!50,minimum size=0.5cm}}
   \centering
   \subfloat[Theory model]{
    \centering
    \begin{quantikz}[column sep=0.33cm,row sep=0cm]
 \lstick{$\ket{0}$} & \gate[2][0.7cm]{} & \qw & \qw 		& \qw & \gate[2][0.7cm]{} & \qw & \qw		& \qw & \meter{} \\
 \lstick{$\ket{0}$} & 			& \qw & \gate[2][0.7cm]{} & \qw &			& \qw & \gate[2][0.7cm]{} & \qw & \meter{} \\
 \lstick{$\ket{0}$} & \gate[2][0.7cm]{} & \qw & 			& \qw & \gate[2][0.7cm]{} & \qw &			& \qw & \meter{} \\
 \lstick{$\ket{0}$} & 	 		& \qw & \gate[2][0.7cm]{} & \qw &			& \qw & \gate[2][0.7cm]{} & \qw & \meter{} \\
 \lstick{$\ket{0}$} & \qw 		& \qw & 			& \qw & \qw	    & \qw & 			& \qw & \meter{}
 \end{quantikz}}
 \subfloat[Experimental implementation]{
  \tikzset{operator/.append style={}}
    \centering
    \begin{quantikz}[column sep=0.4cm,row sep=0cm]
 \lstick{$\ket{0}$} & \gate[]{} & \gate[2,style={fill=white}]{} & \gate[]{} & \qw & \gate[]{} & \gate[2,style={fill=white}]{} &\gate[]{} & \qw & \meter{} \\
 \lstick{$\ket{0}$} & \gate[]{} & 			& \gate[]{} & \gate[2,style={fill=white}]{} & \gate[]{} &			& \gate[]{} & \gate[2,style={fill=white}]{}  & \meter{} \\
 \lstick{$\ket{0}$} & \gate[]{} & \gate[2,style={fill=white}]{} & \gate[]{} & 			& \gate[]{} & \gate[2,style={fill=white}]{} & \gate[]{} &			 & \meter{} \\
 \lstick{$\ket{0}$} & \gate[]{} & 	 		& \gate[]{} & \gate[2,style={fill=white}]{} & \gate[]{} &			& \gate[]{} & \gate[2,style={fill=white}]{}  & \meter{} \\
 \lstick{$\ket{0}$} & \gate[]{} & \qw 		& \gate[]{} & 			& \gate[]{} & \qw	    & \gate[]{} & 			 & \meter{}
 \end{quantikz}}\\
  \caption{RCS benchmarking with circuits consist of local random gates in an alternating architecture. (a) Theory model of RCS, where each two-qubit gate (blue box) is independently drawn from the Haar measure on $\mathbb{U}(4)$. (b) Experimental implementation of RCS, where only random single qubit gates are used (blue box) with fixed two-qubit gates (white box).}
  \label{fig:rcsbenchmarkingapp}
  \end{figure}
  
\begin{figure}[t]
  \begin{algorithm}[H]
    \caption{\raggedright RCS benchmarking (simplified)}\label{alg:rcsbenchmarkingapp}
    \raggedright\textbf{Input:} number of qubits $n$, maximum circuit depth $D$, number of circuits $L$\\
    \textbf{Output:} effective noise rate (ENR)
    \begin{algorithmic}[1]
    \For {$d=1\dots D$}
        \For {$i=1\dots L$}
        \State sample a random circuit $C_i\sim\mathrm{RQC}(n,d)$
        \State estimate the fidelity of $C_i$, denote as $\hat{F}_{d,i}$ \Comment{fidelity estimation}
        \EndFor
        \State $\hat{F}_d:=\frac{1}{L}\sum_{i=1}^L \hat{F}_{d,i}$
    \EndFor
    \State fit exponential decay $F=Ae^{-\lambda d}$ using data $\{\hat{F}_d\}_{d=1}^D$
    \State\textbf{Return} $\lambda$
    \end{algorithmic}
    \end{algorithm}
\end{figure}

\subsection{Fidelity estimation and variance}
\label{sec:fidelityestimation}
The core step of the RCS benchmarking protocol (Algorithm~\ref{alg:rcsbenchmarkingapp}) is fidelity estimation. Next we describe the fidelity estimation methods that we consider in RCS benchmarking and analyze the variance. Here we focus on sample complexity, which is the main bottleneck of the entire procedure for the scale that we consider ($\sim 20-30$ qubits). See section~\ref{sec:discussion} for discussions on computational complexity when considering RCS benchmarking in a scale that is beyond classical simulability.

For a small number of qubits, fidelity estimation can be done by running direct fidelity estimation (DFE) for each random circuit~\cite{flammia2011direct,dasilva2011practical}. Consider the Fourier expansion of the ideal state $\ketbra{\psi}=\frac{1}{2^{n/2}}\sum_\alpha \gamma_\alpha\sigma_\alpha$ with $\gamma_\alpha=\frac{1}{2^{n/2}}\Tr\left[\sigma_\alpha\ketbra{\psi}\right]$ and $\sum_\alpha\gamma_\alpha^2=1$. Let the noisy state from experiment be $\rho=\frac{1}{2^{n/2}}\sum_\alpha \gamma'_\alpha\sigma_\alpha$, then
\begin{equation}
    F=\expval{\rho}{\psi}=\sum_\alpha \gamma_\alpha\gamma'_\alpha=\E_{\alpha\sim \gamma_\alpha^2}\frac{\gamma'_\alpha}{\gamma_\alpha}.
\end{equation}
Therefore $F$ can be estimated by sampling a few Fourier coefficients from the distribution $\{\gamma_\alpha^2\}$, measuring them experimentally, and taking the empirical mean of $\gamma'_\alpha/\gamma_\alpha$. This procedure in general requires $O(2^n/\varepsilon^2)$ measurement samples in the worst case to obtain an estimate of $F$ within $\varepsilon$ additive error. In particular, if the circuit depth in RCS is sufficiently deep to generate a unitary 2-design, then the Fourier distribution of the output state is flat, which corresponds to the worst case in DFE. An interesting question is to study whether DFE can be improved for very low-depth random circuits.

Given the exponential sample complexity of DFE, we also consider sample efficient fidelity estimators based on cross entropy~\cite{Boixo2018Characterizing,Neill2018blueprint,arute2019quantum,rinott2020statistical}. In this work we justify the validity of cross entropy estimators by numerical simulation with different noise models and gate sets, and we leave the theoretical proof that cross entropy agrees with fidelity as important future work.

For a random circuit $C$ with output distribution $p_C(x)=\left|\mel{x}{C}{0^n}\right|^2$, the linear cross entropy estimator with $M$ samples $S=\{x_i\}_{i=1}^M$ is given by
\begin{equation}
    \hat{F}_{\mathrm{XEB}}(S;C)=\frac{2^n}{M}\sum_{i=1}^M p_C(x_i) -1.
\end{equation}
In experiments, after collecting the output samples, we perform exact classical simulation of the ideal circuit $C$ to compute the probabilities.
We also consider the unbiased linear cross entropy estimator~\cite{rinott2020statistical} defined as
\begin{equation}\label{eq:defunbiasedxebapp}
    \hat{F}_{\mathrm{uXEB}}(S;C)=\frac{\frac{2^n}{M}\sum_{i=1}^M p_C(x_i) -1}{2^n\sum_{x\in\{0,1\}^n}p_C(x)^2-1}.
\end{equation}
The term ``unbiased" can be understood as follows: when the samples $S$ come from the ideal distribution $p_C(x)$, we have $\E_S\hat{F}_{\mathrm{uXEB}}(S;C)=1$, while $\E_S\hat{F}_{\mathrm{XEB}}(S;C)$ can be exponentially large. Note that for random quantum circuits the denominator $2^n\sum_{x\in\{0,1\}^n}p_C(x)^2-1$ approaches 1 in log depth~\cite{barak2020spoofing,dalzell2020random}, and therefore the two estimators give the same value as depth increases. Different from the standard linear cross entropy, we need to classically simulate all $2^n$ output probabilities in order to compute the unbiased linear cross entropy estimator from experiment samples. In our experiments we use the unbiased linear cross entropy estimator by default, as it is more accurate at small constant depth. The main advantage of cross entropy estimators compared with DFE is that $O(1/\varepsilon^2)$ measurement samples suffice for estimating the fidelity of a random circuit within $\varepsilon$ additive error, which follows from the property that the output probabilities obey the Porter-Thomas distribution.

In the RCS benchmarking protocol (Algorithm~\ref{alg:rcsbenchmarkingapp}), an estimator for the average fidelity $\E F$ at depth $d$ is obtained by taking the empirical mean of cross entropy estimators of different random circuits independently drawn from $\mathrm{RQC}(n,d)$, given by $\frac{1}{L}\sum_{i=1}^L\hat{F}_{\mathrm{uXEB}}(S;C_i)$. The effective noise rate $\lambda$ is extracted by fitting the curve $\E F=Ae^{-\lambda d}$ using the estimators of $\E F$ with increasing depth. Similar to RB protocols, this fitting procedure decouples the decay rate $\lambda$ from SPAM errors, which is reflected in the depth-independent coefficient $A$.

In order to estimate the uncertainty of the benchmarking result $\lambda$, we need to estimate the variance of the estimator of $\E F$. The total variance is the sum of two parts: the variance of finite sampling for the fidelity estimation of each circuit, and the variance of fidelity across different circuits. In section~\ref{sec:rcsvariance} we give a theoretical model of the total variance,
\begin{equation}\label{eq:varianceofmean}
    \Var\left(\frac{1}{L}\sum_{i=1}^L\hat{F}_{\mathrm{uXEB}}(S;C_i) \right)=\frac{1}{L}O\left(\frac{1}{M}+\lambda^2\left(\E F\right)^2\right),
\end{equation}
which suggests that the strategy for choosing parameters (including the number of samples for each circuit $M$ and the number of circuits $L$) depends on the fidelity of the system. For a small number of qubits, it is necessary to choose a large $L$ due to the second term, and for a large number of qubits such as in Google's experiment, it suffices to choose a small $L$ and large $M$. In our numerical simulations and experiments on IBM Quantum hardware, we use the sample variance of $\hat{F}_{\mathrm{uXEB}}(S;C_i)$ across different circuits as the error bar, which is an unbiased estimator of the total variance.

The flexibility of RCS with respect to gate sets allows us to leverage the maximum amount of randomness that is efficiently implementable in the underlying hardware architecture in experimental implementations. In practice, the constants in Eq.~\eqref{eq:varianceofmean} depends on the gate set, and we observe that the constant can be decreased by increasing the randomness in the gate set, which decreases the sample complexity of the experiment. For example, an arbitrary single qubit gate can be implemented with two X90 pulses combined with phase control of microwave drive on the IBM Quantum hardware platform, therefore we use Haar random single qubit gates between CNOT layers in order to achieve the smallest variance.

\subsection{Relationship with other benchmarking protocols}
\label{sec:relationship}

In RCS benchmarking (Algorithm~\ref{alg:rcsbenchmarkingapp}), the (unbiased) linear cross entropy plays the role of fidelity estimator for random quantum circuits. The idea that (unbiased) linear cross entropy is a sample efficient fidelity estimator for random circuits was originally proposed by Google~\cite{Boixo2018Characterizing,arute2019quantum}. Ref.~\cite{arute2019quantum} also considered a form of unbiased linear cross entropy that is different from the one we use (Eq.~\eqref{eq:defunbiasedxebapp}). In Ref.~\cite{choi2021emergent}, a different form of unbiased linear cross entropy was used to estimate the fidelity of time-independent Hamiltonian evolution for Hamiltonians that lead to thermalization, where they develop an argument that connects the cross entropy estimator to the 2-design property of the projected state ensemble of a subsystem. Ref.~\cite{dalzell2021random} proved the exponential decay of linear cross entropy above log depth under i.i.d. noise models, which can be viewed as theoretical evidence for the agreement between linear cross entropy and fidelity. Ref.~\cite{gao2021limitations} also provides evidence that the effective noise rate $\lambda$ upper bounded by a small constant is necessary and sufficient for the agreement between linear cross entropy and fidelity, which is consistent with our results.

In this work we further justify the observation that cross entropy is a good fidelity estimator for random quantum circuits by performing numerical simulation with practically motivated noise models which capture amplitude and phase decay in superconducting qubits as well as correlated noise. Meanwhile, instead of focusing on estimating fidelity itself, our main result is to prove the exponential decay of fidelity under correlated noise and use this to extract the effective noise rate. 

In addition, in Google's experiment~\cite{arute2019quantum} it was shown that linear cross entropy can also be used as a post-processing method in benchmarking 2-qubit gates with a RB-like protocol. A detailed analysis of this method using global Haar random unitary gates was shown in~\cite{helsen2020general}. In this work, we focus on the application of linear cross entropy in benchmarking a non-trivial number of qubits, where global randomness is hard to implement and only local randomness is accessible. As a simple comparison, note that RCS benchmarking only requires two-qubit gate errors to be smaller than order $1/n$, while for standard RB with global Cliffords the gate errors must be smaller than order $1/n^2$.

Several scalable RB variants were proposed in order to overcome the scalability bottleneck of standard RB, including simultaneous RB~\cite{Gambetta2012characterization}, direct RB~\cite{Proctor2019direct}, and cycle benchmarking~\cite{Erhard2019characterizing}, among others (see Ref.~\cite{helsen2020general} for a comprehensive overview). A recent line of work uses a variant of simultaneous RB to obtain the full description of a gate-independent Pauli noise channel~\cite{Flammia2020efficient,Harper2020efficient,harper2020fast,flammia2021pauli}. Instead of using random global Cliffords, these works use random elements from small subgroups (subsets) of the global Clifford group, while it is unclear if they can be generalized to arbitrary gate sets without any group structure. 

On the one hand, RCS benchmarking can also be put into context of the general framework of RB. Consider \emph{two} layers in our random circuit architecture as a random element of a small subset of the global unitary group (note that this subset is not necessarily a group). This is similar to other ``subset RB" protocols~\cite{Brown2018randomized,Fran_a_2018,Proctor2019direct,Helsen2019new,Erhard2019characterizing} where the distribution of the element is only supported on a small subset of the global unitary/Clifford group. General conditions for subset RB to create exponential decays were also formulated~\cite{helsen2020general}. It is unclear if these previous analysis can be applied to RCS benchmarking which does not have a group structure.

On the other hand, RCS benchmarking is different from RB variants in two main aspects. First, unlike RCS, many scalable RB variants do not scramble across the entire system, and the state of the system remains in a tensor product of few qubit states. Therefore these RB variants create exponential decays that are different from RCS~\footnote{Note that an exception is direct RB~\cite{Proctor2019direct} which gives a heuristic argument of fidelity decay using random Clifford layers}. Second, RCS is very flexible with the gate set, while most RB variants assume Clifford gates with some cases extended to more general gates and finite groups~\cite{Magesan2012efficient,Carignan-Dugas2015characterizing,Cross2016scalable,Harper_2017,Fran_a_2018,Helsen2019new,Baldwin2020subspace,helsen2020general,helsen2020matchgate,Claes2021character}. These differences mainly come from a special property of random quantum circuits: while the analysis of RB and variants rely heavily on the group structure, the fast convergence to unitary designs for random quantum circuits only requires generic gate sets~\cite{Harrow2009random,Brandao2016,harrow2018approximate,haferkamp2020quantum,haferkamp2020improved}. This flexibility allows our RCS protocol to be implemented directly with the native gate set available on any hardware platform, including non-Clifford gates and gates with continuous parameters.

\section{Theory of RCS benchmarking}
\label{sec:rcstheory}

Next we develop the theory of fidelity decay and variance estimation in RCS benchmarking. Throughout this section we work in the theory model shown in Fig.~\ref{fig:rcsbenchmarkingapp}a, which assumes Haar random 2-qubit gates and a gate-independent noise channel acting on $n$ qubits. Our numerical simulation and experiments suggest that RCS benchmarking also works well in practice with less randomness and gate-dependent noise. 

\subsection{Reducing to Pauli noise}
\label{sec:reducetopaulinoise}
We start by showing that general noise channels can be reduced to Pauli noise channels, then develop theoretical and numerical results under Pauli noise. Consider an arbitrary $n$-qubit noise channel $\mc N(\rho)=\sum_{\alpha,\beta\in\{0,1,2,3\}^n}\chi_{\alpha\beta}\sigma_{\alpha}\rho \sigma_{\beta}$ with process matrix $(\chi_{\alpha\beta})$, we define $\mc N^{\mathrm{diag}}$ as the noise channel that has process matrix $\mathrm{diag}(\chi_{\alpha\beta})$. By definition, $\mc N^{\mathrm{diag}}$ is a Pauli noise channel that stochastically applies a Pauli operator $\sigma_\alpha$ with probability $\chi_{\alpha\alpha}$. In addition, $\mc N$ and $\mc N^{\mathrm{diag}}$ have the same effective noise rate which is given by
\begin{equation}
    \mathrm{ENR}(\mc N)=\mathrm{ENR}(\mc N^{\mathrm{diag}})=\sum_{\alpha\neq 0^n}\chi_{\alpha\alpha}.
\end{equation}

Consider two RCS benchmarking experiments with the same circuit architecture, where one has noise channel $\mc N$ and the other has $\mc N^{\mathrm{diag}}$. We show that these experiments have the same average fidelity.

\begin{theorem}\label{thm:reducetopauli}
The average fidelity of RCS benchmarking with noise channel $\mc N$ is equal to the average fidelity with noise channel $\mc N^{\mathrm{diag}}$. That is, without loss of generality we can assume that the underlying noise channel is Pauli noise.
\end{theorem}

Therefore to study the fidelity decay of RCS for general noise channels, it suffices to only consider Pauli noise. Intuitively, Theorem~\ref{thm:reducetopauli} holds because of linearity, where we consider diagonal terms $\sigma_\alpha\rho\sigma_\alpha$ and off-diagonal terms $\sigma_\alpha\rho\sigma_\beta$ ($\alpha\neq \beta$) separately, and show that the off-diagonal terms are ``killed" by the Haar random 2-qubit gates after averaging. Details of the proof are presented in Section~\ref{app:proofdetail}. Although RCS benchmarking is insensitive to off-diagonal terms, this is in general not a limitation, as there exist techniques such as randomized compiling~\cite{Wallman2016noise} that can convert the underlying noise channel to Pauli noise for Clifford+T circuits.

\subsection{Fidelity decay}
\label{sec:fidelitydecay}
\begin{figure*}[t]
  \tikzset{operator/.append style={fill=blue!50,minimum width=0.8cm}}
  \centering
  \begin{equation*}
    \E\left|
      \begin{adjustbox}{width=0.4\textwidth}
      \begin{quantikz}[column sep=0.4cm,row sep=0.4cm]
        \vdots &\\
        \lstick{$\bra{0}$} & \gate[2,label style=white]{G_1} & \qw 		     & \ \ldots\ \qw & \gate[2]{} & \qw                             & \qw		                       & \qw                                  & \gate[2]{} & \ \ldots\ \qw & \qw        & \gate[2,label style=white]{G_1^\dag} &\qw\rstick{$\ket{0}$}\\
        \lstick{$\bra{0}$} & 			                           & \gate[2]{}  & \ \ldots\ \qw &			      & \gate[2,label style=white]{}    & \qw                          & \gate[2,label style=white]{}         &            & \ \ldots\ \qw & \gate[2]{} &                                      &\qw\rstick{$\ket{0}$}\\
        \lstick{$\bra{0}$} & \gate[2,label style=white]{G_2} & 			       & \ \ldots\ \qw & \gate[2]{} &			                            & \qw                          &                                      & \gate[2]{} & \ \ldots\ \qw &            & \gate[2,label style=white]{G_2^\dag} &\qw\rstick{$\ket{0}$}\\
        \lstick{$\bra{0}$} & 	 		                           & \gate[2]{}  & \ \ldots\ \qw &			      & \gate[2,label style=white]{G_m} & \gate[style={fill=white}]{X} & \gate[2,label style=white]{G_m^\dag} &            & \ \ldots\ \qw & \gate[2]{} &                                      &\qw\rstick{$\ket{0}$}\\
        \lstick{$\bra{0}$} & \qw 		                         & 			       & \ \ldots\ \qw & \qw	      & 			                          & \qw                          &                                      & \qw        & \ \ldots\ \qw &            & \qw                                  &\qw\rstick{$\ket{0}$}\\
        \vdots & 
      \end{quantikz}
    \end{adjustbox}
      \right|^2=\frac{4}{15}\mc Z\left(
      \vcenter{\hbox{\includegraphics[height=3cm]{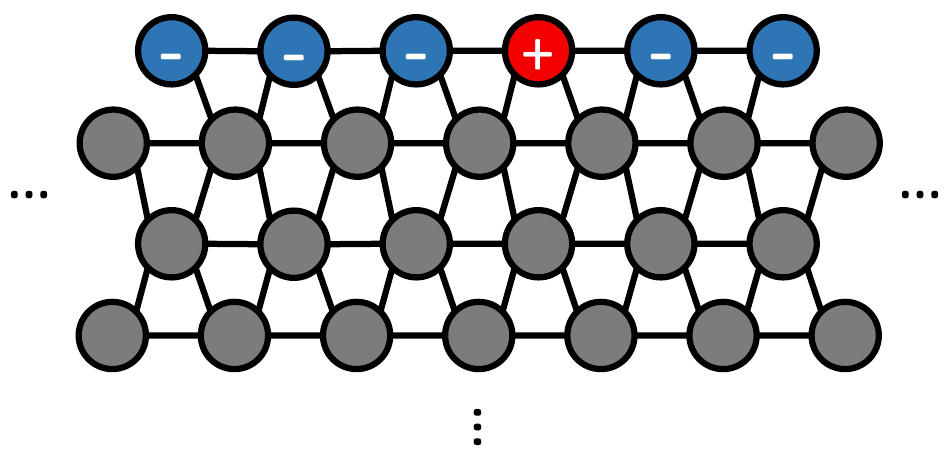}}}
      \right)-\frac{1}{15}
  \end{equation*}
\caption{Mapping random quantum circuits to a statistical mechanical model. LHS: a tensor network diagram for $\E\left|\braket{\psi_l}{\psi}\right|^2$, where an error happens at depth $l$ after gate $G_m$. RHS: this expectation value equals to the partition function of a classical statistical mechanical spin model. Here, each spin corresponds to a 2-qubit gate. On the top boundary, the $+$ spin (red) corresponds to the gate where the error happens ($G_m$), and the other gates at depth $l$ correspond to $-$ spins (blue). All the other spins (grey) can be either $+$ or $-$, and the partition function is the sum of weights of all possible configurations of the grey spins.
}
  \label{fig:rcsspinmodel}
\end{figure*}

Next we develop an argument showing that the average fidelity in RCS benchmarking decays as $\E F\approx e^{-\lambda d}$, where $\lambda$ is the effective noise rate of a global Pauli noise channel acting on all qubits after each layer of two-qubit gates, provided that $\lambda$ is a small constant.

Throughout the theoretical analysis we consider a 1D system of qubits with periodic boundary condition (i.e. they are placed on a ring). For simplicity, below we first present a proof sketch where we consider single qubit Pauli-$X$ noise channel acting on each qubit after each layer of gates, and then discuss generalizations afterwards. More detailed proofs as well as the extensions are presented in Section~\ref{app:proofdetail}. 

Consider single qubit Pauli-$X$ noise channel with Pauli error rate $\varepsilon$ defined as
\begin{equation}
    \mc E(\rho)=(1-\varepsilon)\rho + \varepsilon X\rho X.
\end{equation}
It is easy to see that the global noise channel $\mc N=\mc E^{\otimes n}$ has effective noise rate $\lambda=1-(1-\varepsilon)^n\approx n\varepsilon$. We can view each Pauli-$X$ noise channel as a stochastic process: with probability $\varepsilon$ an $X$ error is applied. We can then write the output density matrix of a $n$-qubit depth-$d$ noisy random circuit as a weighted sum of all possible error patterns,
\begin{equation}
    \begin{split}
        \rho&=(1-\varepsilon)^{nd}\ketbra{\psi}+\sum_{i=1}^{nd}\varepsilon (1-\varepsilon)^{nd-1}\ketbra{\psi_i}\\
    &+\sum_{i=1}^{nd-1}\sum_{j=i+1}^{nd}\varepsilon^2(1-\varepsilon)^{nd-2}\ketbra{\psi_{ij}}+\cdots
    \end{split}
\end{equation}
where $\ketbra{\psi}$ denotes the ideal output state, $\ketbra{\psi_i}$ denotes the ideal output state with an $X$ error at location $i$, etc. Each term in the above sum denotes the density matrix with a fixed number of errors happened in all possible ($nd$) locations. For example, the first term denotes the state with no error and is weighted by the corresponding probability $(1-\varepsilon)^{nd}$. Similarly the other terms denote the state with $1,2,\dots$ errors. We can therefore write the average fidelity as
\begin{equation}
\begin{split}
    \E F&=\E\expval{\rho}{\psi}\\
    &=(1-\varepsilon)^{nd}+\sum_{i=1}^{nd}\varepsilon (1-\varepsilon)^{nd-1}\E\left|\braket{\psi_i}{\psi}\right|^2\\
    &+\sum_{i=1}^{nd-1}\sum_{j=i+1}^{nd}\varepsilon^2(1-\varepsilon)^{nd-2}\E\left|\braket{\psi_{ij}}{\psi}\right|^2+\cdots\\
    &:=F_0+\E F_1+\sum_{k\geq 2}\E F_k.
\end{split}
\end{equation}
Note that $F_0\approx e^{-\varepsilon n d}$, our goal is therefore to prove that $\E F_1+\sum_{k\geq 2}\E F_k$ is small compared with $F_0$.

Next we make a first order approximation by ignoring the term $\sum_{k\geq 2}\E F_k$ and focus on $\E F_1$. Intuitively, the contribution to fidelity should decrease with the number of errors, provided that noise rate is sufficiently small. We verify the validity of this first order approximation via extensive numerical simulations in the next subsection and Section~\ref{app:proofdetail} and \ref{app:numerical}.

Next we focus on proving that $\E F_1$ is small compared with $F_0$. Formally speaking, this is a necessary condition to our main goal $\E F\approx F_0$, as all the higher order terms are positive. First, to simplify $\E F_1$, notice that by assuming a periodic boundary condition, at a fixed depth the specific qubit where the error happens does not matter. We can simplify $\E F_1$ as
\begin{equation}
    \E F_1=n\varepsilon (1-\varepsilon)^{nd-1}\sum_{l=1}^{d}\E\left|\braket{\psi_l}{\psi}\right|^2,
\end{equation}
simplifying the sum and bringing an extra factor $n$, where $\ket{\psi_l}$ denotes the state with an $X$ error at depth $l$ at the first qubit. The problem is then reduced to bounding the sum $\sum_{l=1}^d\E\left|\braket{\psi_l}{\psi}\right|^2$. See LHS of Fig.~\ref{fig:rcsspinmodel} for a demonstration of each term $\E\left|\braket{\psi_l}{\psi}\right|^2$.

Second, note that this sum is at least a constant, $\sum_{l=1}^d\E\left|\braket{\psi_l}{\psi}\right|^2=\Omega(1)$. This is simply because all terms are positive, and in the first term where an error happens at depth 1, most gates cancel with the conjugate except one 2-qubit gate, and 
\begin{equation}
    \E\left|\braket{\psi_1}{\psi}\right|^2=\E_{U\sim\mathbb{U}(4)}|\bra{00}U^\dag X U\ket{00}|^2=\frac{1}{5}.
\end{equation}
Our main result proves a tight upper bound $\sum_{l=1}^d\E\left|\braket{\psi_l}{\psi}\right|^2=O(1)$. This implies that 
\begin{equation}
    \E F_1/F_0=O(n\varepsilon)=O(\lambda).
\end{equation}

Third, note that all of the above arguments can be directly extended to general Pauli noise channels, where the only difference is that $\ket{\psi_l}$ has a general Pauli error at depth $l$ instead of a single qubit Pauli. We extend our rigorous analysis of $\sum_{l=1}^d\E\left|\braket{\psi_l}{\psi}\right|^2$ up to 3-local errors. While this captures most error sources in quantum device, we expect the same result to hold for general Pauli errors with arbitrary weight and locality, as supported by numerical evidence shown in Section~\ref{app:proofdetail} and \ref{app:numerical}.

\begin{theorem}\label{thm:fidelitydecayapp}
For random quantum circuits in 1D with Haar random 2-qubit gates and 3-local noise channel with effective noise rate $\lambda$, the average fidelity is given by
\begin{equation}
    e^{-\lambda d}\leq \E F\leq e^{-\lambda d}(1+K\lambda)
\end{equation}
up to a first-order approximation in $\lambda$. Here $K$ is a universal constant, and we assume $d\ll 2^n$.
\end{theorem}

Our results suggest that the fidelity decay $\E F\approx e^{-\lambda d}$ is a good approximation when $\lambda$ is small, that is, when the effective noise rate per qubit scales like $1/n$. This requirement can be satisfied by the error rates in current quantum hardware. In addition, it is easy to see that Theorem~\ref{thm:fidelitydecayapp} does not hold when circuit depth $d\to\infty$, as in this case $\E F\to\frac{1}{2^n}$ under depolarizing noise while the bounds goes to 0. Here our results mainly focus on low-depth random quantum circuits, which corresponds to the setting that is experimentally implementable.

Our results may seem surprising when comparing with other known results about random quantum circuits. In particular, many properties about random quantum circuits, such as convergence to unitary $t$-designs, are known to hold only above a certain depth~\cite{Brandao2016}, while our results hold already at shallow depth without any threshold requirement. This is partly because our fine grained analysis of $\E\left|\braket{\psi_l}{\psi}\right|^2$ works for any $l=1,\dots,d$, and in particular we show that 
\begin{equation}\label{eq:expdecay}
    \E\left|\braket{\psi_l}{\psi}\right|^2\leq e^{-\Delta l}+\frac{1}{2^n}
\end{equation}
for some constant $\Delta>0$. This suggests that for any $l$, random quantum circuits scramble the error at depth $l$ exponentially fast.

It is easy to see that Eq.~\eqref{eq:expdecay} implies Theorem~\ref{thm:fidelitydecayapp}. Next we give a simplified proof of Eq.~\eqref{eq:expdecay} for single qubit errors; see Section~\ref{app:proofdetail} for details and extensions to higher weight errors. First, note that in $\E\left|\braket{\psi_l}{\psi}\right|^2$ the gates that are applied after the error are canceled with the conjugate. Write the state as $\ket{\psi_l}=C_2 X C_1\ket{0^n}$ and $\ket{\psi}=C_2 C_1\ket{0^n}$ with unitary operators $C_1$ and $C_2$, then $C_2$ is canceled out and we have
\begin{equation}
    \E\left|\braket{\psi_l}{\psi}\right|^2=\E_{C_1\sim\mathrm{RQC}(n,l)}\left|\bra{0^n}C_1^\dag X C_1\ket{0^n}\right|^2,
\end{equation}
therefore Eq.~\eqref{eq:expdecay} is a property about depth-$l$ random quantum circuits.

Second, we evaluate this expectation (LHS of Fig.~\ref{fig:rcsspinmodel}) by taking the expectation of each 2-qubit gate, resulting in a classical statistical mechanical spin model (RHS of Fig.~\ref{fig:rcsspinmodel}). This technique of mapping to classical spin models has been widely used in recent study of random quantum circuits, see e.g.~\cite{Nahum2018operator,Zhou2019emergent,hunterjones2019unitary,bao2020theory,barak2020spoofing,dalzell2020random}. Here, as the expectation $\E\left|\braket{\psi_l}{\psi}\right|^2$ is a second moment property, that is, by linearity of expectation, each 2-qubit gate appears as a rank-2 projector $\E_{U\sim\mathbb{U}(4)}\left[U^{\otimes 2} \otimes U^{*\otimes 2}\right]$, we can represent each 2-qubit gate with a classical spin with 2 degrees of freedom. As a result, the expectation value is related to the partition function of a classical spin model (RHS of Fig.~\ref{fig:rcsspinmodel}), which lives on a triangular lattice with three-body interactions. The weights of the interactions are given in Fig.~\ref{fig:spinmodeldetails}a, and boundary conditions are presented in Fig.~\ref{fig:rcsspinmodel}. See Section~\ref{app:proofdetail} for detailed derivations.

\begin{figure}[t]
   \centering
   \subfloat[Weights]{
    \centering
    \includegraphics[width=0.4\linewidth]{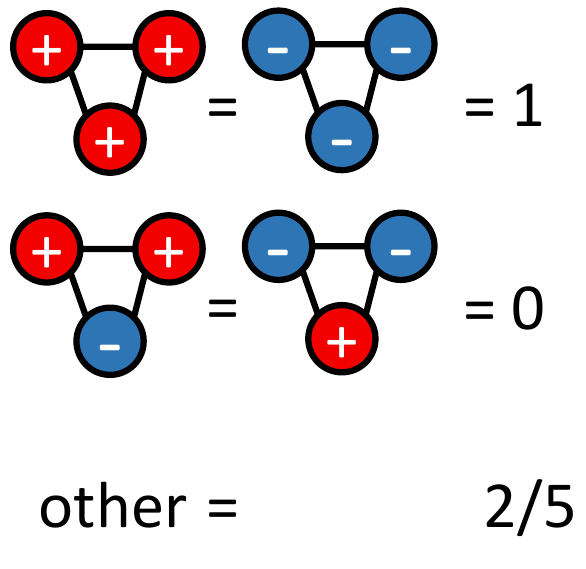}
    }
 \subfloat[Domain wall]{
 \centering
 \includegraphics[width=0.3\linewidth]{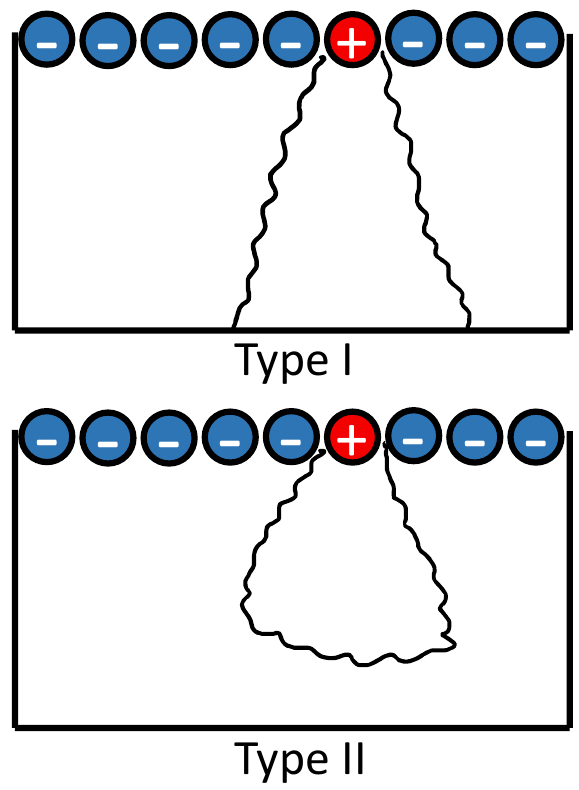}
 }
  \caption{Details of the classical spin model. (a) Weights of the three-body interaction. (b) A non-zero weight configuration corresponds to the configuration of two domain walls, which either intersect and annihilate (Type II) or never intersect (Type I).}
  \label{fig:spinmodeldetails}
\end{figure}

Finally, we prove Eq.~\eqref{eq:expdecay} by showing an analytical bound for the partition function of the spin model. The proof follows from a domain wall (boundary between clusters of $+$ and $-$ spins) argument. Note that from the constraints given by the weights (Fig.~\ref{fig:spinmodeldetails}a), if two spins at the top of a triangle have the same sign, then the spin at the bottom also has the same sign. Therefore the spin configuration that has a non-zero weight must correspond to two domain walls due to the top boundary condition. The domain wall configuration has two cases (Fig.~\ref{fig:spinmodeldetails}b): they either intersect and annihilate (Type II) or never intersect (Type I). We can therefore evaluate the expectation as
\begin{equation}
    \E\left|\braket{\psi_l}{\psi}\right|^2=\frac{4}{15}\left(\mc Z_1+\mc Z_2\right)-\frac{1}{15},
\end{equation}
where $\mc Z_i$ denotes the sum of weights of domain wall configuration of Type $i$. By a simple counting argument, $\mc Z_1\leq (4/5)^{2(l-1)}$, and we also show that $\frac{4}{15}\mc Z_2-\frac{1}{15}\leq \frac{1}{2^n}$, which concludes the proof.

\subsection{Fidelity estimation}
Our main results (Theorem~\ref{thm:reducetopauli} and \ref{thm:fidelitydecayapp}) show rigorous evidence that $\E F\approx e^{-\lambda d}$ for noisy random circuits, 
where $\lambda$ is the effective noise rate for the global noise channel. The main idea of RCS benchmarking is to efficiently estimate $\E F$ for varying depth, and then extract $\lambda$ by fitting an exponential decay curve. Next we provide a detailed investigation of fidelity estimators and develop numerical simulation techniques to verify their correctness.

Let $C$ be a random circuit. The goal is to estimate the fidelity of an unknown quantum state $\rho$ prepared in experiment, with respect to the pure state $\ket{\psi}=C\ket{0^n}$. As discussed partly in section~\ref{sec:fidelityestimation}, estimating the fidelity of an unknown quantum state is in general a hard task, which either requires exponential copies of $\rho$ via direct fidelity estimation~\cite{flammia2011direct,dasilva2011practical}, or requires additional circuit depth overhead that is unrealistic with current hardware~\cite{Huang2020predicting}.

Motivated by the insight developed by Google~\cite{Boixo2018Characterizing,arute2019quantum} that the special property of random quantum circuits allows fidelity to be efficiently estimated from output samples via cross entropy, we consider fidelity estimators of the following general form,
\begin{equation}\label{eq:generalfidelityestimator}
    \expval{\rho}{\psi}\approx \hat{F}(S;C)
\end{equation}
where $S=\{x_1,\dots,x_M\}$ are samples from the output distribution of the noisy circuit, i.e. computational basis measurement results of $\rho$. For example, the linear cross entropy $\hat{F}_{\mathrm{XEB}}(S;C)=\frac{2^n}{M}\sum_{i=1}^M p_C(x_i) -1$ has the above form which is related to the sum of probabilities of output samples. Several unbiased and non-linear variants of cross entropy were also proposed~\cite{arute2019quantum,Boixo2018Characterizing,rinott2020statistical,choi2021emergent}, and we use the unbiased version defined in Eq.~\eqref{eq:defunbiasedxebapp} which we find has the best performance. In general these fidelity estimators are sample efficient, requiring only $M=O(1/\varepsilon^2)$ measurement samples to achieve $\varepsilon$ additive accuracy (see section~\ref{sec:rcsvariance}), and their correctness is based on heuristic arguments and numerical simulation. An important future direction is to provide rigorous justifications as well as develop other efficient fidelity estimators.

Next we show numerical simulation results for noisy random circuits to verify our results on the exponential decay of fidelity, and also show that cross entropy estimators agree well with the fidelity. To accomplish this, we model the system as perfect gates
followed by evolution for one time unit under noisy channels~\cite{otten2019accounting} using the Lindblad master equation~\cite{otten-prb-2015},
\begin{equation}
    \frac{\mathrm{d} \rho}{\mathrm{d} t}
  = \sum_i \gamma_i D[J_i](\rho),
\end{equation}
where the sum is over different noise channels, $D[J_i](\rho) = J_i \rho J_i^\dagger  - \frac{1}{2}(J_i^\dagger J_i \rho + \rho J_i^\dagger J_i)$ is a Lindblad superoperator for generic collapse operator $J_i$, and $\gamma_i$ is a parameter that controls the noise strength.

We implement the evolution of this open quantum system in the open-source simulator, QuaC~\cite{QuaC:17}. Importantly, for the numerical simulation of 20 noisy qubits, the naive 
density matrix simulation is inefficient. We instead use the Monte Carlo wave function (MCWF) method~\cite{plenio1998quantum}, which simulates random ``quantum jumps",
rather than the full dynamics of the density matrix, to reduce the total computational cost of the simulations. We simulate 1D rings of $n$ qubits with
periodic boundary conditions. Full details of the computational method can be found
in Section~\ref{app:numerical}. 

Using the MCWF technique, we simulate a variety of noise models, as summarized in Table~I in main text. These noise models are: 
\begin{enumerate}
    \item $T_1$ and $T_\phi$, which includes single qubit amplitude decay and pure dephasing, and represent the primary noise sources in superconducting qubits~\cite{devoret2013superconducting};
    \item i.i.d. single qubit Pauli-$X$ noise, which models single qubit bit-flip;
    \item nearest-neighbor correlated $XX$ noise, where a Pauli-$XX$ noise channel is applied to all neighboring qubit pairs, which models two-body incoherent coupling;
    \item $n-1$ - weight Pauli-$X$ noise, where a Pauli-$X^{\otimes n-1}$ noise channel is applied to all subsets of size $n-1$, which is an artificial noise model used to test RCS benchmarking with extremely high-weight noise.
\end{enumerate}

Note that the identity of the Pauli operators in the noise models (whether it's $X$, $Y$ or $Z$) does not matter due to averaging over random circuits. The effective noise rate of each of these noise channels is related to the noise strength $\gamma$ in the Lindblad superoperator by a constant factor, while the constant differs for each noise model (see Section~\ref{app:numerical} for details). We manually adjust the coefficients $\gamma_i$ such that the effective noise rate for all noise models in Table~I of main text are equal to $\lambda=n\gamma$, where $\gamma$ is a parameter we control. We simulate 1D rings of $n=20$ qubits, averaging over
100 random circuits consisting of layers of two-qubit Haar-random unitaries and use 400 noise trajectories for each circuit at each depth. 
We fit the uXEB curves from depths 20 to 50.
Here the fidelity and cross entropy for each circuit is calculated by averaging over the stochastic noise trajectories, which is more efficient than simulating individual measurement samples (see Section~\ref{app:numerical} for more details). 

Fig.~2 and Table I in main text shows the results of these simulations, which include the exponential decay curves of fidelity and the unbiased linear cross entropy, and error bars correspond to the standard error of the mean across different circuits which are too small to be seen on the plot. Note that the unbiased linear cross entropy estimates the true fidelity very well in all noise models. In Fig.~\ref{fig:other_mcwf_simulations} in Section~\ref{app:numerical} we also plot several other fidelity estimators using the same data, where they are far from true fidelity until they converge at around depth 15. As discussed in section~\ref{sec:fidelityestimation}, this is because the unbiased version corrects an additional normalization factor which converges to 1 very quickly. Additional simulation results with other system sizes, fidelity estimators, noise rates and gate sets can be found in Section~\ref{app:numerical}.

For comparison, we also consider simultaneous RB~\cite{Gambetta2012characterization}, an alternative benchmarking method that can also be used to heuristically estimate the fidelity. Here, two-qubit RB sequences are simultaneously executed for all neighboring qubit pairs, and a Pauli error rate $e_i$ can be extracted for all two qubit couplings from the RB decay curve. See section~\ref{sec:simurb} for a concrete example. Then, the fidelity of a random circuit with $m$ 2-qubit gates can be estimated by

\begin{equation}\label{eq:simurbestimator}
    \hat{F}_{\mathrm{sRB}}=\prod_{i=1}^m (1-e_i).
\end{equation}
First, note that $\hat{F}_{\mathrm{sRB}}$ is an accurate fidelity estimator for random quantum circuits when the noise sources are limited to single qubits, as the error rates can be estimated from the RB decay, and the estimator $\hat{F}_{\mathrm{sRB}}$ agrees with $\E F$ due to our Theorem~\ref{thm:fidelitydecayapp}. Second, when there is crosstalk, intuitively simultaneous RB can capture some correlated noise. For example, suppose when qubit pair 1 is turned on, it generates an error on a neighboring qubit pair 2. Since the RB sequences are executed simultaneously, this error can be captured by the RB sequence on qubit pair 2. However, this intuition can fail when we have high weight correlated errors, as the error could be over counted by multiple RB sequences (see ``Application to diagnosing crosstalk" and Fig.~4 in main text). 

\subsection{Virtual experiment for extracting correlated noise}

\begin{figure}
    \centering
    \includegraphics[width=0.5\linewidth]{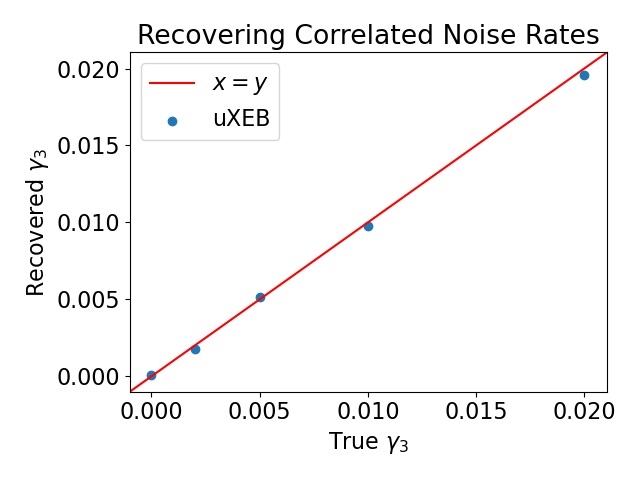}
    \caption{The amount of correlated noise, $\gamma_3$, can be readily extracted by combining amplitude decay, Ramsey, and RCS experiments.}
    \label{fig:correlated_error_model}
\end{figure}

Next, we demonstrate a virtual experiment, where we show that RCS benchmarking can be used to extract the amount of correlated noise in conjunction with other benchmarking methods. Here, we assume that the underlying noise model is known to be a combination of $T_1$, $T_\phi$, and a correlated $ZZ$ noise channel, but the values of the individual noise rates is unknown. For simplicity, we assume that all noise channels of the same type (e.g., all $T_1$ processes) have the same noise rate across different qubits, but the noise rates can differ between types ($T_1$, $T_\phi$, and the correlated $ZZ$ noise channels can have different strengths). This noise
model is described by the Lindblad
\begin{equation}\label{eq:ext_noise_model}
    \mathcal{L}(\rho) = \sum_{i=1}^n \left(\gamma_1 D(\sigma_i)[\rho] + \gamma_2 D(\sigma_i^\dagger \sigma_i)[\rho] + \gamma_3 D(Z_i Z_{i+1})[\rho]\right).
\end{equation}
Because we know the form of the noise channel, we can write down an explicit equation
for the total effective noise rate in terms of the amplitude decay rate $\gamma_1=\frac{1}{T_1}$,
the pure dephasing rate $\gamma_2 = \frac{1}{T_\phi}$ times, as well as the correlated noise rate, $\gamma_3$,
\begin{equation}\label{eq:extractcorr2}
    \lambda = n \Big(\frac{\gamma_1}{2} + \frac{\gamma_2}{4} + \gamma_3\Big), 
\end{equation}
where the prefactors for $\gamma_1$ and $\gamma_2$ are calculated in Section~\ref{app:numerical}.
The $T_1$ time can be measured via a simple amplitude decay experiment, giving an equation
\begin{equation}\label{eq:extractcorr1}
    \Gamma_1 = \gamma_1,
\end{equation}
The $T_\phi$ time 
can similarly be calculated by a Ramsey technique~\cite{ramsey1950molecular}; however, for this noise model,
the $T_\phi$ time will not be recovered accurately due to the additional correlated dephasing caused
by the incoherent $ZZ$ coupling. The observed decay in the Ramsey experiment is a combination
of the $T_1$, $T_\phi$, and $ZZ$ terms,
\begin{equation}\label{eq:extractcorr3}
    \Gamma_2 = \gamma_1 + \gamma_2 + 8 \gamma_3,
\end{equation}
where the prefactor 8 comes from arguments about the relative strengths of noise channels; see Section~\ref{app:numerical}.
These two standard single qubit benchmarking techniques can be combined with the ENR over the whole system provided by RCS.
This gives three equations (from the relaxation decay experiment~\eqref{eq:extractcorr1}, the Ramsey experiment~\eqref{eq:extractcorr3}, and the RCS experiment~\eqref{eq:extractcorr2}) and three
unknowns ($\gamma_1$, $\gamma_2$ and $\gamma_3$), and can be readily solved to extract the correlated noise rate,
\begin{equation}\label{eq:extractcorr4}
\gamma_3 = \frac{\Gamma_1}{4} + \frac{\Gamma_2}{4} - \frac{\lambda}{n},
\end{equation}
which directly combines the experimentally measured $\Gamma_1$, $\Gamma_2$, and $\lambda$.
We simulate the noise model of Eq.~\eqref{eq:ext_noise_model} using $n=10$ qubits
with a ring geometry with $\gamma_1=0.01$, $\gamma_2=0.02$, and $\gamma_3=0.02\alpha$, with
$\alpha \in [0.0, 0.1, 0.25, 0.5, 1.0]$. We extract the total ENR by using two-qubit Haar random circuits,
fitting the resulting uXEB curve from depths 12 to 40, except in the case of $\alpha=1.0$, where 
we fit from 12 to 22, due to the large total ENR of 0.4.
The result of performing this extraction is shown in Figure~\ref{fig:correlated_error_model}.

\subsection{Variance analysis}
\label{sec:rcsvariance}

We provide an analysis of the variance of fidelity estimators in RCS benchmarking. In particular, this can help to decide how many random circuits to implement and how many measurement samples to collect in a RCS benchmarking experiment. Our analysis is a generalization of the statistical analysis of Google's quantum supremacy experiment~\cite{Boixo2018Characterizing,arute2019quantum,rinott2020statistical} and provides a more accurate model for small scale RCS experiments.

In the following we focus on providing a theoretical model for the variance of cross entropy estimators. In experiments, once we have an estimate for the mean and variance of cross entropy estimators of different circuit depth, we can use standard least squares fitting to extract the effective noise rate $\lambda$ as well as its standard error.

Recall the definition of the unbiased linear cross entropy
\begin{equation}
    \hat{F}_{\mathrm{uXEB}}(S;C)=\frac{\frac{2^n}{M}\sum_{i=1}^M p_C(x_i) -1}{2^n\sum_{x\in\{0,1\}^n}p_C(x)^2-1}
\end{equation}
which has two sources of randomness: the random circuit $C$ and the measurement samples $S=\{x_1,\dots,x_M\}$. By the law of total variance, we have
\begin{equation}
\begin{split}
    \Var_{C,S}(\hat{F}_{\mathrm{uXEB}})&=\E_C\left[\Var_S\left(\hat{F}_{\mathrm{uXEB}}|C\right)\right]+\Var_C\left(\E_S\left[\hat{F}_{\mathrm{uXEB}}|C\right]\right).
\end{split}
\end{equation}
This can be understood as the following: the total variance of $\hat{F}_{\mathrm{uXEB}}$ is the sum of the variance across different measurement samples and the variance across different random circuits.

In the previous analysis~\cite{Boixo2018Characterizing,arute2019quantum,rinott2020statistical}, the second term is ignored and only the first term is considered. In particular, they showed that the first term scales like $\left(1+2\E F-(\E F)^2\right)/M$ where $M$ is the number of samples collected for each circuit. This can be derived by assuming that the output distribution of the noisy circuit is a linear combination of an ideal Porter-Thomas distribution with the uniform distribution. When $\E F\ll 1$ this can be well approximated by $1/M$.

For the second term, we assume that $\hat{F}_{\mathrm{uXEB}}$ is an unbiased estimator of the fidelity of the random circuit, and $\E_S\left[\hat{F}_{\mathrm{uXEB}}|C\right]=F_C$ equals to the fidelity of $C$. This implies that
\begin{equation}
    \Var_C\left(\E_S\left[\hat{F}_{\mathrm{uXEB}}|C\right]\right)=\Var_C\left(F_C\right),
\end{equation}
which is the variance of fidelity across different random circuits. This term was ignored based on a concentration of measure argument~\cite{arute2019quantum}, where they argued that for a large number of qubits the fidelity is the same across different random circuits. However, it is unclear how this term scales compared with the first term $1/M$, especially for a small number of qubits. In the following we drop the subscript $C$ as in the previous subsection, and provide an analysis for $\Var(F)$.

For simplicity we work with the same model as in section~\ref{sec:fidelitydecay}, where we consider single qubit Pauli-$X$ noise and assume a first-order approximation. Then the fidelity can be written as 
\begin{equation}
    F\approx (1-\varepsilon)^{nd}+n\varepsilon (1-\varepsilon)^{nd-1}\sum_{l=1}^{d}A_l
\end{equation}
where $A_l:=\left|\braket{\psi_l}{\psi}\right|^2$. Then the variance can be approximated by
\begin{equation}
\begin{split}
    \Var(F)&\approx \left(n\varepsilon\right)^2 (1-\varepsilon)^{2nd-2}\Var\left(\sum_{l=1}^{d}A_l\right)\\
    &\approx \lambda^2\left(\E F\right)^2\sum_{k,l=1}^d\left(\E\left[A_k A_l\right]-\E[A_k]\E[A_l]\right).
\end{split}
\end{equation}
Here we have used the fact that $\lambda\approx n\varepsilon$ and $\E F\approx e^{-\lambda d}$. Recall we have proven in section~\ref{sec:fidelitydecay} that $\E A_l\leq e^{-\Delta l}+1/2^n$ decays exponentially with $l$, and we therefore expect that each term in the above sum also decays exponentially, and the sum can be bounded by a constant,
\begin{equation}
    \Var(F)=O\left(\lambda^2\left(\E F\right)^2\right).
\end{equation}
We leave the rigorous proof of this statement to future work. Note that the same proof technique for the analysis of $\E A_l$ can be applied here; however, the variance contains a fourth moment property $\E\left[A_k A_l\right]$, which corresponds to a more complicated spin model with negative weights, making it difficult to analytically bound the partition function.

Combining both terms, we obtain a theoretical model for the total variance
\begin{equation}\label{eq:var_model}
    \Var_{C,S}(\hat{F}_{\mathrm{uXEB}})=O\left(\frac{1}{M}+\lambda^2\left(\E F\right)^2\right).
\end{equation}
In the next section, we show that this model is well supported by experiment data on IBM Quantum hardware. Additional numerical simulation results verifying our conjecture that $\sum_{k,l=1}^d\left(\E\left[A_k A_l\right]-\E[A_k]\E[A_l]\right)$ can be upper bounded by a constant are presented in Section~\ref{app:numerical}.

Finally, recall that in RCS benchmarking, $L$ random circuits are implemented for a given depth, and $M$ measurement samples are collected for each circuit. The empirical mean of the $L$ cross entropy estimators is used to estimate $\E F$ at the given depth. The variance of this estimator is then given by $1/L\cdot O\left(1/M+\lambda^2\left(\E F\right)^2\right)$. Assuming $M$ is given, then the number of circuits required to estimate $\E F$ within $\varepsilon$ additive error is $L=O\left(1/M\varepsilon^2+\lambda^2\left(\E F\right)^2/\varepsilon^2\right)$. This suggests that the strategy for choosing parameters for a RCS benchmarking experiment depends on the system size. Suppose that the effective noise rate per qubit $\lambda/n$ is fixed, then $\E F\approx e^{-\lambda d}$ decays exponentially with the number of qubits. Therefore, for a large number of qubits, the variance across different circuits $\Var(F)=O\left(\lambda^2\left(\E F\right)^2\right)$ is exponentially small, and a small $L$ with a large number of samples $M$ is sufficient for obtaining a small total variance. For example, in Google's experiment with 53 qubits only 10 random circuits are implemented. However, for a small number of qubits, the variance can be large even if we collect infinite number of samples for few random circuits, due to the second term. Therefore a large number of random circuits are needed, while the number of samples for each circuit can be smaller. For example, in our 5-qubit RCS benchmarking experiment shown below, we implement 100 random circuits and collect 8192 samples for each circuit for a given depth. These observations suggest that our generalized variance model is necessary, especially for characterizing RCS benchmarking experiments for a small number of qubits.

\begin{figure*}[t]
    \centering
   \subfloat[$n=10$]{\includegraphics[width=0.45\linewidth]{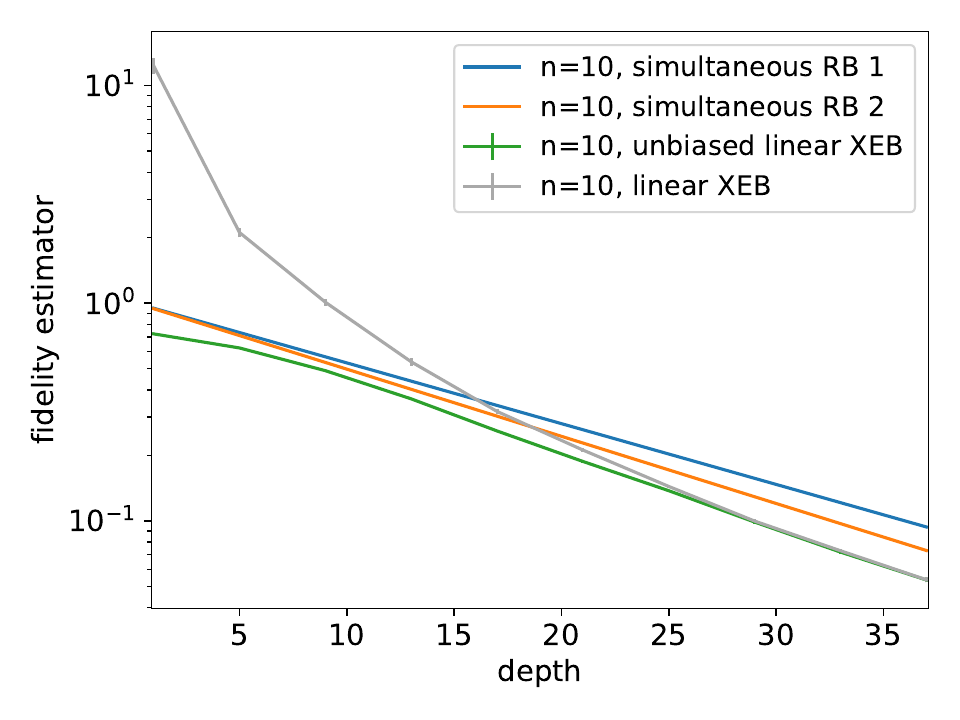}}
   \subfloat[$n=10$]{\includegraphics[width=0.45\linewidth]{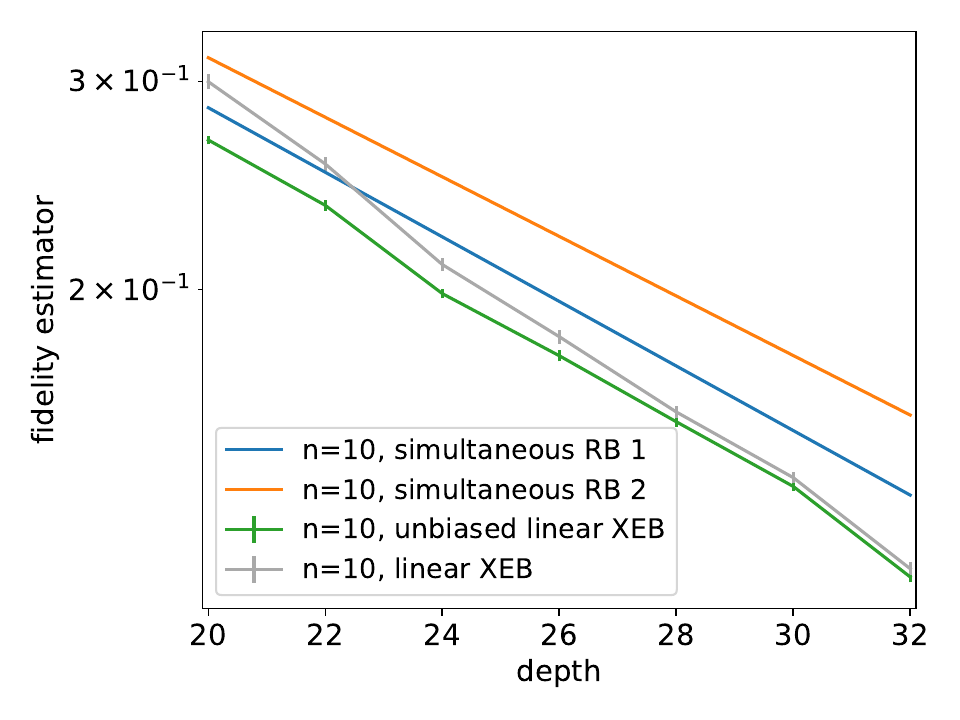}}\\
   \subfloat[$n=20$]{\includegraphics[width=0.45\linewidth]{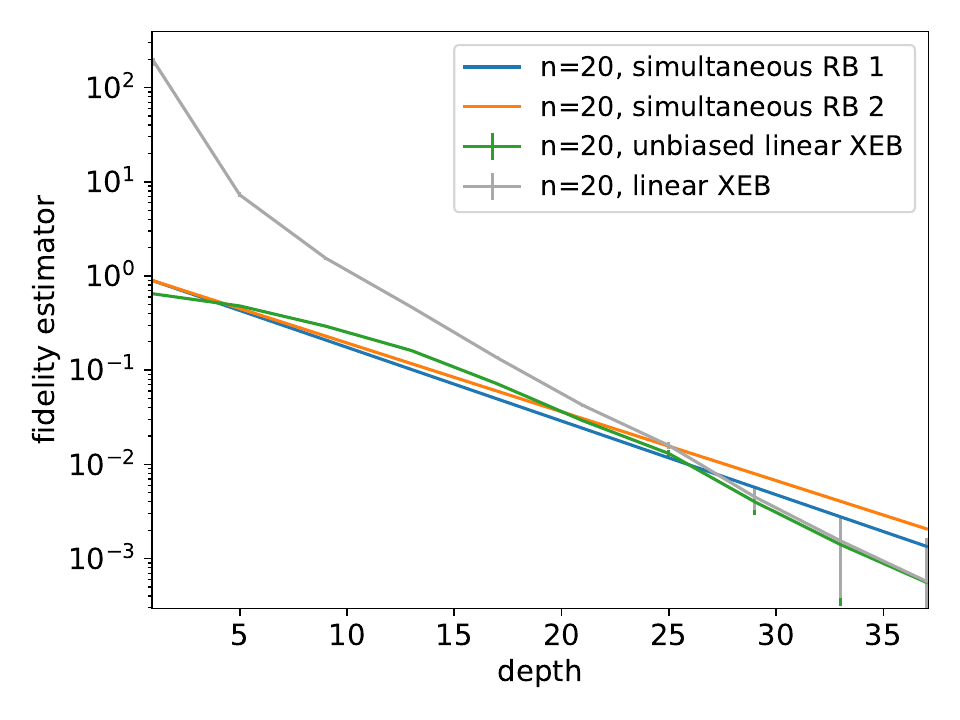}}
   \subfloat[$n=20$]{\includegraphics[width=0.45\linewidth]{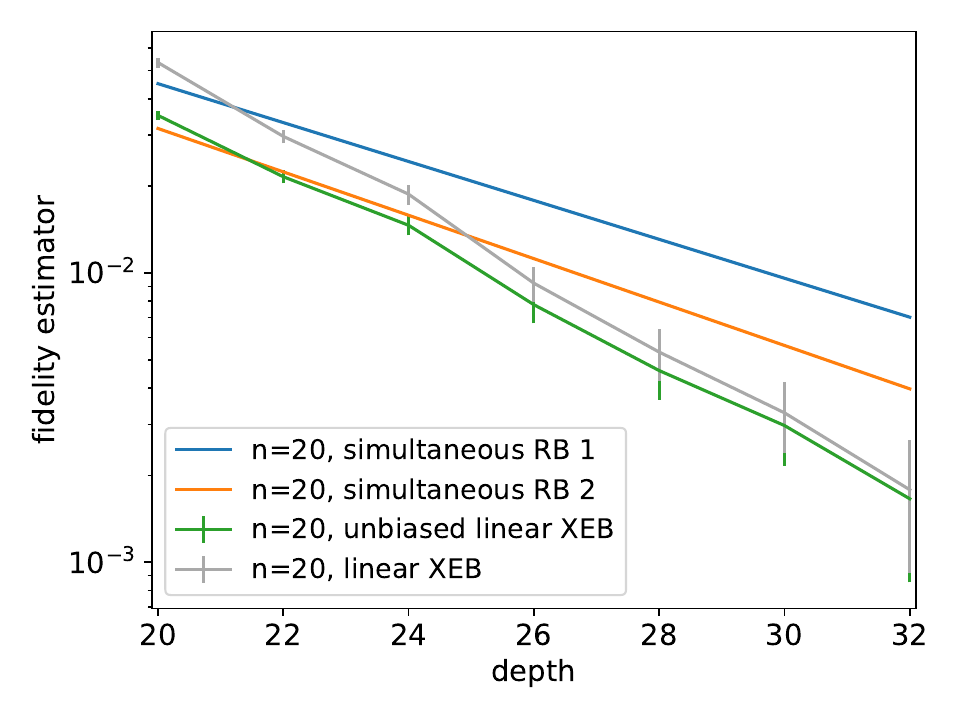}}
  \caption{Experimental implementation of RCS benchmarking on \texttt{ibmq\_mumbai} with 10 and 20 qubits. (a)(c) RCS benchmarking with long depth range (1-37) on 10 and 20 qubits, respectively. The linear cross entropy starts from a large value and converges to the unbiased one at depth 20-25. (b) RCS benchmarking with short depth range (20-32) on 10 qubits. Curve fitting results are $\lambda_{\mathrm{uXEB}}=6.9(1)\%$ and $\lambda_{\mathrm{sRB}}=6.0(2)\%$. (d) RCS benchmarking with short depth range (20-32) on 20 qubits. Curve fitting results are $\lambda_{\mathrm{uXEB}}=24.4(1)\%$ and $\lambda_{\mathrm{sRB}}=16.4(9)\%$.}
  \label{fig:rcs10and20qubit}
\end{figure*}

\begin{figure*}[t]
    \centering
   \subfloat[Cross entropy]{\includegraphics[width=0.5\linewidth]{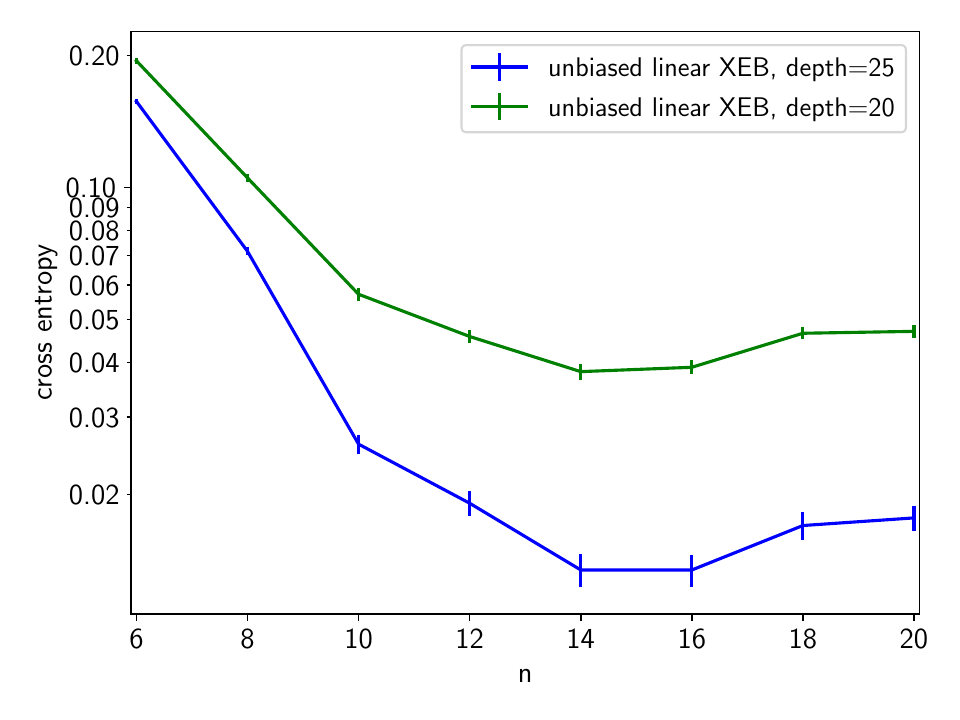}}
 \subfloat[Linear fit]{\includegraphics[width=0.5\linewidth]{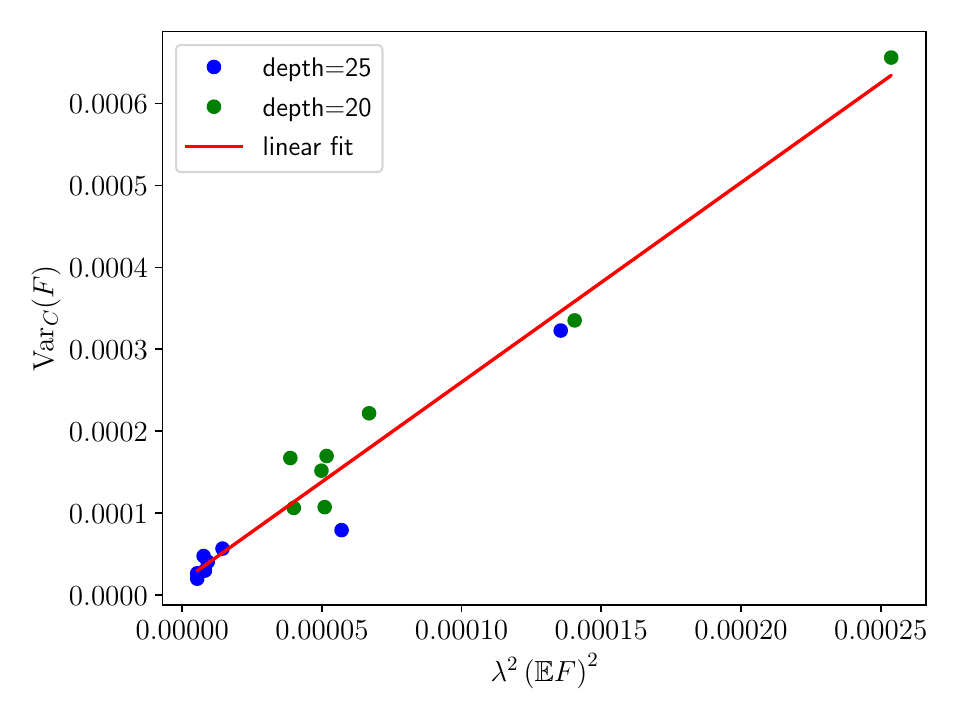}}
  \caption{Measuring cross entropy on \texttt{ibmq\_montreal} with up to 20 qubits and varying depth. (a) Cross entropy as a function of the number of qubits, which does not obey a simple exponential decay. (b) Using these data points, we can verify our variance model $\Var_C(F)=O\left(\lambda^2\left(\E F\right)^2\right)$. A linear dependence can be observed even with noisy estimates for both sides. The results of the linear fit are $\mathrm{slope}=2.4(1)$, $r=0.982$.}
  \label{fig:rcsvariance20qubit}
\end{figure*}

\section{Experiments on IBM Quantum hardware}
\label{sec:experiments}

In the following we show experiment results of RCS benchmarking on IBM's superconducting qubits. We implement RCS benchmarking protocols on Qiskit platform~\cite{Qiskit} and access the devices through cloud~\cite{ibmquantum}. All of our experiments are performed on a 1D system of qubits, where we use a 5-qubit system \texttt{ibmq\_athens} as well as 1D subsets of 27-qubit systems \texttt{ibmq\_montreal} and \texttt{ibmq\_mumbai} with up to 20 qubits. On these devices, $\mathrm{CNOT}$ is the only 2-qubit gate available, and arbitrary single-qubit gates are easy to implement, which have error rates that are roughly 2 orders of magnitude smaller than $\mathrm{CNOT}$. Therefore, in RCS benchmarking we are effectively measuring the total amount of quantum noise of a layer of $\mathrm{CNOT}$ gates, where due to crosstalk this is larger than the sum of individual gate errors measured from individual RB. Details of the gate set and architecture can be found in Section~\ref{app:experiment}.

Similar to our numerical simulation results presented in section~\ref{sec:rcstheory}, we perform three types of experiments: simultaneous RB, RCS with direct fidelity estimation, and RCS with cross entropy. These experimental results support our theoretical arguments for both fidelity decay and variance estimation.

By default, all error rates reported in this paper are represented in Pauli error rates, which is a constant factor larger than the error rates based on average fidelity that is used in RB.

\subsection{Simultaneous RB}
\label{sec:simurb}

In simultaneous RB~\cite{Gambetta2012characterization}, the main idea is to perform RB simultaneously with the gate pattern that is used in applications. That is, instead of performing RB on one pair of qubits while all other qubits are idle, different RB sequences are executed in parallel, which can capture correlated errors between different qubit pairs. A similar experiment was performed in Google's experiment~\cite{arute2019quantum} where linear cross entropy was used as the post processing method instead of standard RB.

As shown in Fig.~\ref{fig:rcsbenchmarkingapp}b, in RCS benchmarking we consider three patterns for applying gates: one layer of single qubit gates, one layer of $\mathrm{CNOT}$ starting with qubit 0, and one layer of $\mathrm{CNOT}$ starting with qubit 1. We implement simultaneous RB on each of these three patterns, and the results are shown in Fig.~3a of main text. Here each row shows one of the three patterns, and a standard Clifford RB sequence is applied to each of the green boxes. After fitting the RB curve to an exponential decay, we obtain gate error rates shown below the green boxes. The numbers underlying single qubit boxes represent the error rates of two X90 pulses ($\sqrt{X}$), which can represent the Pauli error rate of a Haar random single qubit gate. The numbers underlying two qubit boxes represent the Pauli error rate of $\mathrm{CNOT}$ gates.

\subsection{RCS benchmarking}

Next we show results of RCS benchmarking with direct fidelity estimation and cross entropy. Note that DFE is not scalable due to the exponential sample complexity, and is implemented here mainly to verify our theoretical predictions. As a reference, we implement simultaneous RB both before and after the RCS experiments, which can also be used to evaluate the error drift during the experiment period.

We implement RCS benchmarking on a 5-qubit device \texttt{ibmq\_athens} where direct fidelity estimation is tractable. Following our variance model, a large number of random circuits are implemented in order to achieve small total variance. We select 8 depth values ranging from 1 to 29. By default, we take the maximum amount of measurement samples allowed (8192) for all circuits submitted to the hardware platform. For direct fidelity estimation, we implement 30 random circuits for each depth, and estimate the fidelity of each circuit by measuring 20 Fourier coefficients according to the sampling procedure described in section~\ref{sec:fidelityestimation}. That is, 600 circuits are implemented for each depth. For cross entropy, we implement 100 random circuits for each depth.

The results of RCS benchmarking are shown in Fig.~3b of main text, where all curves are exponential decays with roughly the same decay rate. Note that for both DFE and uXEB experiments, the sample variance across different circuits is an unbiased estimator of the total variance. Therefore we use the standard error of the mean across different circuits as the error bar. The decay rate and its standard error can be computed from the data points and error bars in Fig.~3b of main text via standard least squares fitting. For the simultaneous RB estimator, we use half the difference between the two RB experiments as the standard error.

From the curve fitting results, we can see that $\lambda_{\mathrm{DFE}}$ and $\lambda_{\mathrm{uXEB}}$ agrees with each other within the standard error. This confirms our theoretical results on the exponential decay of fidelity, and also verifies the validity of cross entropy as an efficient fidelity estimator. Also, note that uXEB is much more sample efficient than DFE, where the error bars for DFE are larger even when we collect 6 times the amount of samples in uXEB. In addition, note that the simultaneous RB estimator gives a slightly larger prediction for the effective noise rate. As we have discussed before, this could come from overestimating high weight correlated errors.

Similar RCS benchmarking results for 10 and 20 qubits are presented in Fig.~\ref{fig:rcs10and20qubit}, where we do not implement DFE as it requires too many samples. Detailed parameter settings for these experiments are given in Section~\ref{app:experiment}. We perform two types of experiments: long range with depth 1-37 (Fig.~\ref{fig:rcs10and20qubit}a and c) for demonstrating the overall behavior, and short range with depth 20-32 (Fig.~\ref{fig:rcs10and20qubit}b and d) for fitting the curve. From the long range results, we can observe that the linear cross entropy starts from a large value at low depth, and converges to the unbiased version at depth around 20-25. Also, the unbiased linear cross entropy has a small ``bump" at low depth, which is more evident for larger system size. We can observe similar behavior in our numerical simulations with 20 qubits (Fig.~2 of main text). Therefore, to obtain accurate predictions for the effective noise rate, we measure cross entropy for a short depth range starting from 20, and fit the unbiased linear cross entropy using these data points. Interestingly, in both 10 and 20 qubit results, the effective noise rate given by RCS benchmarking with the unbiased linear cross entropy is larger than the simultaneous RB predictions. Assuming our depth fitting range is deep enough so that uXEB correctly estimates fidelity (which is the case in our simulations), this suggests that some error sources were captured by RCS but not by simultaneous RB. These errors should not have come from high weight correlated noise, as in this case simultaneous RB will overestimate instead of underestimate. We leave the identification of these error sources as future work. One possible explanation for these additional errors captured by RCS benchmarking is that in RCS all CNOT gates are running at the same time in each layer, which is not the case in sRB which has independent Clifford sequences. The simultaneous CNOT gates in RCS could introduce additional couplings, leading to a larger effective noise rate. In addition, note that the effective noise rate for 20 qubits is much larger than twice the effective noise rate for 10 qubits on the same device, which suggests that the errors are highly non-uniformly distributed across the device.

Finally, taking the result $\lambda_{\mathrm{uXEB}}=3.08(3)\%$ from Fig.~3b of main text as an example, we can compute the effective noise rate per qubit $\lambda_q=\lambda/n=0.62(1)\%$. This number represents the average quality of the quantum system, which we can understand as the error rate of ``half" of a 2-qubit gate plus a single qubit gate when the gates are implemented in parallel. As a comparison, we can compute from the curve fitting data of Google's experiment~\cite{arute2019quantum,zlokapa2020boundaries} that $\lambda_q=0.43(4)\%$. Recall that our Theorem~\ref{thm:fidelitydecayapp} suggests that a smaller $\lambda_q$ means that a larger number of qubits can be benchmarked together via RCS benchmarking, although the specific constants are unclear. Google's experiment can be interpreted as evidence that $\lambda_q=0.43\%$ is small enough to implement RCS benchmarking on 53 qubits, with their gate set and architecture. We expect that a smaller error rate is needed to benchmark the same number of qubits with 1D connectivity compared with 2D. This is because in some cases random circuits with 2D connectivity are known to scramble faster than 1D~\cite{harrow2018approximate}, and therefore we expect 2D circuits to have a smaller constant $K$ in our Theorem~\ref{thm:fidelitydecayapp}. 

\subsection{Cross entropy with increasing number of qubits}
Next we present results on cross entropy with increasing number of qubits. This experiment shows how cross entropy changes with more qubits added to the subset, and also gives data points that allow us to verify our variance model. We start with a 1D subset of 6 qubits on a 27-qubit device \texttt{ibmq\_montreal}, and add more qubits to the subset until it forms a 1D chain of 20 qubits.

Fig.~\ref{fig:rcsvariance20qubit}a shows the cross entropy results for depth 20 and 25. Each data point in this figure is collected by implementing 100 random circuits. The two curves correspond to different depths have consistent shapes, and the error bars for the blue curve is larger than the error bars in the green curve, which qualitatively agrees with our prediction. Interestingly, different from the results shown in Google's experiment~\cite{arute2019quantum}, here the cross entropy does not decay as a simple exponential function with the number of qubits. This could potentially suggest that adding more qubits to the subset can dramatically change the error pattern, indicating complicated error correlations among the qubits. In particular, as more qubits are added to the subset, the pattern of edge effects and dangling qubits also changes, which can affect the noise of the subset of qubits being benchmarked. For example, in Fig.~\ref{fig:rcsvariance20qubit}a the two system sizes with lowest cross entropy ($n=14,16$) correspond to having additional dangling qubits (qubit 20 and 13 as shown in Fig.~\ref{fig:ibmq_devices}c) at the edge.

The data points collected in Fig.~\ref{fig:rcsvariance20qubit}a can be used to verify our variance model. Recall that our variance model $ \Var_{C,S}(\hat{F}_{\mathrm{uXEB}})=\E_C\left[\Var_S\left(\hat{F}_{\mathrm{uXEB}}|C\right)\right]+\Var_C\left(\E_S\left[\hat{F}_{\mathrm{uXEB}}|C\right]\right)=O\left(1/M+\lambda^2\left(\E F\right)^2\right)$ (Eq.~\eqref{eq:var_model}) states that the total variance equals the variance across samples and the variance across circuits, and we would like to verify our model for the second term $\Var_C(F)=O\left(\lambda^2\left(\E F\right)^2\right)$. For this purpose, we use a rough estimate for both sides from experiment data, and test the linear dependence. For the right hand side, we use $\lambda^2\left(\E F\right)^2\approx (\E F\log \E F)^2/d^2$ and substitute $\E F$ with the empirical mean of $\hat{F}_{\mathrm{uXEB}}$. For the left hand side, we use the sample variance of $\hat{F}_{\mathrm{uXEB}}$ across different circuits minus the average sample variance across the measurement outcome probabilities for each circuit. This value is an unbiased estimator for $\Var_C(F)$. Even though the estimates for both sides are very noisy, we can still observe a linear dependence shown in Fig.~\ref{fig:rcsvariance20qubit}b. This verifies our model $\Var_C(F)=O\left(\lambda^2\left(\E F\right)^2\right)$, which suggests that $\Var_C(F)$ does not directly depend on system size or depth, while the constant depends on circuit architecture and gate set.

\section{Additional discussions}
\label{sec:discussion}

In this paper we develop RCS benchmarking, which allows sample efficient benchmarking of the total amount of quantum noise of a many-body system using only shallow circuits and local randomness. Unlike RB and variants where the fidelity decay comes from the group structure, RCS benchmarking uses the scrambling properties of random quantum circuits, which can be directly implemented with the native gate set on any hardware platform without additional compilation. While this paper considers qubits with 1D connectivity, we expect our results to be naturally generalized to more general connectivity graphs and higher local Hilbert space dimensions.

An interesting future direction is to design variants of RCS benchmarking such that more information about the noise channel can be extracted. For example, recent work~\cite{kim2021hardwareefficient} showed that improved post-processing techniques can be used to classify incoherent vs coherent noise. In addition, we can consider RCS benchmarking with a growing number of qubits (such as in Fig.~\ref{fig:rcsvariance20qubit}) and extract noise correlations from the data. Also, similar to simultaneous RB, we can also consider simultaneous RCS where different subsets of the qubits are benchmarked at the same time. This allows more flexible subset choices than simultaneous RB, and also provides a possible way to avoid the computation barrier. Finally, it is interesting to consider if the learning algorithms for Pauli noise channels~\cite{Flammia2020efficient,Harper2020efficient,harper2020fast,flammia2021pauli} can be extended to general gate sets, using ideas from RCS benchmarking.

\section{Detailed proofs}
\label{app:proofdetail}

\subsection{Mapping random quantum circuits to a classical spin model}
We first describe the basic idea in our analysis of RCS benchmarking which uses the technique that maps random quantum circuits to a classical spin model. This technique has been used to study several other properties of random quantum circuits, see e.g.~\cite{Nahum2018operator,Zhou2019emergent,hunterjones2019unitary,bao2020theory,barak2020spoofing,dalzell2020random}.

In general our analysis of RCS benchmarking concerns the following object
\begin{equation}\label{eq:secondmomentproperty}
    \E_{C\sim\mathrm{RQC}(n,d)}\left[\text{A second moment property of $C$}\right].
\end{equation}
Here the second moment property can be understood as the following: for each individual gate $U_i$ in $C$, this property is a linear function in the second moment of $U_i$, which is $U_i^{\otimes 2}\otimes U_i^{*\otimes 2}$, where $U_i^*$ is the complex conjugate of $U_i$. By linearity of expectation, we have
\begin{equation}
    \begin{split}
        \E_{C=U_m\cdots U_1\sim\mathrm{RQC}(n,d)}\left[\text{A second moment property of $C$}\right]&=\E_{U_m\sim\mathbb{U}(4)}\cdots \E_{U_1\sim\mathbb{U}(4)}\left[\text{A second moment property of $C$}\right]\\
    &=\text{A linear function of} \E_{U_i\sim\mathbb{U}(4)}\left[U_i^{\otimes 2}\otimes U_i^{*\otimes 2}\right].
    \end{split}
\end{equation}
After taking the expectation, the above equation can be viewed as a tensor network, where the basic element in the above object is the following local tensor (re-ordered for convenience):
\begin{equation}\label{eq:localtensor}
    \E_{U\sim\mathbb{U}(4)}\left[U\otimes U^*\otimes U\otimes U^*\right]=\E_{U\sim\mathbb{U}(4)}\Bigg[\vcenter{\hbox{\includegraphics[height=3cm]{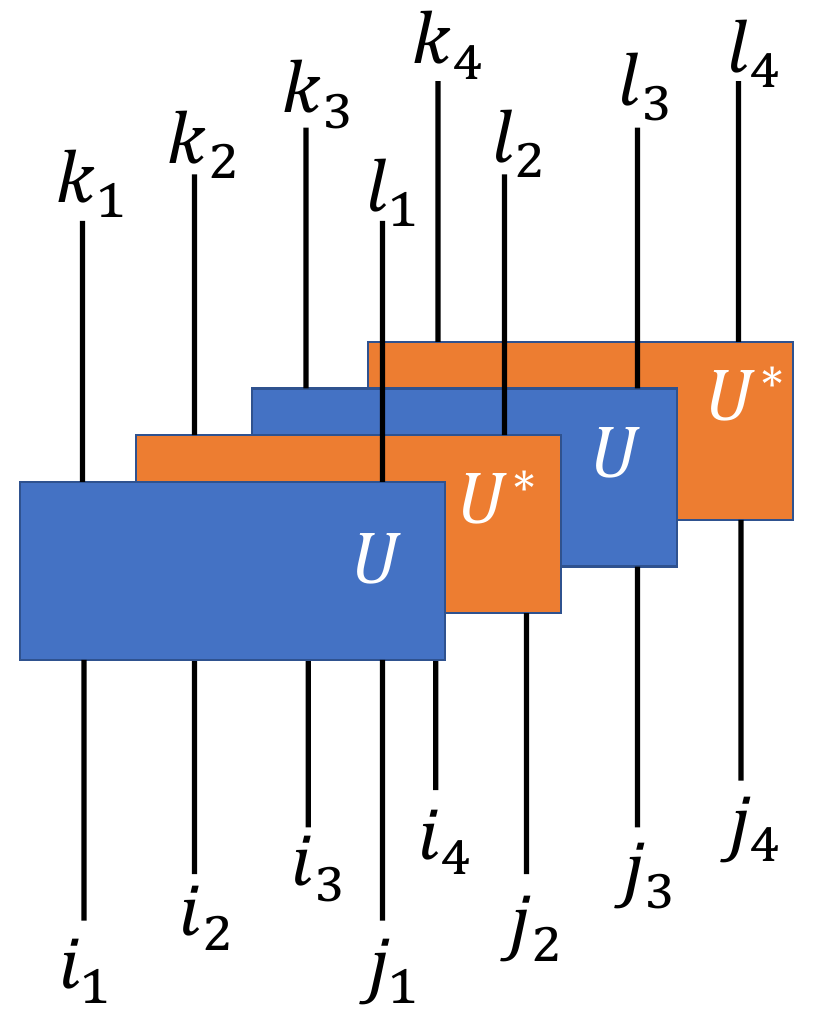}}}\Bigg]=\sum_{\tau,\sigma\in\{-1,1\}}w(\tau,\sigma)\cdot \vcenter{\hbox{\includegraphics[height=2cm]{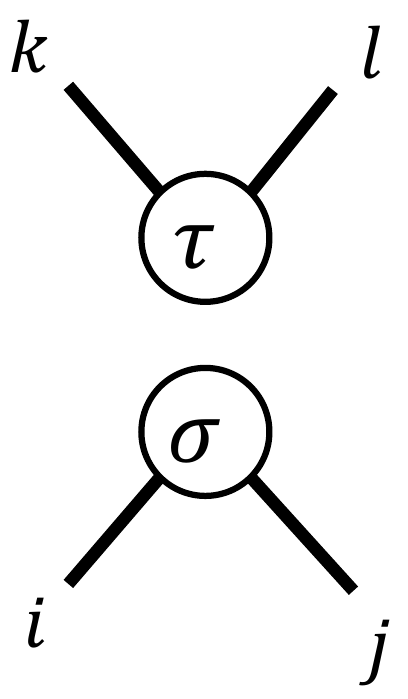}}}
\end{equation}
where $\tau,\sigma$ are classical variables taking values in $\{-1,1\}$, and the white tensors depend on their values. $w(\tau,\sigma)$ are additional weights given by
\begin{nalign}
    w(\tau,\sigma)&=\frac{1}{12}\delta_{\tau\sigma}-\frac{1}{60}\\
    &=\begin{cases}\frac{1}{15},&\tau=\sigma,\\-\frac{1}{60},&\tau\neq\sigma.\end{cases}
\end{nalign}
The white $\sigma/\tau$ tensors have the following form:
\begin{itemize}
    \item $\sigma/\tau=+1$: $\vcenter{\hbox{\includegraphics[height=1cm]{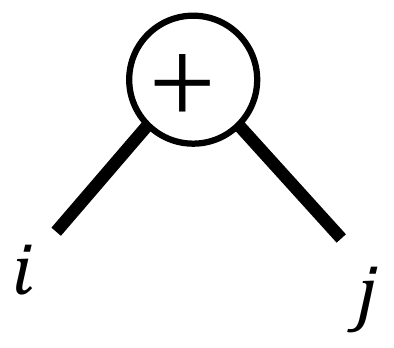}}}=\delta_{i_1 i_2}\cdot \delta_{i_3 i_4}\cdot \delta_{j_1 j_2}\cdot \delta_{j_3 j_4}$,
    \item $\sigma/\tau=-1$: $\vcenter{\hbox{\includegraphics[height=1cm]{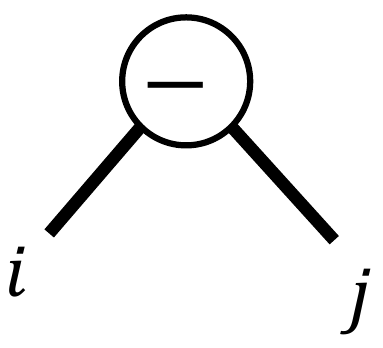}}}=\delta_{i_1 i_4}\cdot \delta_{i_2 i_3}\cdot \delta_{j_1 j_4}\cdot \delta_{j_2 j_3}$.
\end{itemize}

After replacing all two-qubit gates with the above tensor, we have mapped any second moment property of random quantum circuits to a classical spin model with the $\tau,\sigma$ variables. The original quantity in Eq.~\eqref{eq:secondmomentproperty} corresponds to the partition function of this spin model. For each specific assignment of $\tau,\sigma$, Eq.~\eqref{eq:localtensor} is a separable tensor, and the value corresponding to the assignment is easy to calculate. However, this spin model has negative weights, making it hard to directly analyze in this form.

\subsection{Proof of Theorem~\ref{thm:reducetopauli}}
Let $C\sim\mathrm{RQC}(n,d)$ be a random circuit, $\ket{\psi}=C\ket{0^n}$ be the ideal output state, and $\rho$ be the output state of the noisy circuit. It is easy to see that the fidelity
\begin{equation}
    F=\expval{\rho}{\psi}=\Tr\left[\rho\cdot \ketbra{\psi}\right]
\end{equation}
is a second moment property, because trace is a linear operation and each two-qubit gate $U_i$ appears in the fidelity as the form $U_i^{\otimes 2}\otimes U_i^{*\otimes 2}$.

Let $\mc N(\rho)=\sum_{\alpha,\beta\in\{0,1,2,3\}^n}\chi_{\alpha\beta}\sigma_{\alpha}\rho \sigma_{\beta}$ be an arbitrary noise channel which acts on all qubits after each layer of gates. Due to the linearity of quantum channels, we only need to prove that the off-diagonal terms $\sigma_\alpha \rho \sigma_\beta$ ($\alpha\neq\beta$) have 0 contribution in the calculation of $\E F$, which implies that the average fidelity with $\mc N$ is equal to the average fidelity of $\mc N^{\mathrm{diag}}$, which proves Theorem 1.

Let $\sigma_\alpha\neq\sigma_\beta\in\{I,X,Y,Z\}^{\otimes n}$ be two distinct Pauli operators. They have to differ by at least one location, and we denote the single qubit Pauli operator as $P_1\neq P_2\in\{I,X,Y,Z\}$. Without loss of generality, we assume $P_1$ and $P_2$ act on the $l_1,l_2$ index of a two-qubit gate and the $i_1,i_2$ index of another two-qubit gate. We perform the mapping described in the previous subsection which maps the average fidelity to the partition function of a classical spin model. The relevant local tensor for $P_1$ and $P_2$ is given by 
\begin{equation}
    \vcenter{\hbox{\includegraphics[height=3cm]{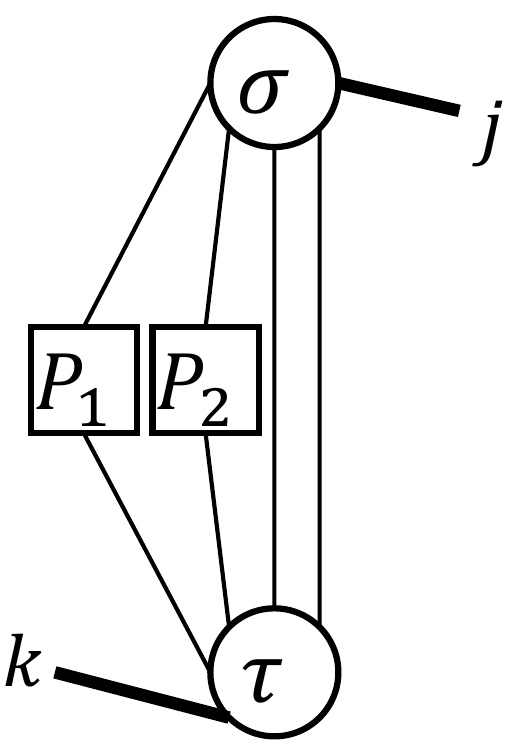}}}
\end{equation}
and it is easy to see that this tensor always equals to 0 no matter what value $\sigma$ and $\tau$ takes, which makes the average fidelity equal to 0 under the off-diagonal term. 

\subsection{Analysis of average fidelity}
Next we analyze $\E F$ under Pauli noise channels. As discussed above, $\E F$ can be directly mapped to the partition function of a classical spin model with $2m$ spins, where $m$ is the number of two-qubit gates. However, as this spin model has negative weights due to the noise channel, it appears difficult to directly analyze its partition function.

Our first idea is to separate the effect of noise from random quantum circuits via an expansion in the noise rate. As described in section~\ref{sec:fidelitydecay}, we first focus on i.i.d. single qubit Pauli-$X$ noise and later address more general cases. Following section~\ref{sec:fidelitydecay} we write the average fidelity as
\begin{equation}\label{eq:appfidelityexpansion}
\begin{split}
    \E F&=(1-\varepsilon)^{nd}+\sum_{i=1}^{nd}\varepsilon (1-\varepsilon)^{nd-1}\E\left|\braket{\psi_i}{\psi}\right|^2+\sum_{i=1}^{nd-1}\sum_{j=i+1}^{nd}\varepsilon^2(1-\varepsilon)^{nd-2}\E\left|\braket{\psi_{ij}}{\psi}\right|^2+\cdots\\
    &:=F_0+\E F_1+\sum_{k\geq 2}\E F_k.
\end{split}
\end{equation}
Here $\ket{\psi_i}$ denotes the ideal state with an $X$ error at location $i$. We can understand $\E F_k$ as the contribution to average fidelity with $k$ errors in the circuit, which has an $\varepsilon^k$ factor. Therefore we expect the higher order contributions to be small compared with $\E F_1$, when $\varepsilon$ is sufficiently small.

\begin{figure}[t]
    \centering
    \includegraphics[width=10cm]{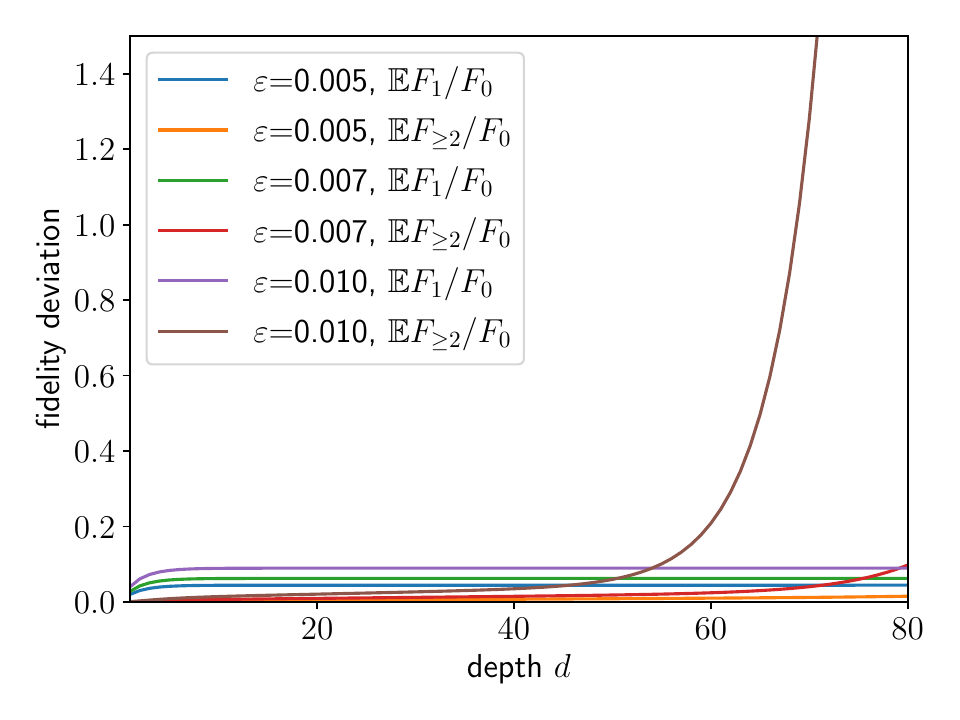}
    \caption{Simulation of the first and higher order terms in average fidelity. Here we consider $n=20$ qubits on a 1D ring and i.i.d. Pauli-$X$ noise with effective noise rate per qubit $\varepsilon$.}
    \label{fig:fidelityhigherorder}
\end{figure}

To verify this intuition, we present numerical simulation results shown in Fig.~\ref{fig:fidelityhigherorder}. Here the simulation is performed by simulating the partition function of the corresponding spin model and therefore has no error bar. For $n=20$ qubits with $\varepsilon\leq 0.01$ Pauli-$X$ noise per qubit, we plot the ratio $\E F_1/F_0$ as well as higher order terms $\E F_{\geq 2}/F_0$, where $\E F_{\geq 2}=\sum_{k\geq 2}\E F_k$. From the simulation results we can observe the following:
\begin{itemize}
    \item First, as we will rigorously prove next, $\E F_1/F_0$ converges to a constant that is proportional to $\varepsilon$ as depth increases.
    \item Second, note that $\E F_{\geq 2}/F_0$ is much smaller compared with $\E F_1/F_0$ for shallow depth circuits. This verifies our intuition to throw away higher order terms in our analysis of RCS benchmarking experiments, as in experiments we are limited to implementing shallow depth circuits.
    \item Third, note that $\E F_{\geq 2}/F_0\to\infty$ as $d\to\infty$.
\end{itemize}
Next we provide a simple argument to explain the third point. As $d\to\infty$, we have $F_0=(1-\varepsilon)^{nd}\to 0$, and $\E F_1\leq O(F_0)\to 0$. However, we know that $\E F\to\frac{1}{2^n}$ as the state converges to the maximally mixed state on average. Taking the limit $d\to\infty$ on both sides of Eq.~\eqref{eq:appfidelityexpansion}, we have $\E F_{\geq 2}\to\frac{1}{2^n}$, and therefore $\E F_{\geq 2}/F_0\to\infty$. This observation suggests a potential improvement to RCS benchmarking. In particular, we can consider other fitting schemes, such as fitting an exponential decay with $\E F-\frac{1}{2^n}$ instead of $\E F$, which might extend the depth range for which RCS benchmarking gives accurate results.

Next we proceed by analyzing the first order term $\E F_1$ and prove Theorem~\ref{thm:fidelitydecayapp}. Following the proof sketch in section~\ref{sec:fidelitydecay}, we focus on proving
\begin{equation}\label{eq:appPauliobservable}
    \E\left|\braket{\psi_l}{\psi}\right|^2=\E_{C\sim\mathrm{RQC}(n,l)}\left|\bra{0^n}C^\dag \sigma_p C\ket{0^n}\right|^2\leq e^{-\Delta l}+\frac{1}{2^n},
\end{equation}
where $\sigma_p\in\{I,X,Y,Z\}^{\otimes n}$ is a Pauli operator. We will formally prove Eq.~\eqref{eq:appPauliobservable} for arbitrary 3-local Pauli errors, and show numerical evidence that shows this holds in general.

It is easy to see that Eq.~\eqref{eq:appPauliobservable} implies that $\E F_1/F_0=O(\lambda)$ where $\lambda$ is the effective noise rate for the global noise channel. For example, for i.i.d. single qubit Pauli-$X$ noise channel where $\lambda\approx n\varepsilon$, Eq.~\eqref{eq:appPauliobservable} with $\sigma_p=X\otimes I^{\otimes n-1}$ implies that
\begin{equation}
    \begin{split}
        \E F_1/F_0&=\frac{n\varepsilon}{1-\varepsilon}\sum_{l=1}^{d}\E\left|\braket{\psi_l}{\psi}\right|^2\\
    &\leq \frac{n\varepsilon}{1-\varepsilon}\sum_{l=1}^{d}\left(e^{-\Delta l}+\frac{1}{2^n}\right)\\
    &=O(n\varepsilon)
    \end{split}
\end{equation}
when $d\leq 2^n$. We can similarly prove $\E F_1/F_0=O(\lambda)$ for general noise channels, as long as Eq.~\eqref{eq:appPauliobservable} still holds, where we replace the single qubit Pauli-$X$ with the Pauli operator in the noise channel.

To complete our analysis on the first order term $\E F_1$, it remains to prove the upper bound on $\E_{C\sim\mathrm{RQC}(n,l)}\left|\bra{0^n}C^\dag \sigma_p C\ket{0^n}\right|^2$. First note that this is a second moment property of ideal depth-$l$ random quantum circuits, as
\begin{equation}\label{eq:apppauliobs}
    \E_{C\sim\mathrm{RQC}(n,l)}\left|\bra{0^n}C^\dag \sigma_p C\ket{0^n}\right|^2=\E_{C\sim \mathrm{RQC}(n,l)}\Tr\left[\sigma_p C\ketbra{0^n}{0^n}C^\dag\sigma_p\cdot  C\ketbra{0^n}{0^n}C^\dag\right].
\end{equation}
We therefore map this quantity to the partition function of a classical spin model, where the bulk of the spin model is a hexagonal lattice consists of local tensors given in Eq.~\eqref{eq:localtensor}. 
    
This was first developed by~\cite{Nahum2018operator}, where it was noted that this spin model can be further simplified by first summing over the $\tau$ variables, and as a result eliminates the negative weights in $w(\tau,\sigma)$. After performing this summation over $\tau$ variables, the spin model corresponds to a triangular lattice of the $\sigma$ variables shown in Fig.~\ref{fig:rcsspinmodel}, and the bulk of the weights corresponding to the three-body interaction are given in Fig.~\ref{fig:spinmodeldetails}. We refer to Refs.~\cite{Nahum2018operator,Zhou2019emergent,hunterjones2019unitary,bao2020theory,barak2020spoofing,dalzell2020random} for the detailed calculations.

Next we derive the boundary conditions in order to fully specify the spin model for Eq.~\eqref{eq:apppauliobs}. The bottom boundary corresponds to open boundary conditions for all $\sigma$ variables, as the input is fixed to be 0. At the top boundary, for each two-qubit gate there is a Pauli operator $P_1\otimes P_2\in\{I,X,Y,Z\}^{\otimes 2}$ acting on it. Summing over the $\tau$ variables for each two-qubit gate at the top, we derive the following boundary condition:

\begin{equation}
    \vcenter{\hbox{\includegraphics[height=3cm]{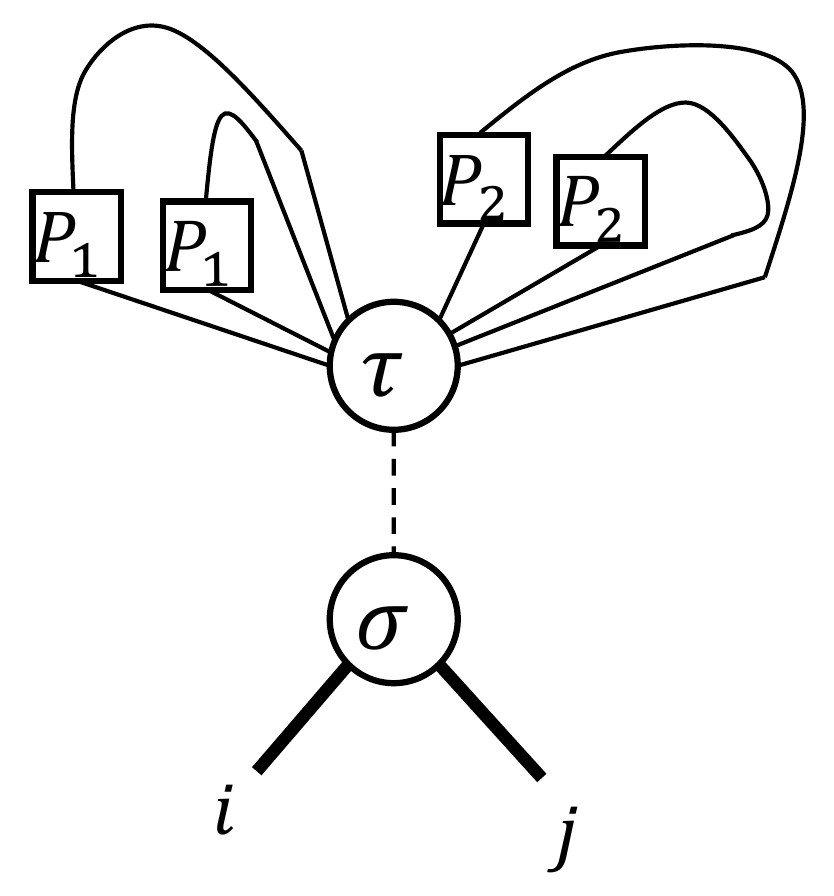}}}=
    \begin{cases}
    \vcenter{\hbox{\includegraphics[height=1cm]{figs/fig_minus_spin.pdf}}},& P_1=P_2=I,\\
    \frac{4}{15}\cdot\vcenter{\hbox{\includegraphics[height=1cm]{figs/fig_plus_spin.pdf}}},& P_1,P_2\neq I,I,\,\,\,\,\sigma=\oplus,\\
    -\frac{1}{15}\cdot\vcenter{\hbox{\includegraphics[height=1cm]{figs/fig_minus_spin.pdf}}},& P_1,P_2\neq I,I,\,\,\,\,\sigma=\ominus.\\
    \end{cases}
\end{equation}
Here, the dashed line denotes the $w(\tau,\sigma)$ interaction. This can be summarized as follows:
\begin{enumerate}
    \item When there is no error, the two-qubit gate corresponds to a fixed $\ominus$ spin.
    \item When there is one or two Pauli errors, the $\sigma$ spin has an additional weight $\frac{4}{15}$ when it equals to $\oplus$, and $-\frac{1}{15}$ when it equals to $\ominus$.
\end{enumerate}

Now we are ready to calculate the expectation value in Eq.~\eqref{eq:apppauliobs}. Define $\mc Z(n,l;b)$ as the partition function of the spin model which corresponds to $n$-qubit, depth-$l$ random quantum circuits, where the top boundary condition is given by $b\in\{\oplus,\ominus\}^{n/2}$. We start with single qubit Pauli noise, where $\sigma_p=X\otimes I^{\otimes n-1}$. In this case, there is one spin at the top which has the second boundary condition, while all other spins are fixed to be $\ominus$. Therefore 
\begin{equation}
    \begin{split}
        \E_{C\sim\mathrm{RQC}(n,l)}\left|\bra{0^n}C^\dag X C\ket{0^n}\right|^2&=\frac{4}{15}\mc Z(n,l;\cdots \ominus\oplus\ominus\cdots)-\frac{1}{15}\mc Z(n,l;\cdots \ominus\ominus\ominus\cdots)\\
        &=\frac{4}{15}\mc Z(n,l;\cdots \ominus\oplus\ominus\cdots)-\frac{1}{15},
    \end{split}
\end{equation}
where in the second line we use the observation that the partition function equals to 1 when all spins at the top boundary are equal to $\ominus$. This proves the equation in Fig.~\ref{fig:rcsspinmodel}.

Following the argument in section~\ref{sec:fidelitydecay} and demonstrated in Fig.~\ref{fig:spinmodeldetails}, we can divide the partition function into two parts,
\begin{equation}
    \mc Z(n,l;\cdots \ominus\oplus\ominus\cdots)=\mc Z_1(n,l;\cdots \ominus\oplus\ominus\cdots)+\mc Z_2(n,l;\cdots \ominus\oplus\ominus\cdots)
\end{equation}
where $\mc Z_i$ denotes the sum of weights of domain wall configuration of type $i$ shown in Fig.~\ref{fig:spinmodeldetails}b. First, note that the analysis of $\mc Z_1$ is simple. Each of the two domain walls has length $l-1$ and weight $\left(\frac{2}{5}\right)^{l-1}$. For each domain wall, the number of possible configurations is at most $2^{l-1}$. Thus, the contribution of two domain walls is at most
\begin{equation}
    \mc Z_1(n,l;\cdots \ominus\oplus\ominus\cdots)\leq \left(\frac{4}{5}\right)^{2(l-1)}.
\end{equation}
In particular, this implies that $\mc Z_1(n,\infty;\cdots \ominus\oplus\ominus\cdots)=0$. Next, we use the following useful fact: when $l\to\infty$, depth-$l$ random quantum circuits converge to an exact unitary 2-design~\cite{Brandao2016}, which implies that
\begin{equation}
    \begin{split}
        \lim_{l\to\infty}\E_{C\sim\mathrm{RQC}(n,l)}\left|\bra{0^n}C^\dag \sigma_p C\ket{0^n}\right|^2=\E_{C\sim\mathbb{U}(2^n)}\left|\bra{0^n}C^\dag \sigma_p C\ket{0^n}\right|^2=\frac{1}{2^n+1},
    \end{split}
\end{equation}
for any non-zero $\sigma_p\in\{I,X,Y,Z\}^{\otimes n}$. See e.g.~\cite{Harrow2009random} for a proof of the second equality. Therefore,
\begin{equation}
    \begin{split}
        \frac{1}{2^n+1}&=\lim_{l\to\infty}\E_{C\sim\mathrm{RQC}(n,l)}\left|\bra{0^n}C^\dag X C\ket{0^n}\right|^2\\
        &=\frac{4}{15}\left(\mc Z_1(n,\infty;\cdots \ominus\oplus\ominus\cdots)+\mc Z_2(n,\infty;\cdots \ominus\oplus\ominus\cdots)\right)-\frac{1}{15}\\
        &=\frac{4}{15}\mc Z_2(n,\infty;\cdots \ominus\oplus\ominus\cdots)-\frac{1}{15}.
    \end{split}
\end{equation}
Combining the above facts, we have
\begin{equation}
\begin{split}
    \E_{C\sim\mathrm{RQC}(n,l)}\left|\bra{0^n}C^\dag X C\ket{0^n}\right|^2&=\frac{4}{15}\left(\mc Z_1(n,l;\cdots \ominus\oplus\ominus\cdots)+\mc Z_2(n,l;\cdots \ominus\oplus\ominus\cdots)\right)-\frac{1}{15}\\
    &\leq \frac{4}{15}\left(\frac{4}{5}\right)^{2(l-1)}+\frac{4}{15}\mc Z_2(n,l;\cdots \ominus\oplus\ominus\cdots)-\frac{1}{15}\\
    &\leq \frac{4}{15}\left(\frac{4}{5}\right)^{2(l-1)}+\frac{4}{15}\mc Z_2(n,\infty;\cdots \ominus\oplus\ominus\cdots)-\frac{1}{15}\\
    &=\frac{4}{15}\left(\frac{4}{5}\right)^{2(l-1)}+\frac{1}{2^n+1}.
\end{split}
\end{equation}
Here, the third line follows from the simple observation that $\mc Z_2(n,l;\cdots \ominus\oplus\ominus\cdots)$ is monotonically non-decreasing with respect to $l$. This proves Eq.~\eqref{eq:appPauliobservable} for arbitrary 1-local errors.

Next we extend the above argument to arbitrary 3-local errors, of the form $X\otimes X\otimes X\otimes I_{\mathrm{else}}$, $X\otimes I\otimes Z\otimes I_{\mathrm{else}}$, etc. The 2-local case either reduces to 1-local (when the two errors act on the same two-qubit gate) or 3-local (when the two errors act on different two-qubit gates). For arbitrary 3-local errors, the top boundary has two neighboring spins with the second boundary condition, and all other spins are fixed to be $\ominus$. After calculating the boundary conditions, we have
\begin{equation}
    \E_{C\sim\mathrm{RQC}(n,l)}\left|\bra{0^n}C^\dag \sigma_{\text{3-local}} C\ket{0^n}\right|^2=\frac{16}{225} \mc Z(n,l;\cdots\ominus\oplus\oplus\ominus\cdots)-\frac{8}{225}\mc Z(n,l;\cdots\ominus\oplus\ominus\cdots)+\frac{1}{225}.
\end{equation}
To relate the first two terms, we derive the following recursive formula for the partition function:
\begin{equation}
    \mc Z(n,l;\cdots\ominus\oplus\ominus\cdots)=\frac{4}{25}\mc Z(n,l-1;\cdots\ominus\oplus\oplus\ominus\cdots)+\frac{8}{25}\mc Z(n,l-1;\cdots\ominus\oplus\ominus\cdots)+\frac{4}{25}.
\end{equation}
Combining both equations above, we have
\begin{equation}
\begin{split}
    &\E_{C\sim\mathrm{RQC}(n,l)}\left|\bra{0^n}C^\dag \sigma_{\text{3-local}} C\ket{0^n}\right|^2\\
    &=\frac{4}{9} \mc Z(n,l+1;\cdots\ominus\oplus\ominus\cdots)-\frac{8}{45}\mc Z(n,l;\cdots\ominus\oplus\ominus\cdots)-\frac{1}{15}\\
    &\leq O\left(e^{-\Delta l}\right) + \frac{4}{9} \mc Z_2(n,l+1;\cdots\ominus\oplus\ominus\cdots)-\frac{8}{45}\mc Z_2(n,l;\cdots\ominus\oplus\ominus\cdots)-\frac{1}{15}.
\end{split}
\end{equation}
Here in the inequality we have combined the exponentially decaying $\mc Z_1$ terms. Now, notice that the difference between $\mc Z_2(n,l+1;\cdots\ominus\oplus\ominus\cdots)$ and $\mc Z_2(n,l;\cdots\ominus\oplus\ominus\cdots)$ can be bounded: it corresponds to domain walls with depth $l$, and therefore has weight that is exponentially small in $l$, which gives
\begin{nalign}
    &\E_{C\sim\mathrm{RQC}(n,l)}\left|\bra{0^n}C^\dag \sigma_{\text{3-local}} C\ket{0^n}\right|^2\\
    &\leq O\left(e^{-\Delta l}\right) + \frac{4}{9} \mc Z_2(n,l+1;\cdots\ominus\oplus\ominus\cdots)-\frac{8}{45}\mc Z_2(n,l;\cdots\ominus\oplus\ominus\cdots)-\frac{1}{15}\\
    &\leq O\left(e^{-\Delta l}\right) + \frac{4}{9} \left(\mc Z_2(n,l;\cdots\ominus\oplus\ominus\cdots)+O\left(e^{-\Delta l}\right)\right)-\frac{8}{45}\mc Z_2(n,l;\cdots\ominus\oplus\ominus\cdots)-\frac{1}{15}\\
    &=O(e^{-\Delta l}) + \frac{4}{15}\mc Z_2(n,l;\cdots\ominus\oplus\ominus\cdots)-\frac{1}{15}\\
    &\leq O(e^{-\Delta l}) + \frac{4}{15}\mc Z_2(n,\infty;\cdots\ominus\oplus\ominus\cdots)-\frac{1}{15}\\
    &= O(e^{-\Delta l}) + \frac{1}{2^n+1}.
\end{nalign}
This proves Eq.~\eqref{eq:appPauliobservable} for arbitrary 3-local errors, and concludes the proof of Theorem~\ref{thm:fidelitydecayapp}.

\begin{figure}[t]
    \centering
    \includegraphics[width=10cm]{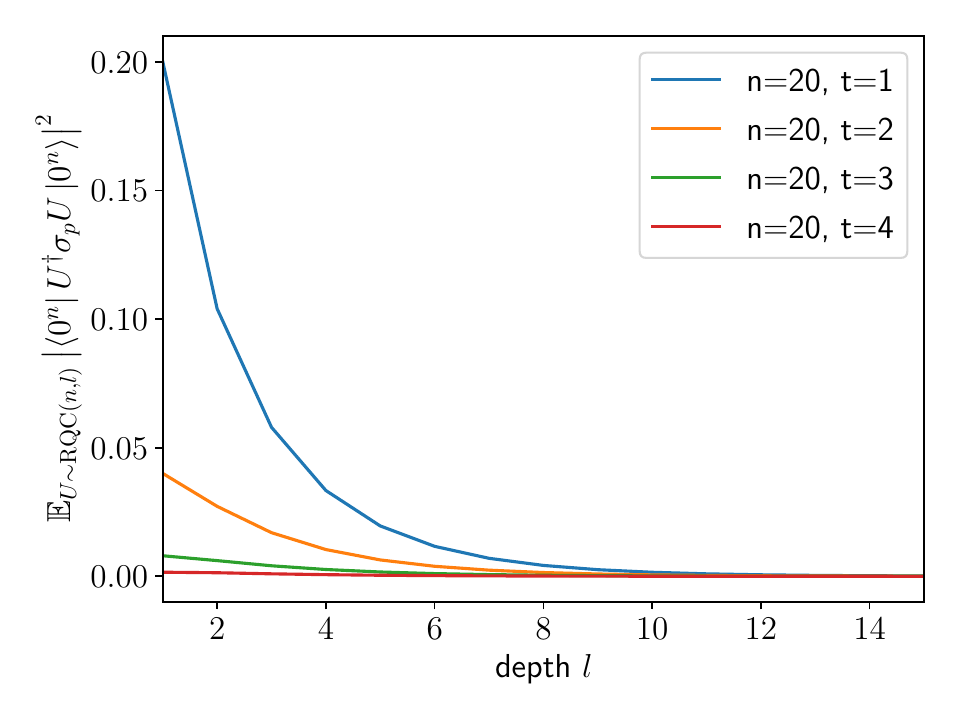}
    \caption{Simulation of the second moment of Pauli observable for depth-$l$ random quantum circuits on $n=20$ qubits, which corresponds to a classical spin model with depth $l$ and width 10. Here $t$ denotes the number of consecutive spins that has the second boundary condition.}
    \label{fig:pauliobservabledecay}
\end{figure}

Finally we present numerical simulation results which verify our proof and also show the correctness of Eq.~\eqref{eq:appPauliobservable} for higher weight errors. In Fig.~\ref{fig:pauliobservabledecay} we show simulation results for depth-$l$ random quantum circuits on 20 qubits. Here we simulate the spin model of width 10, with $t$ consecutive spins at the top which have the second boundary condition. $t$ corresponds to the locality of the Pauli error, as for a $k$-local Pauli operator we have $\frac{k}{2}\leq t\leq \frac{k}{2}+1$. From Fig.~\ref{fig:pauliobservabledecay} we can observe exponential decays as $l$ increases, which eventually converge to $\frac{1}{2^n+1}$. Interestingly, the curve also decreases as $t$ increases. We have already proven the exponential decay for $t=1$ (blue curve). The simulation results suggest that Eq.~\eqref{eq:appPauliobservable} can actually be strengthened for higher weight errors. Here we conjecture that 
\begin{equation}
    \E_{C\sim\mathrm{RQC}(n,l)}\left|\bra{0^n}C^\dag \sigma_{k\text{-local}} C\ket{0^n}\right|^2\leq e^{-O(k)} e^{-\Delta l}+\frac{1}{2^n},
\end{equation}
which follows intuitively because of the additional factor of $\sim (4/15)^t$ due to the second boundary condition.

\subsection{Alternative proof idea}
Here we also present an alternative proof idea for analyzing $\E_{C\sim\mathrm{RQC}(n,l)}\left|\bra{0^n}C^\dag \sigma_p C\ket{0^n}\right|^2$ using the result of~\cite{Brandao2016}. This proof idea works for arbitrary Pauli observables but has an additional $2^n$ factor due to converting between different norms, and therefore does not satisfy our purpose for proving Eq.~\eqref{eq:appPauliobservable}. It is interesting to see whether this idea can be improved to directly prove Eq.~\eqref{eq:appPauliobservable} with arbitrary locality.

First, recall the following result from~\cite{Brandao2016},
\begin{equation}\label{eq:bhh16}
    \left\|\E_{C\sim\mathrm{RQC}(n,l)}C^{\otimes t}\otimes(C^*)^{\otimes t}-\E_{C\sim\mathbb{U}(2^n)}C^{\otimes t}\otimes(C^*)^{\otimes t}\right\|_{2\to 2}\leq e^{-\Delta l},
\end{equation}
where $\mathrm{RQC}(n,l)$ denotes depth-$l$ random quantum circuits on $n$ qubits, $\|\cdot\|_{2\to 2}$ denotes the super-operator 2 norm, and $\Delta$ is a constant that depends on $t$. When $t=2$, $\E [C^{\otimes 2}\otimes(C^*)^{\otimes 2}]$ can be understood as a quantum channel that acts on $2n$ qubits. As we have a fixed input $\ketbra{0^n}^{\otimes 2}$, Eq.~\eqref{eq:bhh16} implies that the output states are close in 2 norm, which combined with Lemma~\ref{lemma:2normobservable} gives
\begin{equation}
    \E_{C\sim\mathrm{RQC}(n,l)}\left|\bra{0^n}C^\dag \sigma_p C\ket{0^n}\right|^2\leq 2^n e^{-\Delta l}+\frac{1}{2^n+1},
\end{equation}
which has an additional dimension factor $2^n$ due to the conversion between 1 and 2 norm. However, our numerical simulation results in Fig.~\ref{fig:pauliobservabledecay} suggest that this additional $2^n$ factor can be replaced by $O(1)$, for single qubit as well as high weight Pauli errors.

\begin{lemma}\label{lemma:2normobservable}
For a Hermitian observable $W$ and $n$-qubit quantum states $\rho,\sigma$, we have 
\begin{equation}
    \Tr[W(\rho-\sigma)]\leq 2^{n/2}\|W\|_\infty\cdot \|\rho-\sigma\|_2,
\end{equation}
where $\|\cdot\|_p$ denotes the Schatten $p$-norm.
\end{lemma}
\begin{proof}
Let $\lambda_1\geq\cdots\geq\lambda_{2^n}$ denote the absolute eigenvalues of $\rho-\sigma$. Then
\begin{nalign}
    \Tr[W(\rho-\sigma)]&\leq \|W\|_\infty(\lambda_1+\cdots+\lambda_{2^n})\\
    &\leq 2^{\frac{n}{2}} \|W\|_\infty(\lambda_1^2+\cdots+\lambda_{2^n}^2)^{\frac{1}{2}}\\
    &=2^{\frac{n}{2}} \|W\|_\infty\cdot \|\rho-\sigma\|_2.
\end{nalign}
\end{proof}

\section{Additional numerical simulation results}
\label{app:numerical}

\subsection{Computing the effective noise rate}
\begin{table*}[t]
  \begin{tabular}{ |l|c|c|c|c| }
    \hline
    & Name & Lindblad Superoperator & Description & ENR\\
    \hline
     \multirow{3}{*}{\makecell[c]{Single qubit\\ noise}}  &  $T_1$ & $\gamma D[\sigma]$ & Amplitude Decay & $\gamma/2$\\
       & $T_\phi$ & $\gamma D[\sigma^\dagger\sigma]$ & Dephasing & $\gamma/4$ \\
       & Pauli-$X$ & $\gamma D[X]$ & Pauli-$X$ & $\gamma$ \\\hline
      \multirow{3}{*}{\makecell[c]{Correlated\\ noise}}  & Corr-$T_1$ & $\gamma D[\sigma_i \sigma_j]$ &  Amplitude Decay & $\gamma/4$ \\
        & Corr-$T_\phi$ & $\gamma D[\sigma^\dagger_i\sigma_i \sigma^\dagger_j\sigma_j]$ & Dephasing & $3\gamma/16$ \\
        & Corr-Pauli & $\gamma D[\sigma_p]$ &  Arbitrary Pauli noise &$\gamma$\\
    \hline
  \end{tabular}
  \caption{Effective noise rate for various noise models, including those studied in this paper. Here $\sigma=\ketbra{0}{1}$, $\sigma_p\in\{I,X,Y,Z\}^{\otimes n}$. $\gamma$ is the noise strength used in simulating the Lindblad evolution, and ENR stands for effective noise rate.}
  \label{tab:noisechannelENR}
\end{table*}

In this section we show how to compute the effective noise rate given the description of a noise model, where the results are given in Table~\ref{tab:noisechannelENR}. These results are used to control the total effective noise rate in our numerical simulations. For example, in the first noise model of Table~\ref{tab:noise_table_app}, the total effective noise rate is $\gamma/2+2\times \gamma/4=\gamma$.

Next we show how to calculate the numbers in Table~\ref{tab:noisechannelENR} with an example. Consider the $T_\phi$ noise model with Lindblad operator $\gamma D[\ketbra{1}]$. The evolution of the density matrix is given by
\begin{equation}
    \frac{\mathrm{d} \rho}{\mathrm{d} t}
  =\gamma \left(\ketbra{1}\rho\ketbra{1}-\frac{1}{2}\ketbra{1}\rho-\frac{1}{2}\rho\ketbra{1}\right).
\end{equation}
In our simulation the evolution continues for 1 time unit. Using a first order approximation, we can write the noise channel as
\begin{equation}
    \mc N(\rho)=\rho+\gamma \left(\ketbra{1}\rho\ketbra{1}-\frac{1}{2}\ketbra{1}\rho-\frac{1}{2}\rho\ketbra{1}\right).
\end{equation}
To represent the noise channel in Pauli basis, we use $\ketbra{1}=\frac{I-Z}{2}$ and get
\begin{equation}
\begin{split}
    \mc N(\rho)&=\rho+\frac{\gamma}{4}(Z\rho Z-\rho)\\
    &=\left(1-\frac{\gamma}{4}\right)\rho+\frac{\gamma}{4}Z\rho Z.
\end{split}
\end{equation}
Therefore the effective noise rate of $T_\phi$ is $\gamma/4$, and the other numbers can be calculated similarly.

\subsection{Simulation algorithm}
For a generic noise source, the incoherent dynamics of the density matrix between gates can be described by the master equation,
\begin{equation}
  \frac{\mathrm{d} \rho}{\mathrm{d} t}
  = \sum_l \gamma_l D[J_l](\rho),
\end{equation}
where $\rho$ is the density matrix, $D[J_i](\rho) = J_i \rho J_i^\dagger  - \frac{1}{2}(J_i^\dagger J_i \rho + \rho J_i^\dagger J_i)$ is a Lindblad superoperator for generic collapse operator $J_i$, 
and we use units where $\hbar=1$. A Lindblad superoperator can represent various Markovian noise sources, such as amplitude decay, dephasing, correlated noise, etc. Lindblad operators for the various noise
sources studied in this work are given in Table~\ref{tab:noisechannelENR}. We take the system Hamiltonian, $H$, to be zero and assume that gates are applied perfectly and instantaneously one time unit apart.

For 10 qubit simulations, we evolve the density matrix according to the Lindblad master equation.
For a large number of qubits (such as the 20 qubit runs of the main text), storing the full density matrix $\rho$ becomes unfeasible. Instead, we simulate the action of the 
noise sources using the Monte Carlo wave function (MCWF) method~\cite{plenio1998quantum}. We evolve under an effective, non-Hermitian Hamiltonian
\begin{equation}
H_{\mathrm{eff}} = \frac{-i}{2} \sum_l \gamma_l J_l^\dagger J_l.
\end{equation}
As the state evolves, it will gradually lose norm. A random number $p$ is drawn, 
and the state is evolved under $H_{\mathrm{eff}}$ until the norm falls below $p$. 
At this point a ``quantum jump" occurs. At this point, one of the noise channels is
randomly chosen. The probability
of the jump happening due to noise channel $l$ is given by
\begin{equation}
P_l = \frac{\gamma_l \langle \psi | J_l^\dagger J_l | \psi \rangle}{\sum_l \gamma_l \langle \psi | J_l^\dagger J_l | \psi \rangle}.
\end{equation}
Once a specific noise channel is chosen, the jump is applied according to
\begin{equation}
    |\psi \rangle = \frac{J_l |\psi\rangle}{\langle \psi | J_l^\dagger J_l | \psi \rangle}.
\end{equation}
A new random number, $p$, is drawn and the evolution then proceeds under the effective non-Hermitian Hamiltonian $H_{\mathrm{eff}}$
until the next jump occurs. This process is applied between gates until all gates have been applied. This 
process is repeated many times, where each time a pure state trajectory of the noisy circuit is generated. Due to linearity, the fidelity and fidelity estimators of a noisy circuit can be computed by averaging over the pure state trajectories.

\subsection{Additional simulation results}
\begin{figure}[t]
   \centering
   \subfloat[]{
    \centering
    \includegraphics[width=0.4\linewidth]{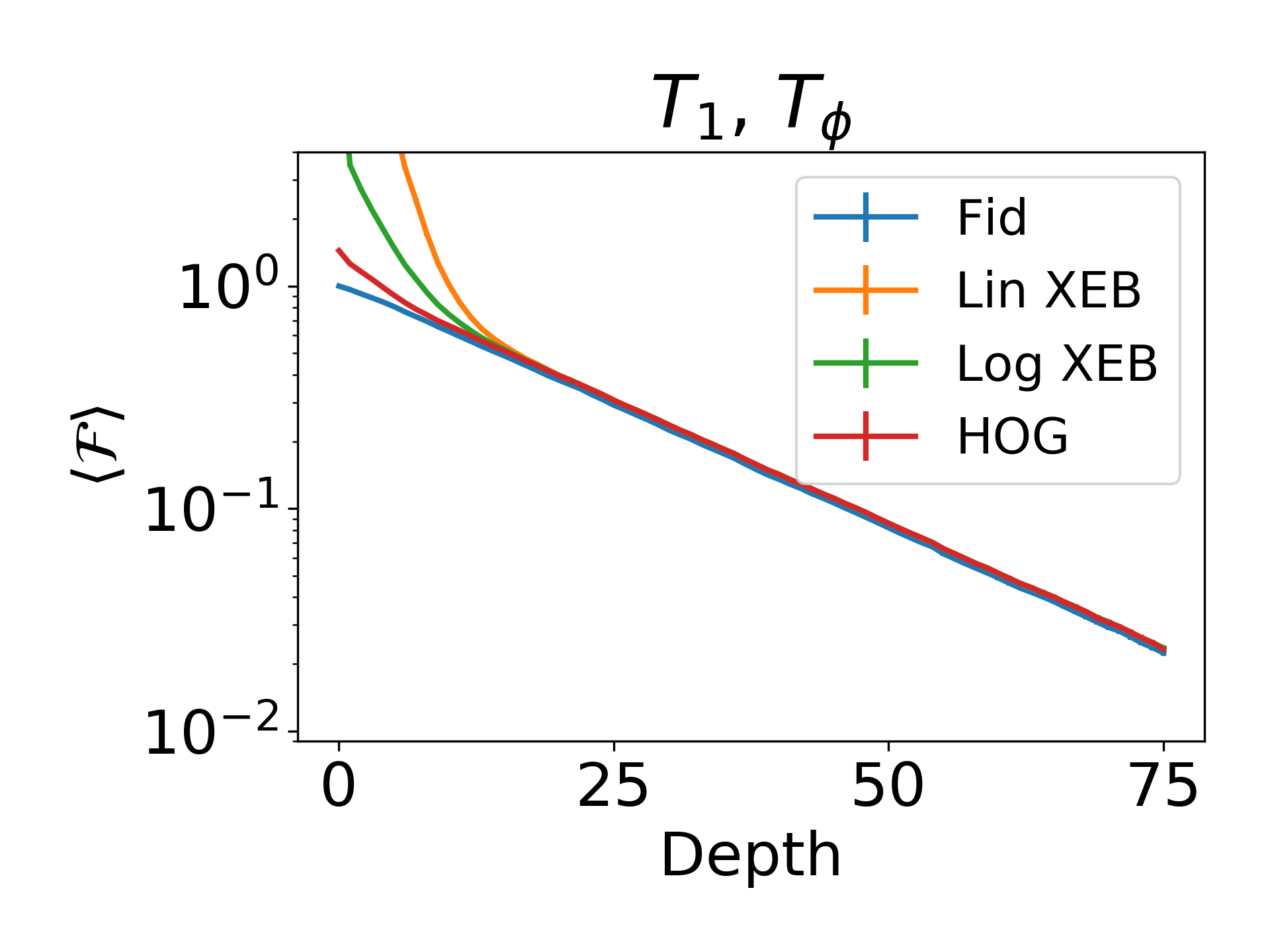}
    }
       \subfloat[]{
    \centering
    \includegraphics[width=0.4\linewidth]{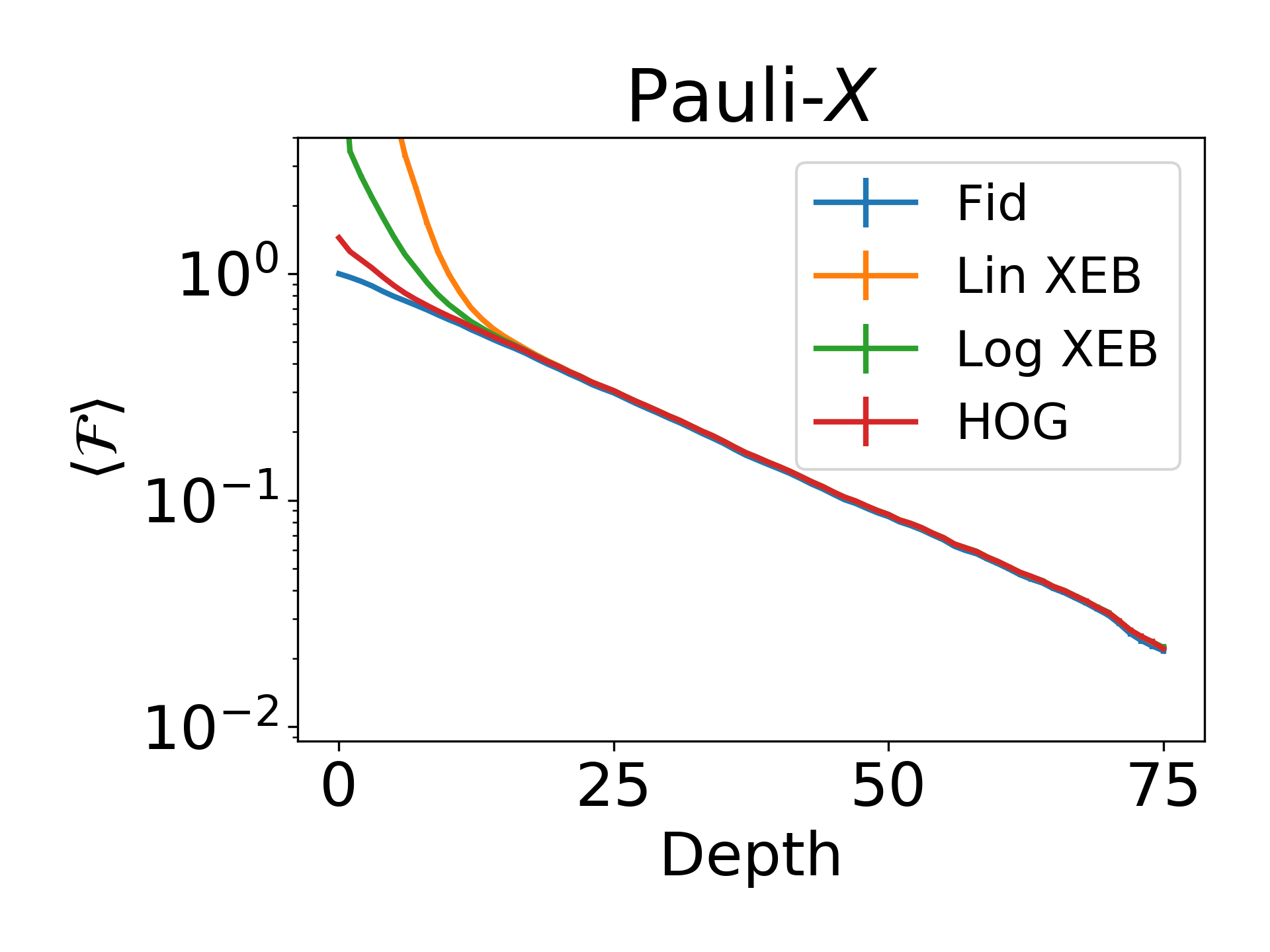}
    }
    \\
       \subfloat[]{
    \centering
    \includegraphics[width=0.4\linewidth]{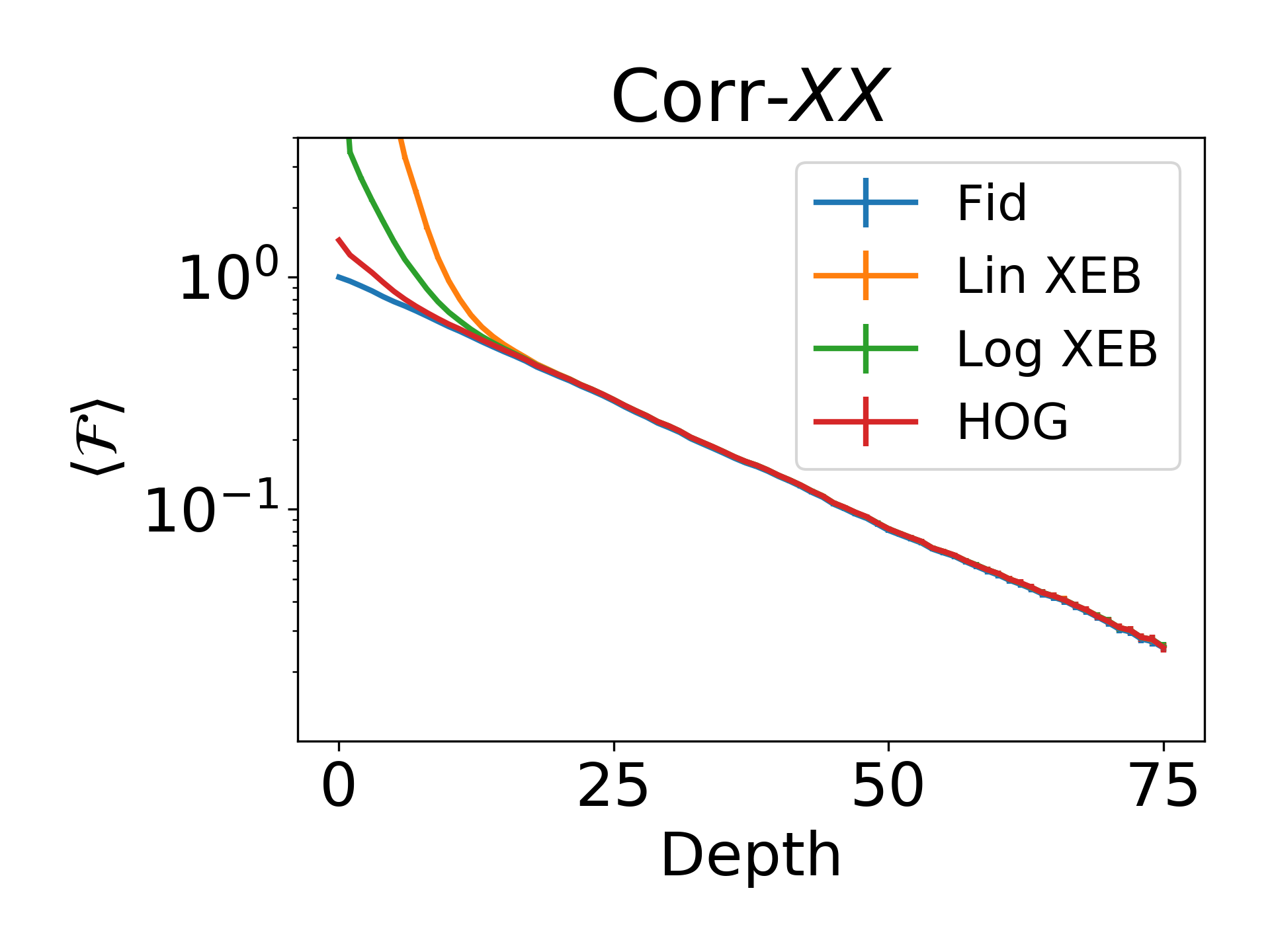}
    }
    \subfloat[]{
    \centering
    \includegraphics[width=0.4\linewidth]{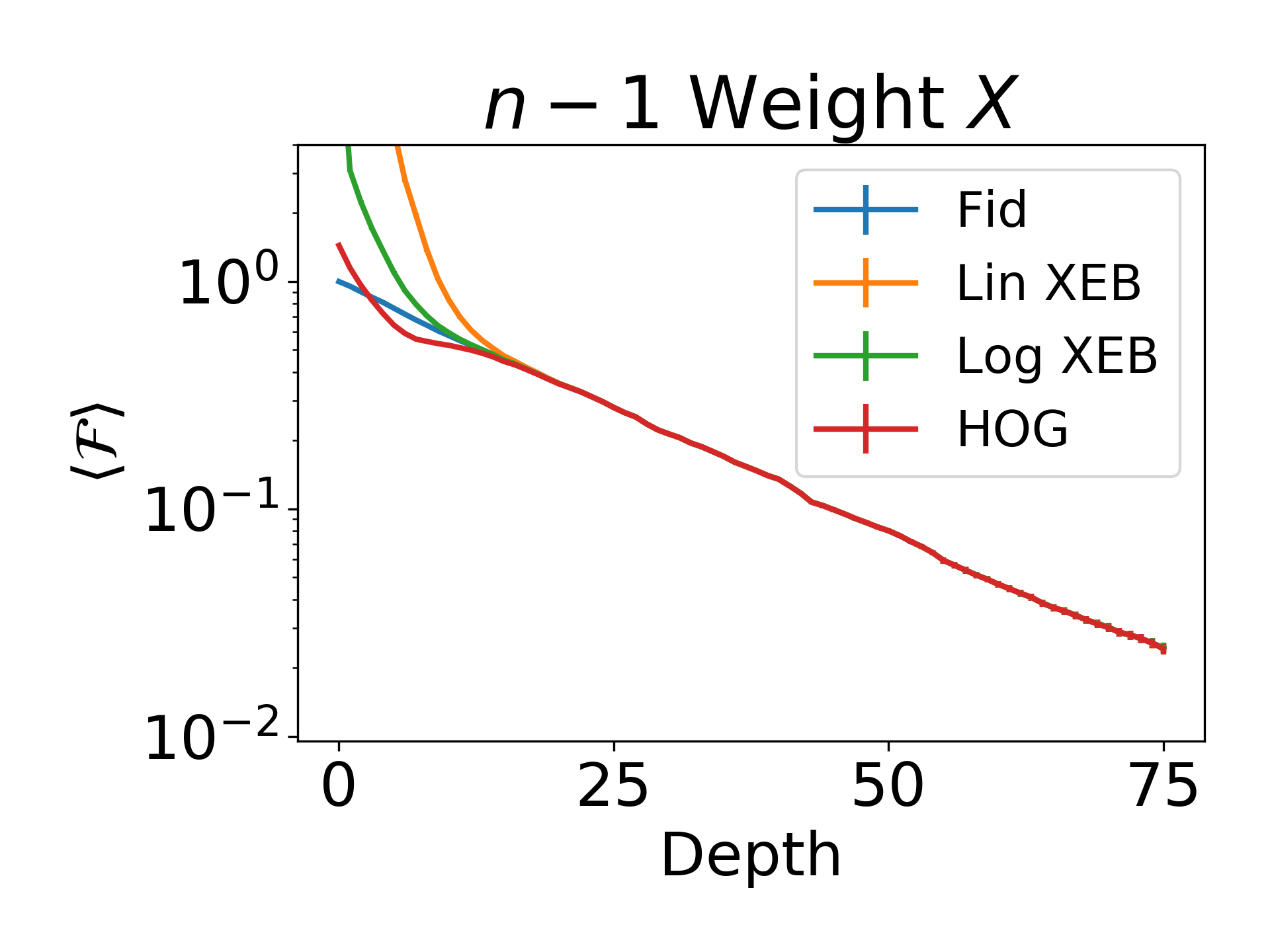}
    }

\caption{Numerical simulations of XEB, log XEB, and HOG fidelity using the Monte Carlo wave function (MCWF) method for various noise models. As opposed to unbiased XEB, these estimators only converge with sufficient depth. (a) Single qubit amplitude-decay and pure dephasing. (b) Nearest-neighbor correlated $XX$ noise. (c) Correlated $X$ noise with weight $n_q-1$ = 19.}
  \label{fig:other_mcwf_simulations}
\end{figure}

\begin{table}
  \begin{tabular}{ |l|c|l| }
    \hline
    Name & Formula \\
    \hline
    uXEB & $\hat{F}_{\mathrm{uXEB}} =  \E \frac{\sum_x D p(x) q(x) - 1}{D \sum_x p(x)^2 -1 }$\\
    XEB & $\hat{F}_{\mathrm{XEB}} = \E \sum_x D p(x) q(x) - 1$\\
    Log XEB & $\hat{F}_{\mathrm{log}} = \log D + \gamma + \E \sum_x q(x) \log(p(x))$\\
    HOG Fidelity & $\hat{F}_{\mathrm{HOG}} = \E(2 \sum_x q(x) 1[ p(x)\geq \frac{\log 2}{D}] - 1)/\log 2$\\
    \hline
  \end{tabular}
  \caption{Fidelity estimators studied in this paper. Since our numerical simulations have access to the full wavefunction, 
  these formulas are slightly different than those used when using experimental samples. $x\in\{0,1\}^n$ represents all possible bitstrings, 
  $D=2^n$ is the dimension of the Hilbert space, $p(x)$ is the ideal output probability,
    $q(x)$ is the output probability of the noisy circuit, $\gamma$ is Euler's
    constant, and $1[\cdot]$ is the indicator function. The expectation value is taken over both random circuits $C$ and, in the case of the MCWF solver, noise trajectories.}\label{tab:app_metrics_table}  
\end{table}

\begin{table*}[t]

  \begin{tabular}{ |c|c|c|c|c|c|c| }
    \hline
    Description & Lindblad &  $\lambda_{F}$ & $\lambda_{\mathrm{uXEB}}$ & $\lambda_{\mathrm{XEB}}$ & $\lambda_{\mathrm{LOG}}$ & $\lambda_{\mathrm{HOG}}$ \\
    \hline
    $T_1$, $T_\phi$ & $\gamma D[\sigma] + 2\gamma D[\sigma^\dagger \sigma]$ & 0.0511(2) & 0.0511(2) & 0.0511(2) & 0.0511(2) & 0.0510(2) \\\hline
Pauli-$X$ & $\gamma D[X]$ & 0.0508(2) & 0.0509(2) & 0.0509(2) & 0.0509(2) & 0.0508(2) \\\hline
Corr-$XX$ & $\gamma D[X_i X_{i+1}]$ & 0.0505(3) & 0.0505(3) & 0.0505(3) & 0.0505(3) & 0.0505(3) \\\hline
$n-1$ Weight $X$ & $\gamma D[\prod_{i\ne j} X_i]$ & 0.0506(3) & 0.0506(3) & 0.0506(3) & 0.0506(3) & 0.0506(3) \\\hline
  \end{tabular}
  \caption{Numerical simulation of RCS benchmarking using the Monte Carlo wave function (MCWF) technique. Here we simulate $n=20$ qubits with noise strength $\gamma=0.0025$, and $\sigma=\ketbra{0}{1}$. Each global noise model has a total ENR of $\lambda_{\mathrm{true}} = n\gamma = 0.05$ by design. $\lambda_F$ and $\lambda_{\mathrm{uXEB}}$ shows the simulated RCS benchmarking result, which corresponds to the decay rate of fidelity and unbiased linear cross entropy, respectively. $\lambda_{\mathrm{XEB}}$, $\lambda_{\mathrm{LOG}}$, and $\lambda_{\mathrm{HOG}}$ correspond to the alternative fidelity estimators of standard linear cross entropy, cross entropy, and a HOG-score based estimator.}\label{tab:noise_table_app}
\end{table*}

\begin{table*}[t]
  \begin{tabular}{ |c|c|c|c|c|c|c|c|c|c|c| }
    \hline
    \multirow{2}{*}{Noise rate}&\multirow{2}{*}{Description} & \multicolumn{2}{c|}{$n=10$, single qubit} & \multicolumn{2}{c|}{$n=10$, two qubit} & \multicolumn{2}{c|}{$n=20$, single qubit} & \multicolumn{2}{c|}{$n=20$, two qubit} \\
    \cline{3-10}
     &  &  $\lambda_{F}$ & $\lambda_{\mathrm{uXEB}}$&  $\lambda_{F}$ & $\lambda_{\mathrm{uXEB}}$&  $\lambda_{F}$ & $\lambda_{\mathrm{uXEB}}$&  $\lambda_{F}$ & $\lambda_{\mathrm{uXEB}}$ \\
    \hline
    \multirow{4}{*}{$\lambda=0.05$}& $T_1$, $T_\phi$ & 0.0491(2) & 0.0511(7) & 0.04978(4) & 0.04990(2) & 0.0495(4) & 0.0523(1) & 0.0511(2) & 0.0511(2) \\\cline{2-10}
 & Pauli-$X$ & 0.0488(6) & 0.0515(9) & 0.04970(4) & 0.04983(2) & 0.0482(4) & 0.0527(1) & 0.0508(2) & 0.0509(2) \\\cline{2-10}
 & Corr-$XX$ & 0.0494(2) & 0.0517(9) & 0.04974(2) & 0.04987(1) & 0.0494(4) & 0.0527(2) & 0.0505(3) & 0.0505(3) \\\cline{2-10}
 & $n-1$ Weight $X$ & 0.0500(2) & 0.0509(1) & 0.049761(1) & 0.049879(6) & 0.0499(3) & 0.0505(8) & 0.0506(3) & 0.0506(3) \\\hline
\multirow{4}{*}{$\lambda=0.1$}& $T_1$, $T_\phi$ & 0.0968(9) & 0.1029(2) & 0.09902(9) & 0.09964(1) & 0.0971(8) & 0.1127(3) & 0.0998(6) & 0.1000(6) \\\cline{2-10}
 & Pauli-$X$ & 0.0968(6) & 0.1049(4) & 0.09871(8) & 0.09929(4) & 0.0964(6) & 0.1112(2) & 0.1020(8) & 0.1020(7) \\\cline{2-10}
 & Corr-$XX$ & 0.0973(5) & 0.1042(2) & 0.09888(5) & 0.09946(2) & 0.0980(7) & 0.1163(3) & 0.0996(8) & 0.1000(8) \\\cline{2-10}
 & $n-1$ Weight $X$ & 0.100(1) & 0.099(1) & 0.098954(2) & 0.099512(1) & 0.0998(7) & 0.1013(3) & 0.0948(8) & 0.0948(8) \\\hline
\multirow{4}{*}{$\lambda=0.15$}& $T_1$, $T_\phi$ & 0.1436(4) & 0.1557(5) & 0.1471(1) & 0.1498(5) & 0.1435(8) & 0.1811(5) & 0.150(1) & 0.150(1) \\\cline{2-10}
 & Pauli-$X$ & 0.146(3) & 0.150(9) & 0.1464(1) & 0.1484(6) & 0.143(1) & 0.182(6) & 0.144(2) & 0.145(2) \\\cline{2-10}
 & Corr-$XX$ & 0.145(1) & 0.152(1) & 0.14677(7) & 0.14878(4) & 0.142(1) & 0.175(6) & 0.164(2) & 0.164(2) \\\cline{2-10}
 & $n-1$ Weight $X$ & 0.150(1) & 0.144(4) & 0.146874(4) & 0.148911(2) & 0.146(1) & 0.159(9) & 0.155(2) & 0.155(2) \\\hline
  \end{tabular}
  \caption{Additional numerical simulation results of RCS benchmarking. We simulate three effective noise rates with different noise models, two system sizes (10 and 20 qubits), and two gate sets (two-qubit Haar random gates, $\mathrm{CNOT}$+single qubit Haar random gates).}\label{tab:additionalsimulation}
\end{table*}

In addition to the unbiased linear cross entropy estimator, we also run three other fidelity estimators: estimators based on the standard linear cross entropy~\cite{arute2019quantum}, cross entropy~\cite{Boixo2018Characterizing},
and heavy-output-generation (HOG) score~\cite{aaronson2017complexity,arute2019quantum}.
All estimators are summarized in Table~\ref{tab:app_metrics_table}, and they take a slightly different form than the equations used in the main text 
due to the simulations having access to the full wavefunction for a single noise trajectory.
We use these estimators to study the $n=20$ linear ring and plot the results in Fig.~\ref{fig:other_mcwf_simulations} and summarize the results in Table~\ref{tab:noise_table_app}.
These simulations, as well as the $n=20$ results in the main text, were sampled over 100 circuits and 400 noise trajectories for each circuit at each given depth. All
estimators are fit from depths 20 to 50.

We also compare the results from both true fidelity and the unbiased linear cross entropy for $n=10$ and $n=20$ qubits with both two qubit Haar random unitaries and CNOT+single qubit Haar random unitaries at three different effective noise rates for our four different noise models. The results are summarized in Table~\ref{tab:additionalsimulation}. We find that the results are consistent over all variations.
We fit the $n=20$, two qubit Haar random results from depths 20 to 50. We fit the $n=10$, two qubit Haar random results from depths 10 to 25. 
We fit the $n=20$ and $n=10$ single qubit Haar random results from depths 20 to 34. 

\subsection{Variance analysis}
\begin{figure}[t]
\centering
\subfloat[two-qubit Haar]{\includegraphics[width=0.45\linewidth]{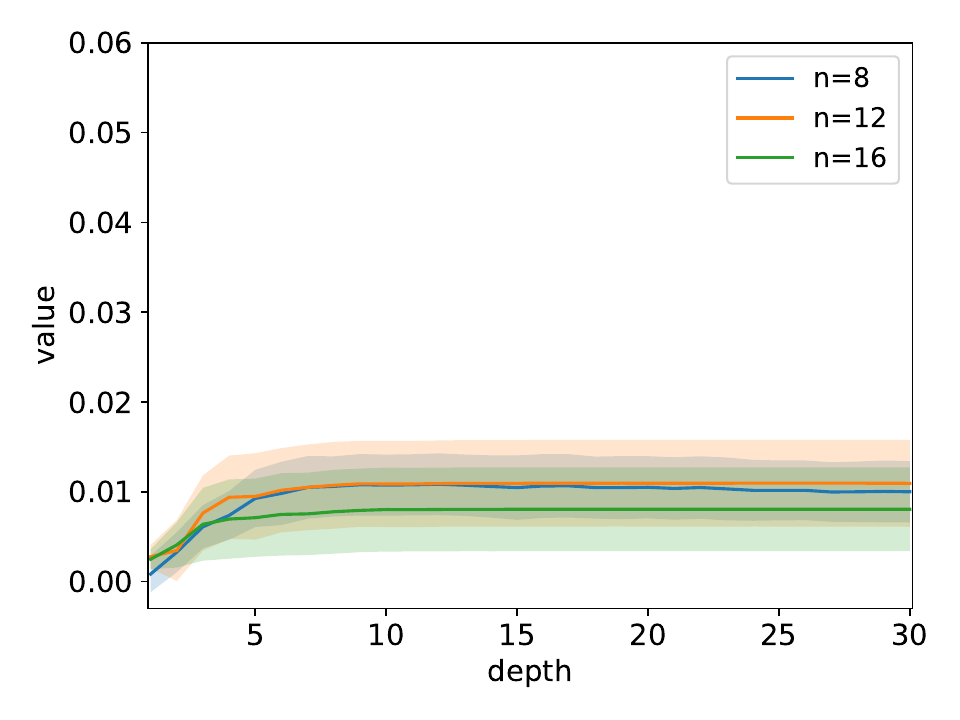}}
 \subfloat[$\mathrm{CNOT}$+single qubit Haar]{\includegraphics[width=0.45\linewidth]{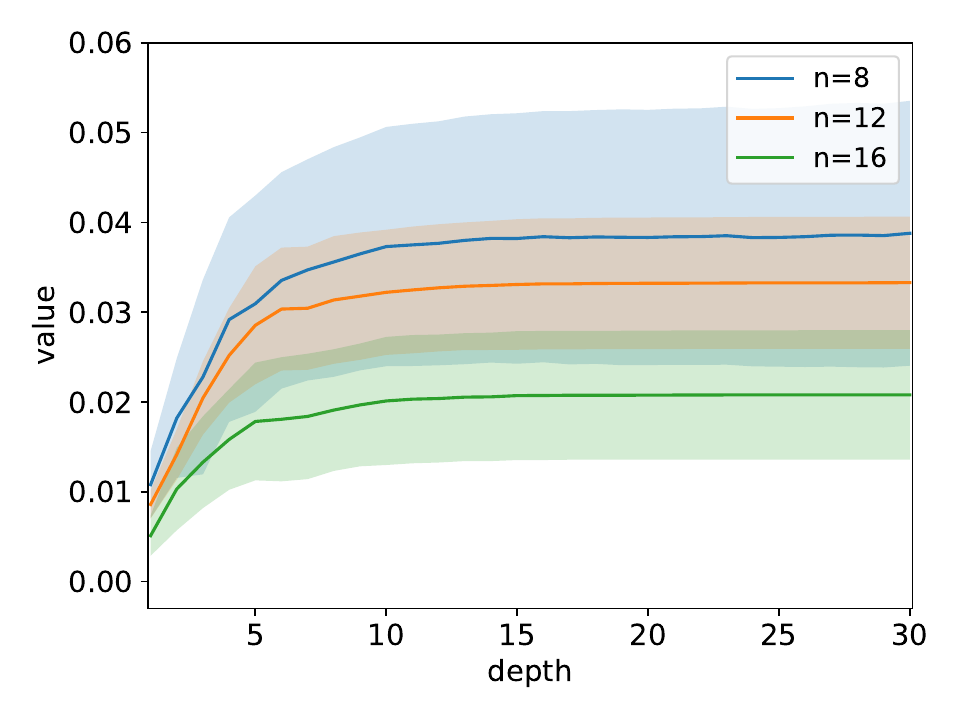}}
\caption{Numerical simulation of $\Var\left(\sum_{l=1}^{d}A_l\right)=\sum_{k,l=1}^d\left(\E[A_k A_l]-\E[A_k]\E[A_l]\right)$ which verifies our variance model $\Var(F)=O\left(\lambda^2 \left(\E F\right)^2\right)$. (a) Two-qubit Haar random gates. (b) $\mathrm{CNOT}$+single qubit Haar random gates. $\Var\left(\sum_{l=1}^{d}A_l\right)$ converges to a constant with both gate sets, while the constant for single qubit Haar random gates is larger.}
\label{fig:rcs_variance_simulation}
\end{figure}

Next we present numerical simulation results which support our analysis of the variance of fidelity in section~\ref{sec:rcsvariance}. Recall the setting in section~\ref{sec:rcsvariance} where we consider i.i.d. single qubit Pauli-$X$ noise and use a first-order approximation of the fidelity. For a random circuit $C\sim\mathrm{RQC}(n,d)$, we can write the fidelity as
\begin{equation}
\begin{split}
    F&\approx(1-\varepsilon)^{nd}+\sum_{i=1}^{n}\sum_{l=1}^{d}\varepsilon (1-\varepsilon)^{nd-1}\left|\braket{\psi_{i,l}}{\psi}\right|^2,
\end{split}
\end{equation}
where $\ket{\psi_{i,l}}$ denotes the ideal state with an $X$ error on qubit $i$ at depth $l$. Let
\begin{equation}
    A_l:=\frac{1}{n}\sum_{i=1}^n \left|\braket{\psi_{i,l}}{\psi}\right|^2,
\end{equation}
then
\begin{equation}
    F\approx (1-\varepsilon)^{nd}+n\varepsilon (1-\varepsilon)^{nd-1}\sum_{l=1}^{d}A_l.
\end{equation}
Therefore,
\begin{equation}
    \Var(F)\approx \left(n\varepsilon\right)^2 (1-\varepsilon)^{2nd-2}\Var\left(\sum_{l=1}^{d}A_l\right)\approx \lambda^2 \left(\E F\right)^2 \Var\left(\sum_{l=1}^{d}A_l\right),
\end{equation}
where we have used the fact that $\lambda\approx n\varepsilon$ and $\E F\approx (1-\varepsilon)^{nd}$. Next, we show with numerical simulation that
\begin{equation}\label{eq:appvariance}
    \Var\left(\sum_{l=1}^{d}A_l\right)=\sum_{k,l=1}^d\left(\E[A_k A_l]-\E[A_k]\E[A_l]\right)=O(1).
\end{equation}

In the following we present simulation results to verify Eq.~\eqref{eq:appvariance} with two gate sets: 2-qubit Haar random gates (as in Fig.~\ref{fig:rcsbenchmarkingapp}a) and $\mathrm{CNOT}$+single qubit Haar random gates (as in Fig.~\ref{fig:rcsbenchmarkingapp}b). For each $k,l=1,\dots,30$, we simulate $\E[A_k A_l]$, $\E[A_k]$ and $\E[A_l]$ by averaging over 100 random circuits. For each random circuit, we randomly sample 100 error locations. The solid lines and color regions in Fig.~\ref{fig:rcs_variance_simulation} denote mean and standard error across 5 independent experiments.

In Fig.~\ref{fig:rcs_variance_simulation} we present simulation results for $n=8,12,16$ qubits. We observe that $\Var\left(\sum_{l=1}^{d}A_l\right)$ converges to a constant as depth increases for both gate sets. Interestingly, although the simulation results are very noisy with large error bars, we can observe a clear difference between the two gate sets. That is, $\Var\left(\sum_{l=1}^{d}A_l\right)$ converges to a larger constant with single qubit Haar random gates. This verifies our intuition that a gate set with more randomness has smaller variance in RCS benchmarking.

\section{Additional experiment details}
\label{app:experiment}
\begin{figure}[t]
\centering
\sbox{\measurebox}{
  \begin{minipage}[b]{.62\linewidth}
  \subfloat
    []
    {\includegraphics[width=\linewidth]{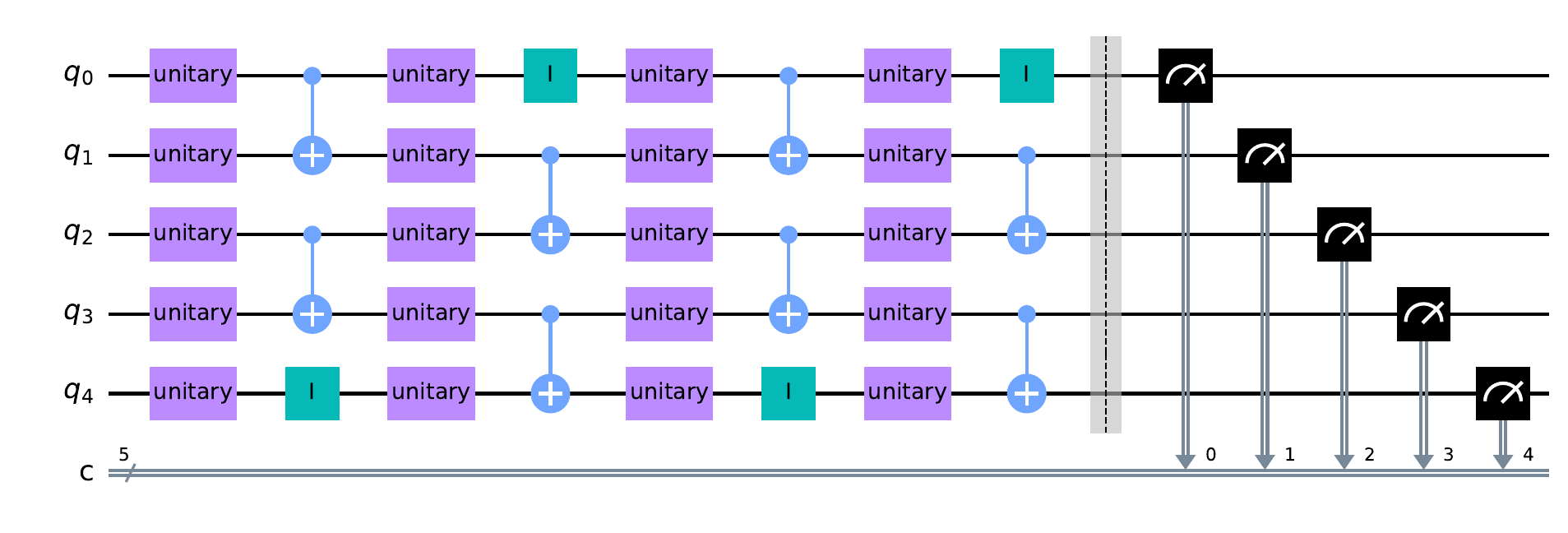}}
  \end{minipage}}
\usebox{\measurebox}\enskip
\begin{minipage}[b][\ht\measurebox+0.5cm][s]{.35\linewidth}
\centering
\subfloat[5 qubit device]
  {\includegraphics[width=0.6\linewidth]{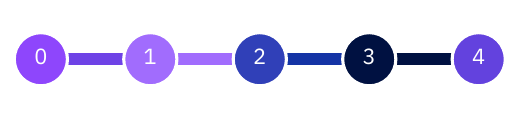}}
\vfill
\subfloat[27 qubit device]
  {\includegraphics[width=0.8\linewidth]{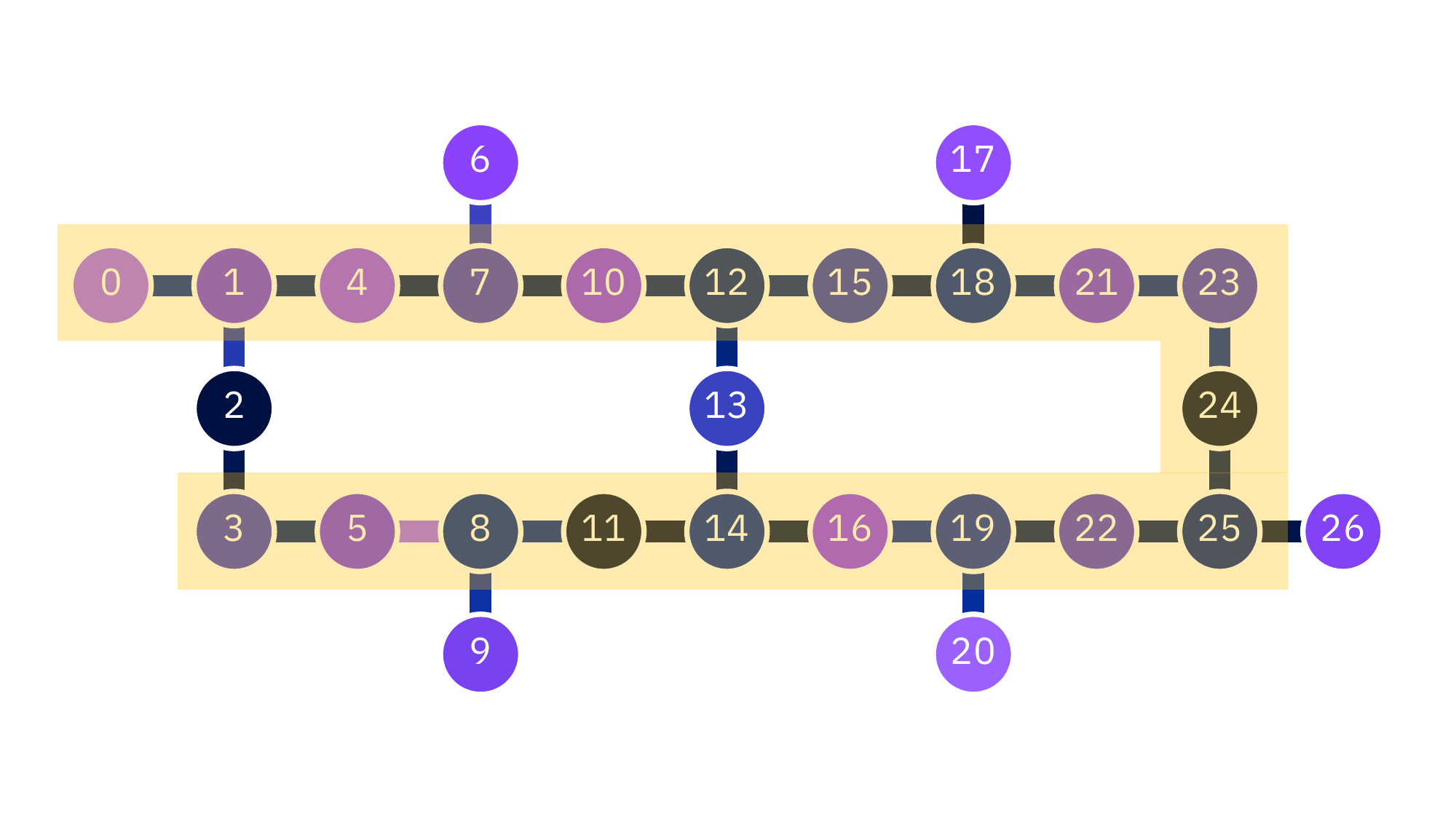}}
\end{minipage}
\caption{Architecture of the IBM Quantum superconducting qubit devices used in our experiments. (a) An example circuit on 5 qubits with depth 4. Each purple box represents a Haar random single qubit unitary gate. Figure generated by Qiskit~\cite{Qiskit}. (b)(c) Architecture for the 5 and 27 qubit devices. Figures retrieved from \url{https://quantum-computing.ibm.com/services?services=systems}. The subset of qubits used in the 27 qubit device is highlighted. The experiment in Fig.~\ref{fig:rcsvariance20qubit} is performed by adding qubits in the order given by $\texttt{[0,1,4,7,10,12,15,18,21,23,24,25,22,19,16,14,11,8,5,3]}$.}
\label{fig:ibmq_devices}
\end{figure}

Fig.~\ref{fig:ibmq_devices} shows the architecture of the devices used in our experiments. For the 27 qubit device, we consider a 20 qubit subset indexed by $\texttt{[0,1,4,7,10,12,15,18,21,23,24,25,22,19,16,14,11,8,5,3]}$ for the 20 qubit experiments, and use the first 10 qubits in this list for the 10 qubit experiments.

We give an example of the circuits implemented in RCS benchmarking in Fig.~\ref{fig:ibmq_devices}a. We implement Haar random single qubit gates between $\mathrm{CNOT}$ layers. Note that an arbitrary single qubit gate can be decomposed by
\begin{equation}\label{eq:singlequbitgate}
    U(\theta,\phi,\lambda)=R_z(\phi-\pi/2)R_x(\pi/2)R_z(\pi-\theta)R_x(\pi/2)R_z(\lambda-\pi/2).
\end{equation}
(See \url{https://qiskit.org/documentation/stubs/qiskit.circuit.library.UGate.html}.) Here $R_x(\pi/2)$ is the $\sqrt{X}$ gate, and also called a X90 pulse. The $R_z$ gates, which are rotations around $z$-axis, are implemented \emph{virtually} in hardware via framechanges and do not have any error. Therefore the error rate of a Haar random single qubit gate equals twice the error rate of a X90 pulse. $R_z(\theta)$ and $\sqrt{X}$ are the native single qubit gates supported by the devices. Therefore, when submitting circuits to the devices through cloud, each single qubit gate is first decomposed to the form in Eq.~\eqref{eq:singlequbitgate}.

By default, for any circuit submitted to the cloud, we take the maximum amount of measurement samples allowed (8192). We can submit the same circuit multiple times if more samples are needed. Below we give the parameters used in the experiments presented in Fig.~\ref{fig:rcs10and20qubit}. The parameters for the other experiments are already given in the main text. Here ``repeat" refers to how many times each circuit is submitted. The total circuit count includes the repeated ones, so the total amount of samples collected equals this number times 8192.
\begin{itemize}
    \item Fig.~\ref{fig:rcs10and20qubit}a: \texttt{depth=[1, 5, 9, 13, 17, 21, 25, 29, 33, 37]}, 90 circuits for each depth, \texttt{repeat=[1, 1, 1, 1, 1, 1, 1, 1, 1, 1]}, 900 circuits in total.
    \item Fig.~\ref{fig:rcs10and20qubit}b: \texttt{depth=[20, 22, 24, 26, 28, 30, 32]}, 100 circuits for each depth, \texttt{repeat=[1, 1, 1, 1, 2, 2, 2]}, 1000 circuits in total.
    \item Fig.~\ref{fig:rcs10and20qubit}c: \texttt{depth=[1, 5, 9, 13, 17, 21, 25, 29, 33, 37]}, 100 circuits for each depth, \texttt{repeat=[1, 1, 1, 1, 1, 1, 1, 1, 1, 1]}, 1000 circuits in total.
    \item Fig.~\ref{fig:rcs10and20qubit}d: \texttt{depth=[20, 22, 24, 26, 28, 30, 32]}, 100 circuits for each depth, \texttt{repeat=[1, 1, 1, 1, 2, 2, 2]}, 1000 circuits in total.
\end{itemize}

\end{document}